\documentclass[12pt]{article}
\usepackage[english]{babel}
\usepackage{csquotes}
\usepackage[nodisplayskipstretch]{setspace}
\usepackage{booktabs}
\usepackage{multirow}
\usepackage{multicol}
\usepackage{tabu}
\usepackage[T1]{fontenc}
\usepackage{amssymb}
\usepackage{enumerate}
\usepackage{tikz}
\usetikzlibrary{arrows}
\usepackage{bbm}
\usepackage{subcaption}
\usepackage{tabularx}
\usepackage{siunitx}
\usepackage{adjustbox} 
\usepackage[labelfont=bf,format=plain,justification=raggedright,singlelinecheck=false]{caption}
\usepackage{comment}
\usepackage{natbib}
\usepackage{pdflscape}
\usepackage{afterpage}
\usepackage{makecell}

\usepackage[a4paper,top=2.54cm,bottom=2.54cm,left=2.58cm,right=2.58cm,marginparwidth=1.75cm]{geometry}

\usepackage{amsmath}
\usepackage{amsthm}
\usepackage{thmtools, thm-restate}
\usepackage{graphicx}
\usepackage[colorlinks=true, allcolors=blue]{hyperref}
\usepackage{floatrow}
\usepackage{longtable}

\DeclareMathOperator{\Tr}{Tr}

\makeatletter
\renewenvironment{abstract}{%
    \if@twocolumn
      \section*{\abstractname}%
    \else 
      \begin{center}%
        {\bfseries \Large\abstractname\vspace{\z@}}
      \end{center}%
      \quotation
    \fi}
    {\if@twocolumn\else\endquotation\fi}
\makeatother

\theoremstyle{plain}

\newtheorem*{theoremnamed}{Theorem}

\newtheorem{lemma}{Lemma}
\newtheorem{prop}{Proposition}
\theoremstyle{definition}
\newtheorem{definition}{Definition}[section]
\theoremstyle{remark}

\newtheorem{assumption}{Assumption}[section]
\newtheorem{proposition}{Proposition}[section]
\newtheorem*{namedassumption}{Assumption}

\newenvironment{customproof}[1]{%
  \par\noindent\textit{Proof of #1:} \rmfamily\ignorespaces
}{\hfill$\square$\par}

\begin{document}
\title{Empirical Challenges with Peers-of-Peers Instruments in the Linear-In-Means Model}

\author{Nathan Canen\thanks{Corresponding Author. \\ Department of Economics, University of Warwick, and CEPR. Coventry, CV4 7AL, United Kingdom.\\E-mail: \url{Nathan.Canen@warwick.ac.uk}} \and Shantanu Chadha\thanks{Department of Economics, University of Warwick. Coventry, CV4 7AL, United Kingdom. \\ E-mail: \url{Shantanu.Chadha@warwick.ac.uk}.}}

\maketitle

\begin{abstract}
In the linear-in-means model, endogeneity arises naturally due to the reflection problem. A common solution is to use Instrumental Variables (IVs) based on higher-order network links, such as using friends-of-friends' characteristics. In this paper, we show that such instruments are unlikely to work well in many applied settings due to a specific sparse/dense-network mechanism: in extremely sparse networks, friends-of-friends instruments may become degenerate, while in denser 
networks they may still provide too little first-stage information. This implies that the IVs may be weak or that the first-stage estimand is undefined. We use random graph theory to characterize the rates at which these issues arise for a benchmark class of random graphs. 
This allows us to link network topology to first-stage information accumulation and to identify when such instruments are likely to perform well. We show how existing weak-IV robust inference can be adapted to this environment, and how scaling the network provides an alternative specification that can mitigate some of these challenges. 
We provide extensive Monte Carlo simulations and revisit empirical applications, showing the prevalence of such issues in empirical practice, and how our results apply. 

\end{abstract} 


\newpage
\doublespacing

\section{Introduction}\label{S1}

Humans are inherently social beings, frequently interacting in groups and affecting the behavior of their friends and neighbors. Thus, it comes as no surprise that the study of peer effects has become extremely popular in empirical research in Economics and Social Sciences more generally\footnote{Examples include peer effects in education, e.g., \cite{Sacerdote2001, calvo2009peer}, worker productivity and labor markets (e.g., \cite{MasMoretti2009, caria2024village}), Finance (e.g., \cite{LoughranSchultz2004} on the impact of IPOs on competitors); development and public goods (e.g., \cite{acemoglu2015}), among many others. See \cite{bramoulle2020peer} for a recent survey.}, especially with the emergence of high quality data on social interactions and advances in network statistics. 

The most popular model of peer effects is arguably the linear-in-means model illustrated in equation \eqref{struc_model_intro}. In this model, one's outcome ($Y_i$) depends linearly on the mean outcome across $i$'s group, denoted $\bar{Y}_i$. The outcome may also depend on the exogenous characteristics of one's self $X_i$, the group itself ($P(i)$), the average characteristics of the group, and $\varepsilon_{i}$ capturing the unobserved error\footnote{For example, in \cite{calvo2009peer}, the outcome (grades of student $i$) can depend on the average of its peers' grades ($\beta$), their own personal traits/parental education ($X_i$), the average group characteristics through $\delta$ and an unobserved error term ($\varepsilon_{i}$).}:
\begin{align}\label{struc_model_intro}
    Y_{i} = \beta \bar{Y}_i +  X_{i}'\boldsymbol{\gamma} + \bar{X}_i'\boldsymbol{\delta}   + \varepsilon_{i},
\end{align}
where $\bar{Y}_i = \frac{\sum_{j\in P_i}Y_{j}}{n_i}$ and $\bar{X}_i = \frac{\sum_{j\in P_i}X_{j}}{n_i}$, and $n_i$ is the number of $i$'s peers. Even if the error term, $\varepsilon_{i}$, is exogenous to $X_i$ and to the peer groups, endogeneity still arises in this model due to the simultaneous determination of behavior within the groups: the reflection problem (\cite{manski1993identification}). After all, an increase in $\varepsilon_i$ affects one's $Y_i$, which then affects others' $Y_j$, leading to correlation between the average group outcome and the error.

Empirical papers have typically solved this challenge by exploiting additional information such as external variables for instruments or randomization.\footnote{For example, \cite{Sacerdote2001} exploits randomization of individuals to groups, \cite{brock2001interactions} exploits a specific block structure of groups, while other works use Instrumental Variables (IV) based on historical (e.g., \cite{acemoglu2015}) or other external restrictions (\cite{ioannides2003neighbourhood} and \cite{durlauf2008understanding}).} Unfortunately, these solutions are unavailable for many settings. Yet, in a seminal contribution, \cite{bramoulle2009identification} showed that $(\beta, \boldsymbol{\delta}', \boldsymbol{\gamma}')'$ could be identified using only the model above, the existing $X_i$ and the network structure itself, $\mathbf{G}$, where its $ij$ element $g_{ij} > 0$ represents that $j \in P(i)$. They proposed using instruments based on the characteristics of friends-of-friends or higher-order connections (i.e. $\mathbf{G}^k \mathbf{X}$, where $\mathbf{X}$ is the matrix stacking $X_i'$).\footnote{Such instruments are valid because they are excluded from \eqref{struc_model_intro} and $\mathbf{X}$ and $\mathbf{G}$ are exogenous. They are relevant because, in general, multiplying $G_i$ on both sides of \eqref{struc_model_intro} implies that $\bar{Y}_i$ is a function of $G_i^2 X_i$, where $G_i, G_i^2$ represent $i$'s friends and friends-of-friends, respectively.} This solution provided an easy to implement identification strategy and a natural estimator based on instruments readily available to the researcher.

In this paper, we show that "friends-of-friends" instrumental variables are not a panacea for linear-in-means applications in Economics. In particular, we study a specific mechanism related to network topology for why peers-of-peers instruments may not perform well, and characterize these mechanisms in a specific random-graph benchmark. 

First, we show that these instruments can fail for different reasons depending on the density of the network: whether it is sparse or sufficiently dense. To do so, we link the network topology to the first-stage estimand from instrumenting \eqref{struc_model_intro}, and distinguishing between two related but different problems: weak identification and asymptotic degeneracy of the first-stage estimand. In sparse networks, higher-order neighborhoods may contain too little stable variation, making the first-stage weak or, in extreme cases, causing the population first-stage estimand to become ill-defined. In dense or near-homogeneous networks, higher-order links may add little independent variation beyond the original network, again weakening the first-stage. This may result in instruments with low variance and potentially low covariance with the endogenous variable. 

Sparse networks are a prevalent feature in empirical work, so this mechanism is likely to be present in many applied settings. For example, Table \ref{tab:deg_sum} shows the degree distribution statistics for each of the networks $(\mathbf{G})$ as well as the squared counterparts ($\mathbf{G}^2$) used to construct the instrument in salient examples in political economy (alumni networks in the U.S. Congress, \cite{battaglini2018}) and development economics (network of allies/enemies in the Second Congo War, \cite{Konig2017}). The networks are indeed very sparse: the modal degree for $\mathbf{G}$ is 0. In fact, this is also true for $\mathbf{G}^2$! We visualize these networks in Figures \ref{fig:bat_graph}-\ref{fig:konig_graph_allies} and we revisit them further below.



As a second contribution, we characterize regimes where these challenges may arise using a benchmark class of random graphs: Erdős--Rényi random graphs (\cite{erdds1959random}). These graphs are extensively used in both theoretical economics and econometrics as tractable models of network formation (see \cite{jackson10, mele17, campbell24} for examples and discussions). Using tools from random graph theory we characterize how rates at which network sparsity/density, captured by the average degree $d_n$, induce weakness or ill-defined first-stages as network size grows. We formally show that the identification strategy of \cite{bramoulle2009identification} performs well when networks exhibit intermediate levels of connectivity, but first-stage issues may arise as the network becomes extremely sparse or extremely dense. In the baseline model, when the adjacency matrix $\mathbf{G}$ is left unscaled, we find that when the average degree ($d_n$) decreases with \(n\), the variance of the instrument tends to zero faster than its covariance with the endogenous regressor. Then, the first-stage becomes asymptotically ill-defined. On the other extreme, when $d_n$ grows with \(n\), $\mathbf G^{2}$ becomes asymptotically collinear with $\mathbf G$. Thus, the instrument adds little independent variation and first-stage relevance vanishes. Between these extremes, when average degree is bounded, first-stage strength depends on how quickly sampling noise dissipates which, in turn depends on network dependence and the assumed variance structure. We further discuss alternatives that may mitigate these challenges, including scaling the adjacency matrix, which acts to regularize the spectrum of the latter. These, however, come with their own challenges, such as the comparability of the scaled model to the original one.


For settings where weak network-based instruments arise (and scaling is insufficient/infeasible), we adapt standard weak-IV robust inference procedures to the case with peers-of-peers instruments. In particular, we implement the Anderson--Rubin test (\cite{anderson1949estimation}) and combine it with the network-dependent variance estimator of \cite{kojevnikov2021limit}. We show that this yields asymptotically valid inference under weak-IV asymptotics. This approach explicitly takes the network cross-sectional dependence into account and, in Monte Carlo simulations, is shown to perform well. Yet, we further show that using a simpler variance estimator that assumes homoskedasticity still performs very well in finite samples in sparse settings. This is because, in sparse settings, network spillovers are limited and this restricts heteroskedasticity.

We show that accounting for the very sparse nature of some important networks in economics (and their effects on estimation and inference) can lead to different conclusions in empirical examples. In particular, we revisit the setting of \cite{Konig2017} who studied the effects of allied (or enemy) networks across ethnicities in Africa and their effects on conflict. To account for endogeneity, they propose instruments that use the network structure of such linkages. Due to the sparsity of the network shown in Table \ref{tab:deg_sum} and Figure \ref{fig:konig_graph_allies}, their original work already suggested that their instrument was weak. Our proposed inference yields confidence intervals for the parameters of interest ($\beta$ above) that are much larger, and include 0.

Finally, we conclude by discussing that our insights extend beyond the linear-in-means model, to other linear regression models with network-based instruments. Thus, close attention to network structure and its growth with sample size must be considered when implementing such instruments.

\begin{table}[h!]
\centering
\caption{Degree Distribution Statistics for Empirical Graphs and their Second Power}
\label{tab:deg_sum}
\begin{adjustbox}{max width=\textwidth}
\begin{tabular}{|l|cc|cc|cc|}

\hline
& \multicolumn{2}{c|}{US Legislative Network} 
& \multicolumn{4}{c|}{Network of Allies/Enemies in 2nd Congo War} \\
 & \multicolumn{2}{c|}{\cite{battaglini2018}} 
 & \multicolumn{4}{c|}{\cite{Konig2017}} \\

\hline
 & Non-Normalized & Non-Normalized Squared & Allies & Allies Squared 
 & Enemies & Enemies Squared \\
\hline
Min & 0 & 0 & 0 & 0 & 0 & 0  \\
Median & 0 & 0 & 1 & 5.5 & 1 & 24 \\
Mean & 1.8 & 4.2 & 2.4 & 15.3 & 3 & 24.3   \\
Mode & 0 & 0 & 0 & 0 & 1 & 0  \\
Max & 28 & 196 & 21 & 49 & 26 & 94   \\
\hline
\end{tabular}
\end{adjustbox}
\end{table}

The rest of the paper is organized as follows: Section \ref{S2} provides a brief review of literature, while Section \ref{S3} contains the main theoretical results of our paper. Section \ref{S4} provides results on our Monte Carlo simulations, while Section \ref{S5} provides the empirical applications. We conclude in Section \ref{S6}. All proofs are provided in Appendix Section \ref{SA1}.

\section{Related Literature}\label{S2}

There has been a steady growth in the literature dealing with econometric issues related to peer effects and social interactions. 
\cite{manski1993identification} studied the (lack of) identification in the linear-in-means model due to the "reflection problem", spurring a large literature (see \cite{brock2001interactions}, \cite{durlauf2004neighborhood}, and \cite{blume2005identifying}, \cite{bramoulle2020peer} for surveys). Empirical research has proposed different solutions to this problem, including structure on the types of social interactions (e.g., \cite{gaviria2001school}), or the validity of instrumental variables (e.g., \cite{ioannides2003neighbourhood} and \cite{durlauf2008understanding} which use group analogues of individual characteristics satisfying an exclusion restriction, or \cite{acemoglu2015} using historical variables). Others have used randomization as an appropriate identification strategy like \cite{field2016friendship}. However, these solutions are often case specific and not easy to generalize. In the absence of such identification strategies, \cite{bramoulle2009identification} suggested using the network structure itself to generate valid instruments using friends-of-friends' characteristics or other higher powers of the adjacency matrix to construct instruments, as discussed above. We focus on the standard linear-in-means model and their proposed instruments. We derive novel results on the role of the specific sparse/dense-network topology with two associated challenges: weak instruments and ill-defined first-stage estimands. We characterize the rates under a salient benchmark: Erdős--Rényi random graphs.

The possibility of weak instruments due to the correlation between a network and its higher-order counterparts was pointed out in a discussion in \cite{gibbons} (p.179) in the spatial econometrics context, and is excluded from the identification results in \cite{bramoulle2009identification}, as we revisit below. Yet, three recent papers, \cite{tchuente2019weakidentification}, \cite{ross2022} and \cite{Wang2025} study weak identification in social-interaction settings with network-based instruments in the linear-in-means setting.

\cite{tchuente2019weakidentification} studies weak identification arising from high transitivity, which can generate near-rank deficiency in the first stage, and proposes a regularized Two Stage Least Squares (TSLS) procedure to mitigate small-sample bias. They explicitly note that weak identification can also occur when there are too many isolated individuals, but their focus is on highly transitive networks (p.~2). By comparison, our results focus on alternative networks (including extremely sparse networks, salient in empirical settings), we provide new characterization results for Erdős--Rényi graphs, and our inference is not based on the regularized TSLS estimator. A crucial feature of our derived rates is to distinguish between cases where weak identification occurs, and when the first-stage estimand is ill-defined (as the instrument variance goes to 0 faster than its covariance with the endogenous variable).

\cite{Wang2025} studies weak identification under a near-degree-regularity assumption, where neighborhoods become asymptotically identical. Again, our paper differs in several ways: both in the mechanism by which weak identification can occur (here, network sparsity and heterogeneous degree growth are key determinants of identification strength), by our characterization of degenerate, weak and informative regimes in the Erdős--Rényi graphs as a benchmark (including explicit asymptotic rates), by our explicit discussions of the empirical literature and examples where these issues arise, and by adapting and applying weak-IV robust inference for the resulting network-based weak identification problem.

\cite{ross2022} discusses weak-IV robust inference with a type of network-based instruments based on partially overlapping peer groups (from quasi-experimental roommate allocations) which differs from those in \cite{bramoulle2009identification}. 
They explicitly contrast their instruments to those that we study, emphasizing that ``identification from a network configuration can also give rise to weak identification, but \textit{in our case identification comes from a discrete transition over time and weakness in instruments is not due to reliance on the spatial structure of the networks}, but rather due to the relatively weak first-period contextual effects on student performance" (p.998, see also Supplement B.6). Our paper instead studies how the topology of an observed network affects friends-of-friends instruments of the form \(\mathbf G^2\mathbf X\), and combines weak-IV robust inference with network-dependent variance estimation. We also provide Monte Carlo evidence on when conventional \(t\)-test-based inference breaks down and when weak-IV robust procedures deliver more reliable coverage.

As a result, our paper is related to the literature on weak instruments (see \cite{andrews2019weak} for a detailed survey). This literature, beginning with contributions such as \cite{dufour97} and \cite{staigerstock97}, has primarily developed weak-IV theory and robust inference in the classical IV framework. We adapt these tools to linear-in-means models with network-based instruments. While the linear-in-means model shares many similarities with the linear IV model, first-stage strength in the former depends on network topology and growth, and spillovers may introduce cross-sectional dependence via correlation between relevant sample moments across connected individuals. Thus, we show how weak-IV robust procedures, such as Anderson--Rubin and conditional likelihood ratio tests \citep{anderson1949estimation,moreira03}, can be implemented in this environment together with variance estimators that account for network dependence, building on results such as \cite{kojevnikov2021limit} and \cite{conley1999gmm}.

\section{Weak Instruments in the Linear-in-Means Model}\label{S3}

\subsection{Model}

We use the extended linear-in-means model of \cite{bramoulle2009identification} where \(n\)-agents interact over an exogenously given network with the $N \times N$ adjacency matrix $\mathbf{G}$. Each element of this adjacency matrix is given by $g_{ij}$, where:
\begin{align*}
    g_{ij} = \begin{cases}
    1 & \text{if $i$ is connected to $j$ for } i \neq j  \\
    0 & otherwise
    \end{cases}
\end{align*}
We consider the matrix version of the structural model in equation \eqref{struc_model_intro} and, for simplicity, consider the case without group fixed-effects (called correlated effects), as they can be differenced out (see \cite{bramoulle2009identification}). Formally,
\begin{align}\label{struc_model}
    \mathbf{Y} = \alpha \mathbf{\iota} + \beta \mathbf{GY} + \mathbf{X} \boldsymbol{\gamma} + \mathbf{G X} \boldsymbol{\delta} + \mathbf{\varepsilon},
\end{align}
where $\mathbf Y$ and $\varepsilon$ is the $N \times 1$ vector of outcomes and errors, respectively, and $\mathbf{X}$ the $N \times d$ matrix of observable characteristics, and we have separated the constant from $\mathbf{X}$. $\boldsymbol{\gamma}$ and $\boldsymbol{\delta}$ are the $d  \times 1$ vector of coefficients associated with $\mathbf{X}$ and $\mathbf{GX}$ respectively. We make the following standard assumptions:
\begin{assumption}\label{A3.1}
\begin{enumerate}[(i)]
\item
$(\mathbf{\iota},\mathbf{X})$ is full rank, 
\item 
$|\beta|<1$, 
\item 
$\mathbf{G} \neq \mathbf{0}$, 
\item
$\mathbb{E}[\varepsilon | \mathbf{X}, \mathbf{G}] = 0$
\end{enumerate}
\end{assumption}

Assumption \ref{A3.1} (i) and (ii) are standard to guarantee the behavior of the model, including the invertibility and stability of the system. Condition (ii) ensures that ($\mathbf{I} - \beta \mathbf{G}$) is invertible and can be expanded into an infinite matrix (Neumann) series. Condition (iii) assumes that not all nodes are isolated, necessary for identification and non-trivial results. Then, (iv) assumes strict exogeneity of individual characteristics and the network respectively. Note that endogeneity still arises due to the reflection problem. Now, we can write the reduced-form of (\ref{struc_model}) as:
\begin{align}\label{red_model}
    \mathbf{Y} = \alpha (\mathbf{I} - \beta \mathbf{G})^{-1} \mathbf{\iota} + (\mathbf{I} - \beta \mathbf{G})^{-1} (\mathbf{X} \boldsymbol{\gamma} + \mathbf{G} \mathbf{X} \boldsymbol{\delta} )  + (\mathbf{I} - \beta \mathbf{G})^{-1} \mathbf{\varepsilon}.
\end{align}
Under Assumptions \ref{A3.1} (i)-(ii),we can expand $(\mathbf{I} - \beta \mathbf{G})^{-1}$ into a Neumann series:
\begin{align}\label{exp_eqn}
\mathbf{Y} = \alpha (\mathbf{I} - \beta \mathbf{G})^{-1} \mathbf{\iota} + \mathbf{X} \boldsymbol{\gamma} + \sum_{k=0}^{\infty}\beta^k\mathbf{G}^{k+1} \mathbf{X}(\beta \boldsymbol{\gamma} + \boldsymbol{\delta} ) + \sum_{k=0}^{\infty}\beta^k\mathbf{G}^{k+1} \mathbf{\varepsilon}
\end{align}
Then, from the strict exogeneity assumption, we get:
\begin{align}\label{cond_mean}
\mathbb{E}(\mathbf{GY} | \mathbf{X}) = \frac{\alpha}{(1-\beta)} \mathbf{G} \iota + \mathbf{GX} \boldsymbol{\gamma} + \sum_{k=0}^{\infty}\beta^k\mathbf{G}^{k+2} \mathbf{X}(\beta \boldsymbol{\gamma} + \boldsymbol{\delta} ).
\end{align}

\subsection{Identification}

The model is said to be identified if $\boldsymbol{\theta} = (\alpha, \beta, \boldsymbol{\gamma}', \boldsymbol{\delta}')'$ is identified.\footnote{We assume there is a super-population of exogenous networks from which the sample $\mathbf{G}$ is drawn, thereby defining identification relative to this super-population, and the DGP given in equation \eqref{struc_model}.} The reflection problem discussed by \cite{manski1993identification} is evident from equation \eqref{struc_model} where individual outcomes are affected by the respective expected group outcome which in itself is impacted by the former. \cite{bramoulle2009identification} suggested using the exogenously given network structure to construct valid instruments. The main identification result of their paper states that, as long as attributes/characteristics of neighbors have some direct or indirect effect, i.e. $( \beta \boldsymbol{\gamma} + \boldsymbol{\delta}) \neq \mathbf{0}$ and $\mathbf{I}$, $\mathbf{G}$ and its higher powers are not linearly dependent, then $\mathbf{G}^k \mathbf{X}$ for $k \geq 2 $ can be used as valid instruments for $\mathbf{GY}$. This is summarized in their Proposition 1, rewritten for convenience below.

\begin{proposition}[\cite{bramoulle2009identification}, Proposition 1, for non-row normalized $\mathbf{G}$]\label{prop1}
Suppose that Assumption 3.1 holds, that $\beta \boldsymbol{\gamma} + \boldsymbol{\delta} \neq 0$ and that the matrices $I, \mathbf{G}, \mathbf{G}^2$ are linearly independent. Then, the social effects $\mathbf{\theta}=(\alpha, \beta, \boldsymbol{\gamma}', \boldsymbol{\delta}')'$ are identified. 
\end{proposition}

Hence, if the first-stage is given by
\begin{align}\label{eq:first_stage}
    \mathbf{GY} = \Tilde{\alpha}\iota +  \mathbf{X} \boldsymbol{\nu} +  \mathbf{G X} \boldsymbol{\Tilde{\gamma}} + \mathbf{G}^{2} \mathbf{X} \boldsymbol{\pi} + \tilde{\epsilon},
\end{align}
the structural equation \eqref{struc_model} can be re-written as,
\begin{align}\label{eq:struct_eq2}
    \mathbf{Y} = (\alpha + \beta \Tilde{\alpha}) \mathbf{\iota} +   \mathbf{X} (\boldsymbol{\gamma} + \beta \boldsymbol{\nu}) + \mathbf{GX} (\boldsymbol{\delta} + \beta \boldsymbol{\Tilde{\gamma}})  +  \mathbf{G}^2 \mathbf{X} \boldsymbol{\xi} + \mathbf{\eta},
\end{align}
where $\eta = \epsilon + \beta \Tilde{\epsilon}$ and $\boldsymbol{\xi}  = \beta \boldsymbol{\pi}$. Thus, the endogenous effect $\beta$ represents the proportionality constant linking the coefficient on $\mathbf{G}^{2} \mathbf{X}$ in the first-stage regression to that in the structural equation \eqref{eq:struct_eq2} which can be estimated using an appropriate estimator $\hat{\boldsymbol{\xi}}  = \hat{\beta} \hat{ \boldsymbol{\pi}}$.

\subsection{Empirical Challenges with Network-Based Instruments}

Proposition~\ref{prop1} shows that, for non-trivial combinations of parameters, identification hinges on the informational content of friends-of-friends' networks ($\mathbf{G}^2$) relative to the network $\mathbf{G}$ and to a constant. See also \cite{gibbons} in the spatial econometrics context. This suggests that, even when the exclusion restriction and rank conditions hold in principle, the effectiveness of $\mathbf{G}^2 \mathbf X$ as an instrument depends critically on whether it introduces sufficient independent variation beyond $\mathbf{GX}$. In the next section, we formalize the conditions for this to arise in a salient class of random graph models. Here, we outline the main ideas that apply more generally.

First, we demonstrate that the key assumptions in Proposition \ref{prop1}, namely the linear independence of $\mathbf{G}$ and $\mathbf{G}^2$ as well as $\mathbf{G}^2$ being non-zero (see footnote 23 of \cite{bramoulle2009identification}), are likely violated (or close to violated) in many empirical settings. Table \ref{tab:deg_sum} already provides evidence of this problem arising in various empirically observed networks. This is also easily observed when the graphs themselves are plotted in Figures \ref{fig:bat_graph} and \ref{fig:konig_graph_allies}. 
 These examples illustrate situations where $\mathbf{G}^2$ is extremely sparse. On the other hand, Figure \ref{fig:cruz_graph} showcases a much denser network where higher-order links become very close to a completed network. Our second main contribution is to characterize the behavior of the network-based instruments in different scenarios. The two extremes mentioned above lead to distinct implications for identification and inference.

\begin{figure}[ht]
    \centering
    \captionsetup[subfigure]{justification=centering}
    \begin{subfigure}[b]{0.45\textwidth}
        \centering
        \includegraphics[width=\textwidth]{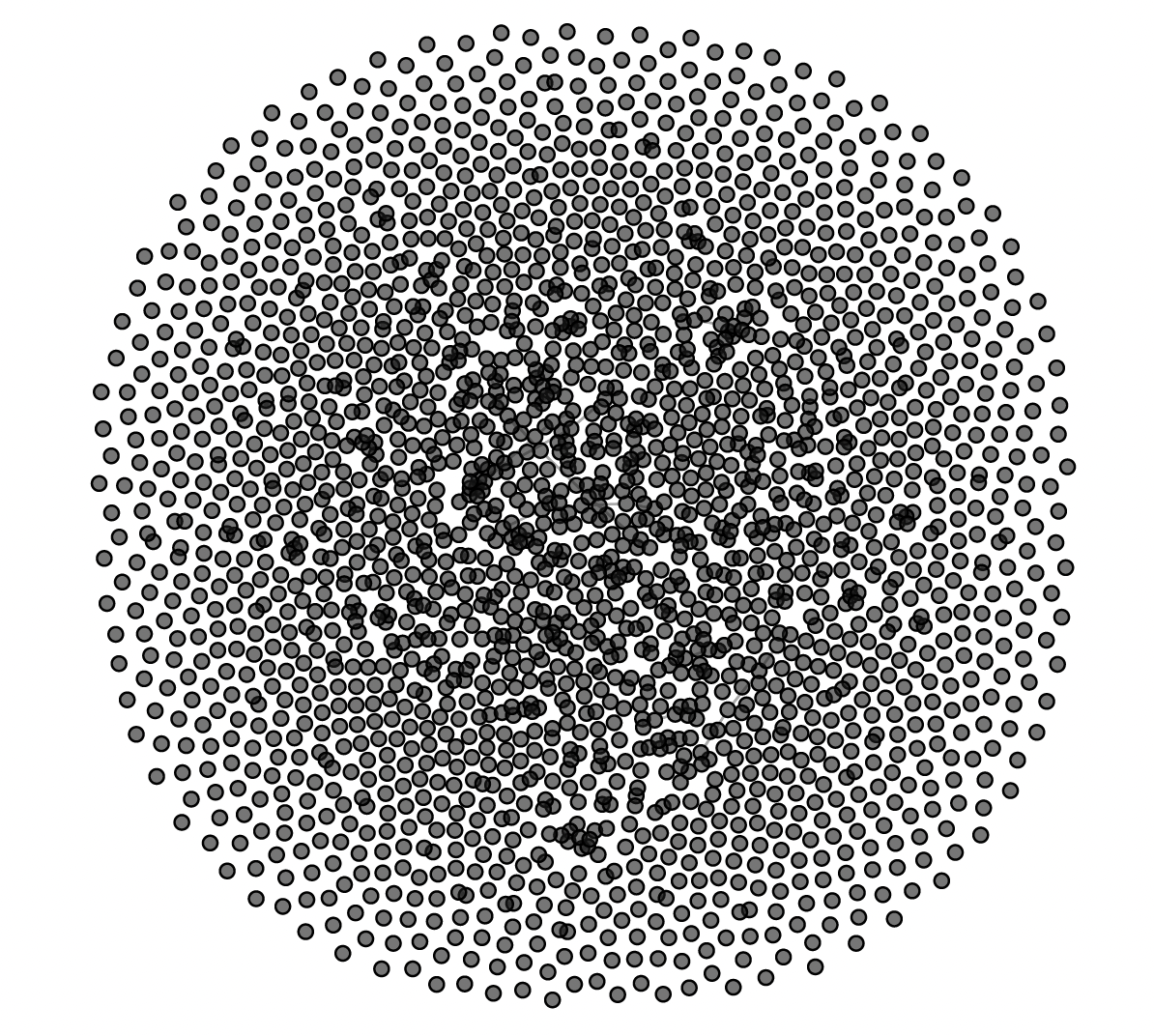}
        \caption{Network}
        \label{fig:subBa}
    \end{subfigure}
    \hfill
    \begin{subfigure}[b]{0.45\textwidth}
        \centering
        \includegraphics[width=\textwidth]{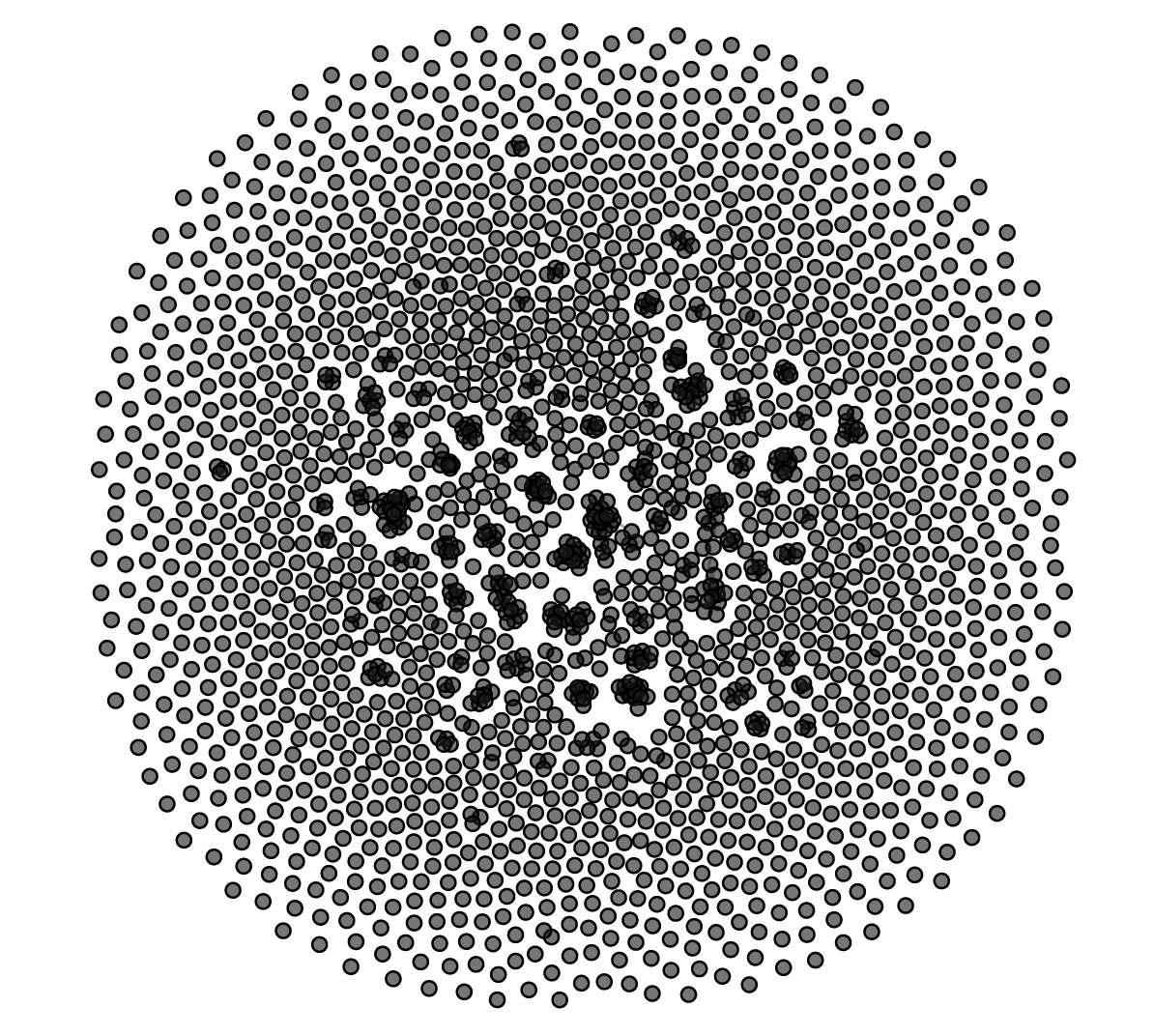}
        \caption{Network-squared}
        \label{fig:subBb}
    \end{subfigure}
    \caption{US Congressional Alumni Network 109-113 (\citealp{battaglini2018})}
    \label{fig:bat_graph}
\end{figure}

\begin{figure}[ht]
     \centering
    \captionsetup[subfigure]{justification=centering}
    \begin{subfigure}[b]{0.45\textwidth}
        \centering
        \includegraphics[width=\textwidth]{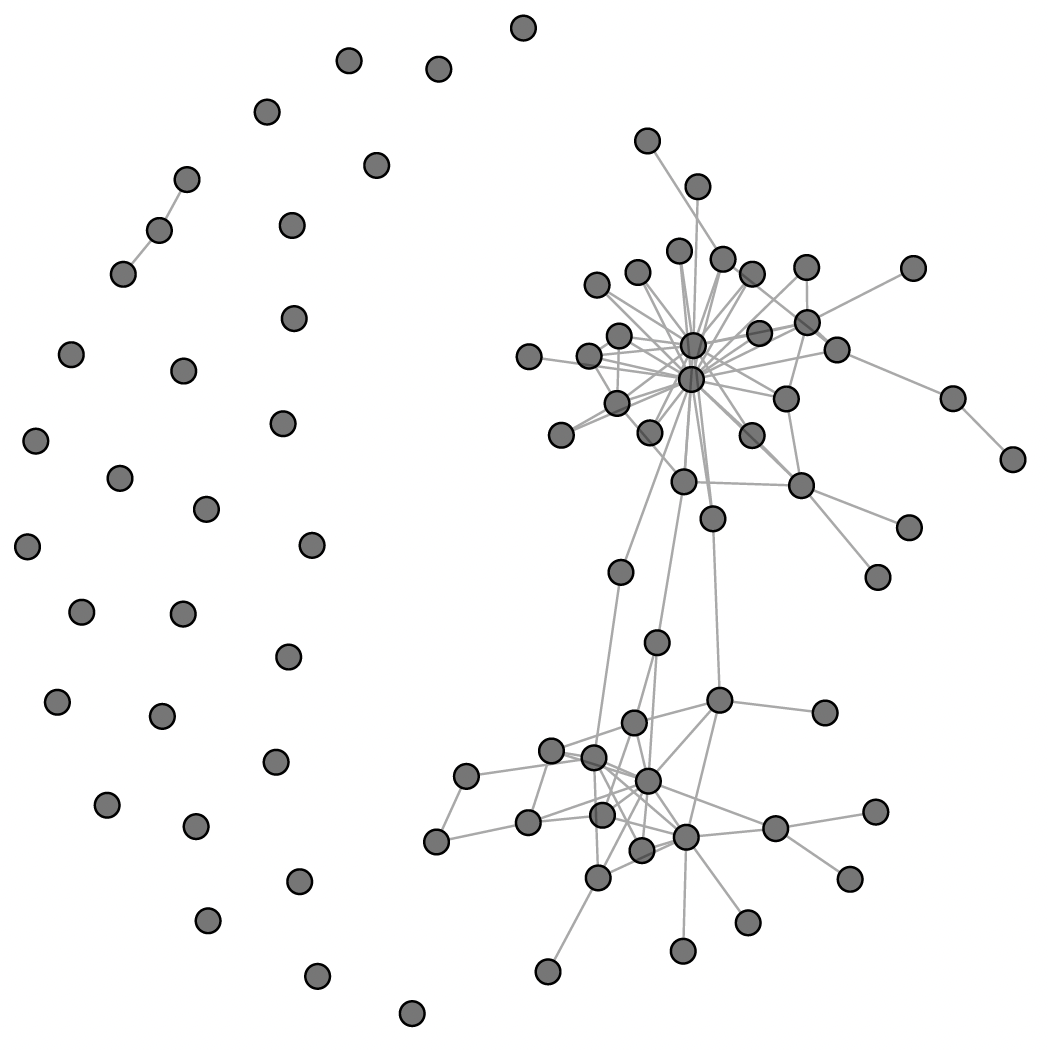}
        \caption{Network}
        \label{fig:subKa_allies}
    \end{subfigure}
    \hfill
    \begin{subfigure}[b]{0.45\textwidth}
        \centering
        \includegraphics[width=\textwidth]{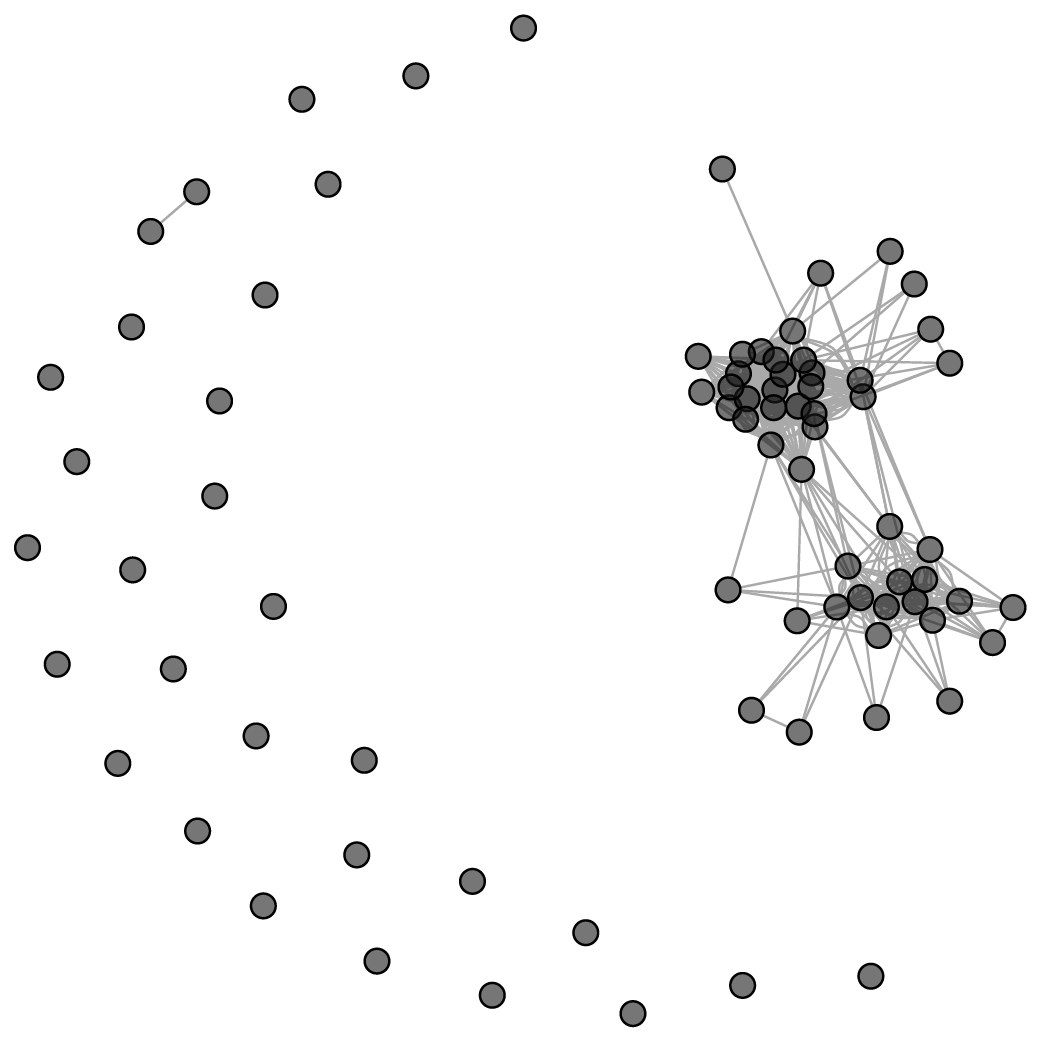}
        \caption{Network-squared}
        \label{fig:subKb_allies}
    \end{subfigure}

    \caption{Network of Allies from the Second Congo War (\citealp{Konig2017})}
    \label{fig:konig_graph_allies}
\end{figure}

\begin{figure}[ht]
    \centering
    \captionsetup[subfigure]{justification=centering}
    \begin{subfigure}[b]{0.45\textwidth}
        \centering
        \includegraphics[width=\textwidth]{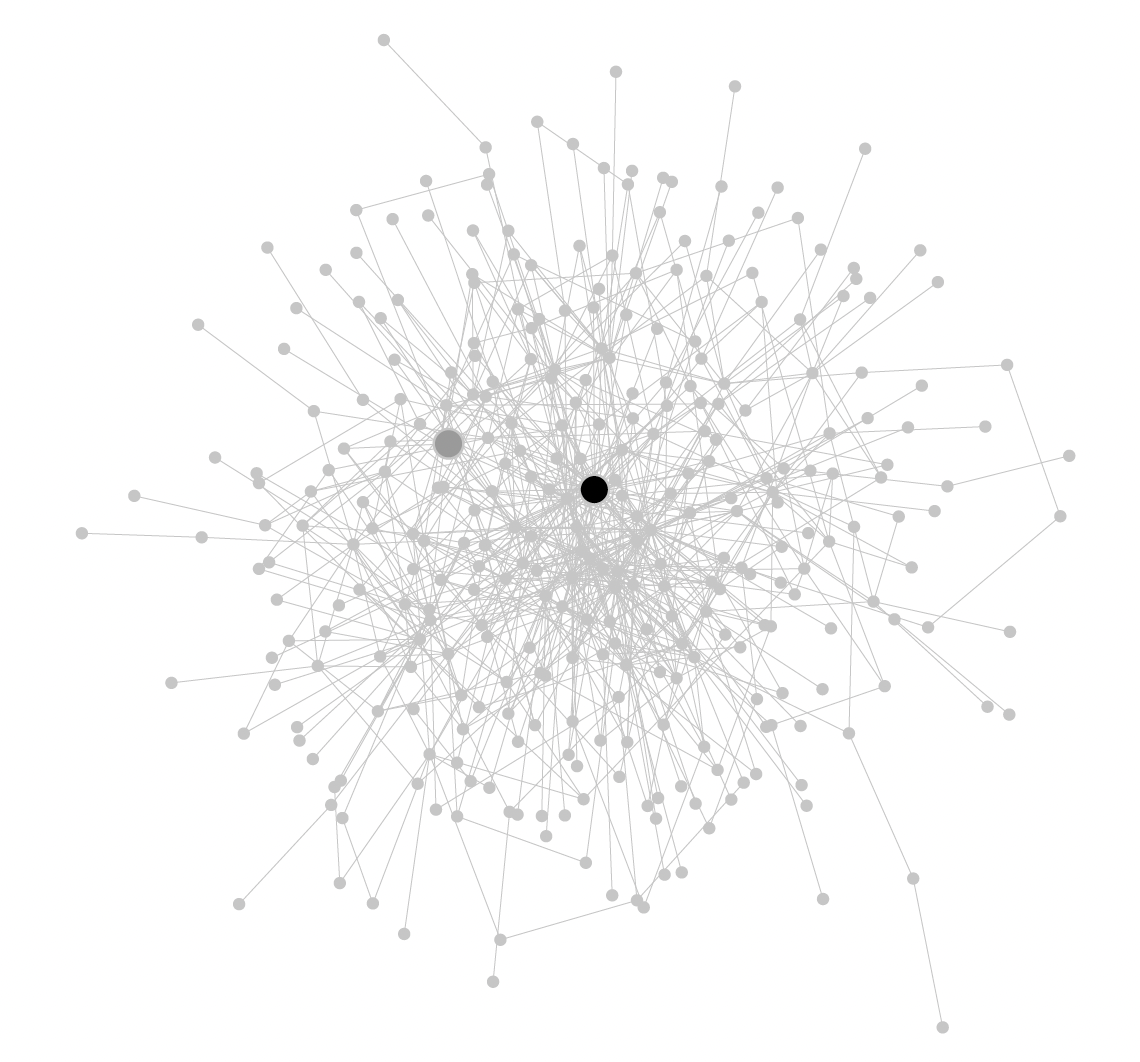}
        \caption{Network}
        \label{fig:subCa}
    \end{subfigure}
    \hfill
    \begin{subfigure}[b]{0.45\textwidth}
        \centering
        \includegraphics[width=\textwidth]{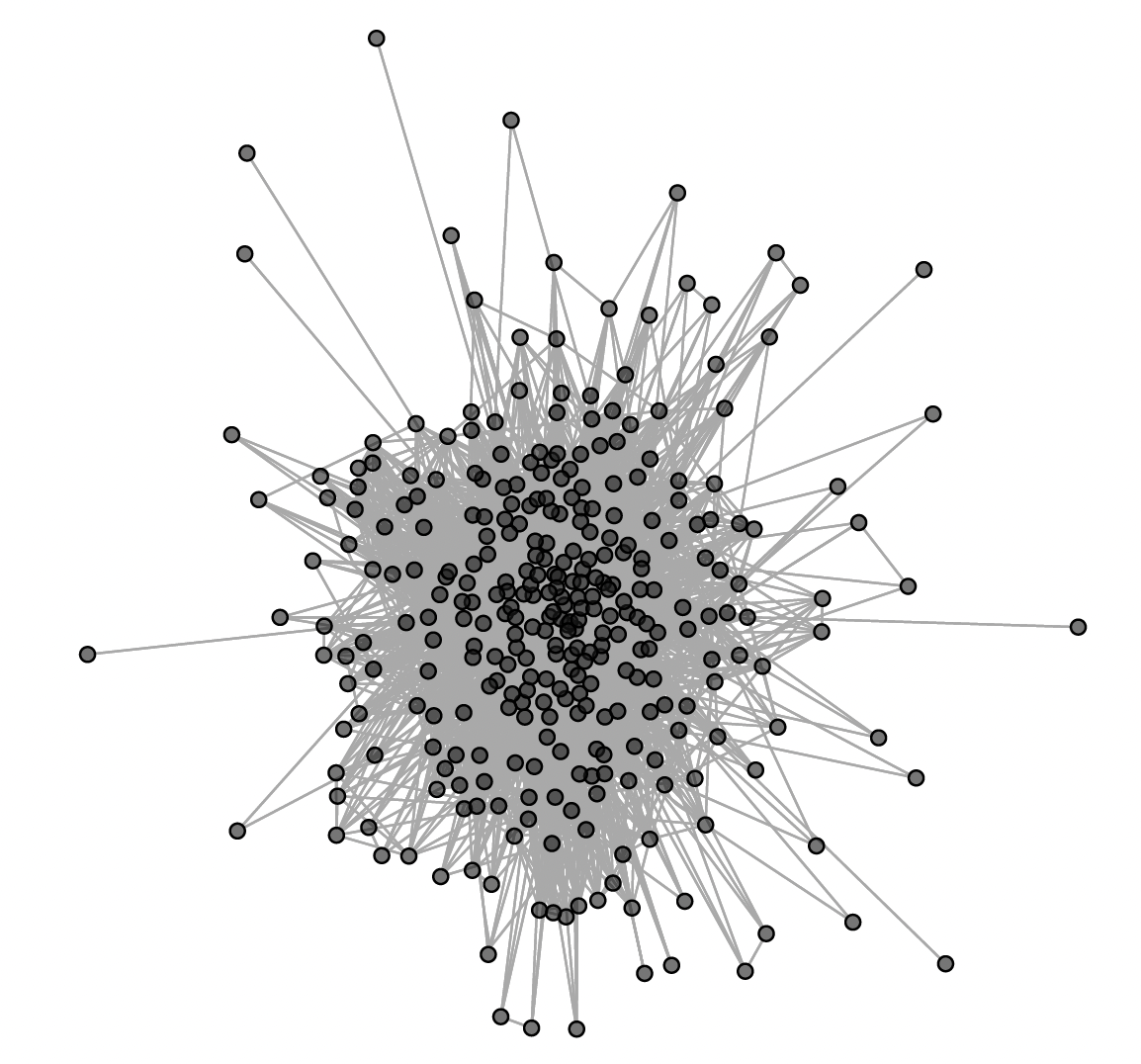}
        \caption{Network-squared}
        \label{fig:subCb}
    \end{subfigure}

    \caption{Family Network from sample municipality from Philippines (\citealp{Cruz2017})}
    \label{fig:cruz_graph}
\end{figure}

If network sparsity increases with \(n\), \(\mathbf G^2\) collapses to the zero matrix faster than \(\mathbf G\). The instrument then becomes asymptotically degenerate: its variance tends to zero faster than its covariance with the endogenous variable. This is not identification failure in the usual sense, but rather asymptotic degeneracy of the instrument itself, where the limit estimand is ill-defined. In the opposite extreme, when the network becomes very dense, \(\mathbf G^2\) becomes approximately proportional to \(\mathbf G\). In this case, \(\mathbf G^2\mathbf X\) adds little independent variation beyond \(\mathbf GX\), the population first-stage coefficient tends to zero, and identification becomes weak through asymptotic collinearity. This is the dense-network mechanism highlighted in \cite{bramoulle2009identification} and further studied by \cite{Wang2025} in near-regular graphs. Between these two extremes, \(\mathbf G^2\mathbf X\) may remain non-degenerate and non-collinear, but still be weak when the covariance between \(\mathbf G\mathbf Y\) and \(\mathbf G^k\mathbf X\), \(k\ge2\), is small relative to sampling noise. Conversely, in moderately connected networks where higher-order instruments retain independent variation, the first stage can remain strong.

Following the literature on weak identification,\footnote{In particular \cite{staigerstock97}, \cite{stock00} and \cite{stock2005testing}. See \cite{stockwrightyogo2002} and \cite{andrews2019weak} for a detailed survey and discussion.} we characterize the weak instruments problem in our set-up through the population first-stage coefficient $\boldsymbol{\pi}$ in \eqref{eq:first_stage}. As emphasized by \cite{andrews2019weak}, non-standard asymptotic behavior of IV estimators arises when the population first-stage coefficient $\boldsymbol{\pi}$ is small relative to the sampling variability of the estimator $\hat{\boldsymbol{\pi}}$. This is summarized using a concentration-parameter, $\mu$.\footnote{In the standard setting, the concentration parameter is given by $ \mu^2 \equiv \boldsymbol{\pi}' Var(\hat{\boldsymbol{\pi}} | \mathbf{G}, \mathbf{X})^{-1} \boldsymbol{\pi}$. Instruments are weak when $\mu^2 = O(1)$ and strong when $\mu^2 \rightarrow \infty$.} In the standard \textit{i.i.d.} case studied in \cite{staigerstock97}, $\hat{\boldsymbol{\pi}}$ concentrates at the rate $n^{-1/2}$. Hence, $\mu \asymp \sqrt{n} |\boldsymbol{\pi}|$ and instruments are weak when $\boldsymbol{\pi}$ is local to zero at the same rate, i.e.\ when $\boldsymbol{\pi} = O(n^{-1/2})$. In models with network dependence, however, the concentration rate of the first-stage estimator depends on the strength and structure of cross-sectional dependence induced by the network (see \cite{kojevnikov2021limit} and \cite{Wang2025} among others). Thus, to adapt the standard definition of weak instruments to this set-up, we introduce $k(n)$ to denote the network-dependent concentration rate. To do so, we must introduce additional notation. 

Let $Z_n=\mathbf{G}^{k}\mathbf{X}$ with $k\ge2$, where the subscript \(n\) emphasizes the dependence of the instrument on the size and structure of the network. Let $W = (\iota,\mathbf{X},\mathbf{G}\mathbf{X})$ and define the residualized instrument $\tilde{Z}_n := M_W Z_n$, where $M_W=I-W(W'W)^{-1}W'$. Let $\widetilde{\mathbf{G} \mathbf Y}:=M_W(\mathbf{G} \mathbf Y)$ denote the corresponding residualized endogenous regressor. We let $\boldsymbol{\pi}_n$ denote the population first-stage coefficient, which is allowed to depend on \(n\) and is given by the Frisch--Waugh--Lovell representation
\[
\boldsymbol{\pi}_n
:=
\left(\frac{1}{n}Var(\tilde Z_n)\right)^{-1}
\frac{1}{n}Cov\!\left(\widetilde{\mathbf{G} \mathbf Y},\,\tilde Z_n\right),
\]
where
$Var(\tilde Z_n)=\mathbb{E}[\tilde Z_n'\tilde Z_n]$ and $Cov(\widetilde{\mathbf{G}Y},\tilde Z_n) =\mathbb{E}[\tilde Z_n'\widetilde{\mathbf{G}Y}]$.\footnote{This dependence arises because both the strength of the instrument and the sampling variability of its estimator are functions of the network’s size and connectivity, so asymptotic behavior is governed by the sequence of networks $\{\mathbf{G}_n\}_{n\ge1}$ rather than by sample size alone.}

The sampling variability of \(\hat{\boldsymbol{\pi}}_n\) is driven by the first-stage error $\tilde\epsilon_n := \widetilde{\mathbf{G}Y}-\tilde Z_n\boldsymbol{\pi}_n$. Conditional on \(\mathbf{G}\), the variance of the first-stage estimator is therefore
\[
Var(\hat{\boldsymbol{\pi}}_n\mid\mathbf{G}, \mathbf{X})
=
\left(\frac{1}{n}Var(\tilde Z_n)\right)^{-1}
\frac{1}{n}\boldsymbol{\Omega}_{\pi,n}
\left(\frac{1}{n}Var(\tilde Z_n)\right)^{-1},
\]
where
\[
\boldsymbol{\Omega}_{\pi,n}
:=
Var\!\left(
\frac{1}{\sqrt{n}}\tilde Z_n'\tilde\epsilon_n
\;\middle|\;
\mathbf{G}, \mathbf{X}
\right).
\] Accordingly, define the network-dependent information index
\[
k(n)
\;\asymp\;
\left\|
\,Var(\hat{\boldsymbol{\pi}}_n\mid\mathbf{G}, \mathbf{X})^{-1/2}
\right\|,
\] implying that $k(n)$ grows at the rate as the inverse of the conditional standard deviation of the estimator.\footnote{This variance normalization coincides with that used by \cite{kojevnikov2021limit} to establish a central limit theorem for network-dependent quadratic forms. Under conditional homoskedasticity of \(\tilde\epsilon_n\) given \(\mathbf{G}\), \(\boldsymbol{\Omega}_{\pi,n}\) is proportional to \(\mathbb{E}[\tilde Z_n'\tilde Z_n\mid\mathbf{G}]\), implying that the information content of the first-stage is governed—up to constants—by the Frobenius norm \(\|\mathbf{G}^k\|_F^2\). This scaling is also adopted in \cite{Wang2025}, who shows that it yields stable Gaussian limits under both sparse and dense network sequences.} Thus, the results that follow use the following definition for weak identification with network-based instruments.
\begin{definition}[Weak identification in networks]
\label{def:weak_network}
Instruments of the form $Z_n=\mathbf{G}^k\mathbf{X}$ are said to be \emph{weak} if there exists a finite and fixed matrix $\mathbf{C}$ such that
\[
k(n) \boldsymbol{\pi}_n  \le
\mathbf{C},
\] for large \(n\). 
\end{definition}

Equivalently, instruments are weak when $k(n) \boldsymbol{\pi}_n=O(1)$ as $n\to\infty$. When $\mathbf{X}$ is uni-dimensional, Definition~\ref{def:weak_network} reduces to $k(n)\,\pi_n
=
k(n)\,
\frac{\frac{1}{n}Cov(\widetilde{\mathbf{G}Y},\tilde Z_n)}
     {\frac{1}{n}Var(\tilde Z_n)}
=
O(1),
\quad n\to\infty$, where we use the expression for the first-stage coefficient.

Definition \ref{def:weak_network} mirrors the local-to-zero framework of \cite{staigerstock97}, with \(k(n)\) replacing the usual \(\sqrt n\) rate when the first-stage coefficients share a common network-dependent convergence rate.\footnote{More generally, weakness can be stated through the effective concentration
parameter
\[
\mu_n^2
=
\boldsymbol{\pi}_n'
Var(\hat{\boldsymbol{\pi}}_n\mid \mathbf G,\mathbf X)^{-1}
\boldsymbol{\pi}_n
=O(1).
\]
In the scalar case, \(\mu_n \asymp k(n)|\pi_n|\).} Sparsity therefore leads either to degeneracy (when $\mathbf{G}^2\to 0$ too quickly) or to weak identification (when $\mathrm{Cov}(\mathbf{G} \mathbf Y,\mathbf{G}^k \mathbf{X})$ decreases faster than $\mathrm{Var}(\mathbf{G}^k \mathbf X)$ relative to the sampling noise). Intuitively, in a sparse network the friends-of-friends matrix $\mathbf{G}^2$ (and other higher order terms) contains many zeros and a few nodes with disproportionately large reach. As a result, $\mathbf{G}^k \mathbf X$ tends to have low covariance with $\mathbf{G} \mathbf Y$ in the sparse case.

\subsection{Challenges with Peers-of-Peers Instruments in Random Graph Models}

In which settings are network-based IVs likely to be ill-defined, lead to weak identification or standard identification and inference? To answer this question and to explore the above characterization, we examine the behavior of peers-of-peers instruments within the widely studied Erdős–Rényi (ER) random graph model. 

In an ER graph $G(n,p)$, each potential link between two nodes is formed independently with probability $p$. Sparsity or density of the network is therefore governed entirely by $p$, or equivalently by the average degree $d(n) = n p(n)$. Although conceptually simple, these graphs have been widely studied and can form a foundation for more complex models by providing a useful benchmark (\cite{jackson10}). Within Economics, for example, some structural models of network formation are asymptotically indistinguishable from ER graphs (\cite{mele17}), they are used to model such phenomena as market entry and diffusion (e.g., \cite{campbell24}) and they are also a basis for many simulation designs and comparisons (e.g. \cite{graham20}).

\subsubsection{Theoretical Results}\label{sec:theoretical}
Our characterization and proofs rely on results linking the spectral norm of the adjacency matrix to the degree distribution and use it to establish bounds on the first-stage coefficient. More specifically, we require $(\mathbf{I} - \beta \mathbf{G})$ to be invertible for the first-stage covariance to be finite and bounded. A sufficient condition for this is the following assumption.
\medskip
\begin{assumption}\label{A3.2}
\begin{eqnarray}
 \lambda^G_{1}(n)< \frac{1}{|\beta|} \label{A5},
 \end{eqnarray}
 where $\lambda^G_1(n)$ is the largest eigenvalue of the adjacency matrix $\mathbf{G}_n$, and we make explicit the dependency of $\mathbf{G}$ on the sample size \(n\).\footnote{This follows from the restriction $ 1>||\beta \mathbf{G}_n ||_2= |\beta| || \mathbf{G}_n ||_2 $, where $||G||_2$ is the spectral-norm of $\mathbf{G}$. When $\mathbf{G}_n$ is symmetric, as in the case of undirected graphs, its spectral norm coincides with the largest eigenvalue.}
\end{assumption}
This assumption holds automatically when the degree sequence is uniformly bounded, a condition imposed in several empirical and theoretical network models to control the size of peer effects and maintain stable influence (e.g., \cite{dePaula2018, Leung2020}). This is also a common assumption made by discrete choice peer effect models such as \cite{lambotte2025peer}, required for existence and uniqueness of equilibrium with peer effects. However, this assumption may fail when connectivity, and hence the largest eigenvalue, grows with network size. In the subsequent sub-section, we present an alternative specification for such cases.

In the following proposition, we establish an upper bound on the variance-normalized covariance between the endogenous variable and the network-based instrument on ER graphs. We find that contingent on the regime that we are in (defined by the average degree), very different conclusions about weak instruments and identification generally arise. (Note that for tractability, we work with the non-residualized variance--normalized covariance since it is a convenient proxy for the partial ratio in Definition \ref{def:weak_network} and remains informative about first-stage relevance.)

\begin{prop}[Upper bound on variance--normalized covariance in Erd\H{o}s--R\'enyi graphs]
\label{prop:cov_bound_unscl}
Let $\{G(n,p_n)\}_{n\ge1}$ be a sequence of Erd\H{o}s--R\'enyi graphs with adjacency matrix $\mathbf A_n$, expected degree $d_n=np_n$, and maximum degree $\Delta_n$. Let $\{X_i\}_{i=1}^n$ i.i.d.\ uni-dimensional real-valued random variables with $\mathbb E[X_i]=0$ and $\mathbb E[X_i^2]=\sigma_x^2\in(0,\infty)$, independent of $\mathbf A_n$. Suppose Assumption~\ref{A3.1} holds. Define
\[
\mathbf G_n := \mathbf A_n,
\qquad
\mathbf G_n^{(2)} := \mathbf G_n^2 - \mathbf D_n,
\qquad
\mathbf D_n := \operatorname{diag}(\mathbf G_n^2).
\]
If Assumption~\ref{A3.2} holds, then there exists a constant $c<\infty$ such that,
for all sufficiently large \(n\),
\begin{equation}\label{eq:upper_unscaled}
\left|
\frac{\frac{1}{n}\operatorname{Cov}(\mathbf G_n^{(2)}\mathbf X,\mathbf G_n\mathbf Y)}
     {\frac{1}{n}\operatorname{Var}(\mathbf G_n^{(2)}\mathbf X)}
\right|
\;\le\;
c\,\frac{1}{\sqrt{d_n + d_n^3/n}} .
\end{equation}
\end{prop}

Proposition~\ref{prop:cov_bound_unscl} establishes that the variance--normalized covariance between the friends-of-friends instrument and the endogenous regressor is asymptotically bounded above by a function of the expected degree. Moreover, under asymptotics for ER graphs and the assumptions of the Proposition, the variance of the friends-of-friends instrument admits a sharp rate:
\[
\frac{1}{n}\operatorname{Var}(\mathbf G_n^{(2)}\mathbf X)
=
\frac{\sigma_x^2}{n}\,\mathbb E\|\mathbf G_n^{(2)}\|_F^2
=
\sigma_x^2\left(d_n^2 + \frac{d_n^4}{n}\right)
+
O\!\left(\frac{d_n^2}{n}+\frac{d_n^4}{n^2}\right).
\]
In contrast, the corresponding covariance admits an exact decomposition\footnote{This decomposition follows from substituting the linear representation  $\mathbf Y=(\mathbf I-\beta\mathbf G_n)^{-1}(\alpha\iota+\gamma\mathbf X+\delta\mathbf G_n\mathbf X+\boldsymbol\varepsilon)$  into $\operatorname{Cov}(\mathbf G_n^{(2)}\mathbf X,\mathbf G_n\mathbf Y)$ and using the Neumann-series expansion $(\mathbf I-\beta\mathbf G_n)^{-1}=\sum_{k\ge0}\beta^k\mathbf G_n^k$, which converges under Assumption~\ref{A3.2}.}
 in which the leading term is proportional to the same quantity:
\[
\frac{1}{n}\operatorname{Cov}(\mathbf G_n^{(2)}\mathbf X,\mathbf G_n\mathbf Y)
=
\frac{\sigma_x^2}{n}\Big(
(\beta\gamma+\delta)\,\mathbb E\|\mathbf G_n^{(2)}\|_F^2
+
\mathbb E[R_n]
\Big),
\]
where
\begin{align*}
R_n
&=
\gamma\,\operatorname{Tr}(\mathbf G_n^{(2)}\mathbf G_n)
+
\gamma\sum_{k=2}^{\infty}\beta^k
\operatorname{Tr}\!\big(\mathbf G_n^{(2)}\mathbf G_n^{k+1}\big)
+
\delta\sum_{k=1}^{\infty}\beta^k
\operatorname{Tr}\!\big(\mathbf G_n^{(2)}\mathbf G_n^{k+2}\big).
\end{align*}
This representation implies that the covariance cannot decay faster than the instrument variance: the leading term is of the same order as $\|\mathbf G_n^{(2)}\|_F^2$, while higher-order contributions are controlled by powers of $\beta$ under Assumption~\ref{A3.2}. Consequently, when $d_n=o(1)$ and the graph collapses asymptotically, the variance of $\mathbf G_n^{(2)}\mathbf X$ converges to zero rapidly while the covariance does not vanish faster, so the variance--normalized covariance diverges. In this extremely sparse regime, the population first-stage estimand is therefore ill-defined. 

In the other extreme, when the graph becomes asymptotically dense as $d_n$ increases, higher-order neighborhoods become nearly deterministic and $\mathbf G_n^{(2)}$ becomes asymptotically proportional to $\mathbf G_n$. Consequently, $\mathbf G_n^{(2)}\mathbf X$ becomes asymptotically collinear with $\mathbf G_n\mathbf X$, and the friends-of-friends instrument adds little independent variation beyond first-order neighbors. This mirrors the dense-network identification failure documented in \cite{bramoulle2009identification} and \cite{Wang2025}. In this regime, the upper bound in Proposition~\ref{prop:cov_bound_unscl} vanishes, the population first-stage coefficient $\pi_n$ shrinks, and identification fails due to asymptotic collinearity.

Between these extremes lies the empirically relevant case in which the expected degree is asymptotically bounded, $d_n=O(1)$. In this regime, the upper bound derived in Proposition~\ref{prop:cov_bound_unscl} does not force the population first-stage coefficient $\pi_n$ to diverge or vanish, so the first-stage estimand is well-defined. Identification strength is governed instead by the information index $k(n)$ following Definition~\ref{def:weak_network}. Weak identification in this regime arises when $k(n)$ remains bounded or, in fact, falls, so that sampling uncertainty in the first-stage does not vanish with \(n\). A natural case when $d_n=O(1)$ is with the presence of a large fraction of isolated or weakly connected nodes.\footnote{In Erdős–Rényi graphs with bounded average degree, the probability that a node has degree zero or one does not vanish asymptotically. In fact, when $d_n<1$, the graph fails to form a giant component (see \cite{erdds1959random}) and, hence, a non-negligible fraction of nodes are isolated.} Thus, increasing \(n\) may not eliminate the presence of many isolated or near-isolated nodes with empty or small higher-order neighborhoods. As these nodes contribute little variation to $\mathbf G_n^{(2)}\mathbf X$, 
sampling noise may not vanish at standard rates and first-stage information may accumulate extremely slowly, giving rise to weak instruments.

Taken together, these results show that network topology matters for the first-stage strength of peers-of-peers instruments. Sparse networks can leave too little higher-order variation for identification, while growing connectivity or asymptotic degeneracy can violate the stability conditions needed for the network operator to remain well behaved.

\subsubsection{Scaling: An Alternate Specification}\label{scaling}

As previously discussed, Assumption~\ref{A3.2} may fail for Erdős--Rényi graphs where the average degree or maximum degree grows with network size. This is because the largest eigenvalue of the raw adjacency matrix, $\lambda_1^A(n)$, grows at the same order as $\max\{d_n,\sqrt{\Delta_n}\}$, where $d_n$ denotes the average degree and $\Delta_n$ the maximum degree (see \cite{krivelevich2001}). In such cases, the unscaled adjacency matrix (denoted $\mathbf{A_n}$) may not remain stable as \(n\) grows, and \((\mathbf I-\beta \mathbf A_n)^{-1}\), even when it exists, may fail to be well behaved or admit a convergent Neumann-series representation.

One possible solution to this issue is to scale the adjacency matrix to bound its spectral norm. There are, however, several ways to implement such scaling. One approach is row-normalization, as in \cite{bramoulle2009identification}, which converts peer sums into peer averages and bounds the row sums of the network operator. Another approach is exact spectral scaling, where the adjacency matrix is scaled by its largest eigenvalue. A third approach is to scale by a deterministic, observable degree-based factor that tracks the order of the spectral norm: 
\[
w_n := \max\{d_n,\sqrt{\Delta_n}\}.
\]
This choice is useful in ER graphs as \(w_n\) tracks the growth rate of the largest eigenvalue using only observable network features, thus keeping the leading eigenvalue of \(\mathbf G_n=\mathbf A_n/w_n\) stochastically bounded while preserving symmetry. See Supplemental Appendix~\ref{SA2} for the formal result and related discussion.\footnote{Relatedly, \cite{Wang2025} study the distinction between row-normalized and scaled adjacency matrices in near-degree-regular networks. They show that while row normalization can induce weak identification in such settings, appropriate scaling can mitigate these issues by aligning the rates of network regressors. Our theoretical analysis is more general, applies to networks with heterogeneous degree distributions, and explicitly leverages results from random graph theory.}

The scaled version of Assumption~\ref{A3.2} and the corresponding extension of Proposition~\ref{prop:cov_bound_unscl} are provided in Supplemental Appendix~\ref{SA1.scl}. Proposition~\ref{prop:cov_bound_scaled} shows that scaling by \(w_n\) rescales the variance and covariance of the friends-of-friends instrument at different powers of \(w_n\), aligning growth rates and preventing explosive behavior of the first-stage ratio. While weak identification can still arise in this specification, it is governed by whether the information index \(k(n)\) grows fast enough relative to the scaling factor \(w_n\).

There are two important remarks about scaling. First, scaling can be interpreted as a form of spectral regularization. It replaces the raw feedback operator \(\beta\mathbf A_n\) with \((\beta/w_n)\mathbf A_n\), shrinking the eigenvalues of the network component by \(1/w_n\). This stabilizes the network operator and can make the first-stage well-defined.\footnote{The scaled reduced-form operator is $
\left(\mathbf I-\frac{\beta}{w_n}\mathbf A_n\right)^{-1}.$
Multiplying inside by \(w_n\) gives
\[
\left(\mathbf I-\frac{\beta}{w_n}\mathbf A_n\right)^{-1}
=
w_n\left(w_n\mathbf I-\beta\mathbf A_n\right)^{-1}
=
w_n\left((\mathbf I-\beta\mathbf A_n)+(w_n-1)\mathbf I\right)^{-1}.
\] 
Thus, relative to the baseline operator \(\mathbf I-\beta\mathbf A_n\), scaling is equivalent to adding the diagonal loading \((w_n-1)\mathbf I\).} However, relative to the original specification, the scaled version estimates a regularized version of \(\beta\). Since this difference is not asymptotically negligible unless the scaling perturbation vanishes, for example when \(w_n\to1\) (unlikely to hold in most ER regimes), the scaled parameter generally differs from the original structural parameter.

Second, this implies that scaling changes the economic interpretation of the peer-effect parameter. Compared to the original specification, scaling changes both the endogenous peer exposure, from \(\mathbf A_nY\) to \(\mathbf A_nY/w_n\), and the generated network-based instruments, from approximately \(\mathbf A_n^kX\) to \(\mathbf A_n^kX/w_n^k\). In the original specification, \(\beta\) measures the marginal effect of the raw peer-outcome sum: a one-unit increase in linked peer \(j\)'s outcome changes \(Y_i\) by \(\beta A_{ij}\). In the scaled specification, the corresponding raw-link marginal effect is \((\beta/w_n)A_{ij}\). When \(w_n\) grows, holding \(\beta\) fixed therefore implies a smaller raw marginal effect of each individual peer outcome. Scaling is therefore useful only when the empirical application permits changing the modeled peer exposure from the raw network sum to the scaled network exposure.\footnote{We thank a previous anonymous referee for noting these points.}

\subsubsection{Sign reversals and near-boundary instability}\label{rem:sign_scaled}
Assumptions~\ref{A3.2} and (its scaled specification counterpart) \ref{A3.3} are sufficient conditions for $(\mathbf I-\beta\mathbf G_n)$ to be invertible and for the reduced-form to be well defined. Lemma~\ref{lem:eigen_cov} in Supplemental Appendix Section~\ref{SA1} shows, however, that invertibility alone does not guarantee stable or well-behaved first-stage relationships. In particular, conditional on $\mathbf G_n$, the population first-stage covariance admits the decomposition
\[
\frac{1}{n}Cov(\mathbf G_n^{(2)}\mathbf X,\mathbf G_n\mathbf Y\mid \mathbf G_n)
=
\frac{\sigma_x^2}{n}\sum_{j=1}^n
\frac{\lambda_j^3(n)\big(\gamma+\delta\lambda_j(n)\big)}{1-\beta\lambda_j(n)}
\;+\;
\frac{\sigma_x^2}{n}\,R_{n,\mathrm{diag}},
\]
where the leading term aggregates the contributions of the spectral components of the network operator and the remainder term arises from the diagonal adjustment in $\mathbf G_n^{(2)}$. Moreover, the remainder satisfies the deterministic bound
\[
|R_{n,\mathrm{diag}}|
\;\le\;
\Tr(D_n)\,
\|\mathbf G_n\|_2\,
\|(\mathbf I_n-\beta\mathbf G_n)^{-1}\|_2\,
\|\gamma\mathbf I_n+\delta\mathbf G_n\|_2,
\]
and, therefore, does not introduce additional amplification through factors of $(1-\beta\lambda_j(n))^{-1}$.\footnote{The magnitude of $R_{n,\mathrm{diag}}$ depends on $\Tr(\mathbf D_n)$ as long as $(I_n - \beta \mathbf{G}_n)^{-1}$ is invertible and bounded. In the scaled version, taking expectations and using $\deg(i)\sim\mathrm{Bin}(n-1,p_n)$ gives
\[
\mathbb E [ \Tr(\mathbf D_n) ]
=\frac{n}{w_n^2}\,\mathbb E[\deg(i)]
=\frac{n d_n}{w_n^2}.
\]
Hence, by Markov's inequality, \(\frac{\Tr(\mathbf D_n)}{n}
=O_p\!\Big(\frac{d_n}{w_n^2}\Big). \)
If, in addition, $\|(\mathbf I_n-\beta\mathbf G_n)^{-1}\|_2=O_p(1)$ and
$\|\mathbf G_n\|_2=O_p(1)$, we have that
\[
\frac{1}{n}|R_{n,\mathrm{diag}}|
\le
\frac{\Tr(\mathbf D_n)}{n}\,
\|\mathbf G_n\|_2\,
\|(\mathbf I_n-\beta\mathbf G_n)^{-1}\|_2\,
\|\gamma\mathbf I_n+\delta\mathbf G_n\|_2
=
o_p(1),
\]
so the diagonal correction is asymptotically negligible and cannot affect the sign of the first-stage covariance provided the leading spectral term is bounded away from zero. Here $w_n=\max\{d_n,\sqrt{\Delta_n}\}$ with $\Delta_n$ the maximum degree. A similar conclusion holds in the unscaled case when the expected degree is uniformly bounded, $d_n=O(1)$.}
Two important phenomena follow directly from this spectral decomposition.
\begin{enumerate}
\item \textbf{Sign of the first-stage covariance.}
Lemma~\ref{lem:eigen_cov} shows that the population covariance between the peers-of-peers instrument $\mathbf G_n^{(2)}\mathbf X$ and the endogenous regressor $\mathbf G_n\mathbf Y$ decomposes into a sum of terms indexed by the eigenvalues of $\mathbf G_n$. As long as $(\mathbf I-\beta\mathbf G_n)$ is invertible, each term is proportional to
\[
\frac{\lambda_j(n)^3\big(\gamma+\delta\lambda_j(n)\big)}{1-\beta\lambda_j(n)}.
\]
Thus, the sign of each term is governed by the sign of $1-\beta\lambda_j(n)$ together with the sign of $(\gamma+\delta\lambda_j(n))$. When $\beta\lambda_j(n)$ exceeds unity for eigenvalues that contribute most to equilibrium variation, higher-order network feedback enters with the opposite sign. In this case, the population first-stage covariance may switch signs as network density increases, even when all primitive peer effects are positive.

\item \textbf{Near-boundary instability and weak identification.}
The same decomposition also explains why identification can deteriorate even in the scaled specification. Lemma~\ref{lem:eigen_cov} shows that each spectral contribution is amplified by a factor proportional to $|1-\beta\lambda_j(n)|^{-1}$. When the largest eigenvalue satisfies $\beta\lambda_1^G(n)\approx 1$, the matrix $(\mathbf I-\beta\mathbf G_n)^{-1}$ becomes poorly conditioned. In this region, small changes in the network or in higher-order components of network structure can lead to large changes in the  covariance from the first-stage. As a result, peers-of-peers instruments may be weak -- not because the reduced form ceases to exist -- but because equilibrium feedback makes the first-stage highly sensitive near the stability boundary. Scaling by $w_n =\max\{d_n, \sqrt{\Delta_n}\}$ ensures that $\lambda_1^G(n)$ is asymptotically close to unity. Hence, for $\beta$ close to 1, the inverse can become ill-conditioned in finite samples and the instruments would be weak even in the scaled model.
\end{enumerate}

\subsubsection{Numerical Illustrations of the Bounds}

To illustrate the previous discussions, in Appendix \ref{SA3_figures}, we report simulated upper bounds for the two specifications: baseline unscaled and the scaled model from Propositions~\ref{prop:cov_bound_unscl} and ~\ref{prop:cov_bound_scaled}. For each $n\in\{200,400,800,1600\}$ we generate $500$ Monte Carlo draws under several Erd\H{o}s--R\'enyi regimes. 

With the unscaled model (Figure \ref{fig:upper_non-norm}), we see that, in the dense regime where $d_n\to\infty$, the bound decreases toward zero, consistent with $\mathbf G_n^{(2)}$ becoming nearly collinear with $\mathbf G_n$, so that the instrument adds little independent variation. The regimes with constant and log-log average degree also have upper bounds close to zero, while in the extremely sparse regime ($d_n=o(1)$), the bound increases with \(n\). The latter is consistent with the  instability of the population first-stage estimand. Meanwhile, in the scaled specification  shown in Figure \ref{fig:upper_norm}, where $\mathbf G_n=\mathbf A_n/w_n$, the average upper bound remains well behaved across all the regimes: it does not diverge, and it does not fall quickly. 

\subsection{Weak-IV Robust Testing}

Given that multiple regimes in Propositions \ref{prop:cov_bound_unscl} and \ref{prop:cov_bound_scaled} have identification failures at the limit, the instruments based on peers-of-peers are likely to be weak empirically for those configurations, especially for the unscaled model. For these settings, we adapt inference robust to a weak first-stage.

In the classical IV model, we could proceed by either implementing the Anderson-Rubin (AR) test from \cite{anderson1949estimation} (known to be unbiased and asymptotically efficient in the just-identified case - see \cite{moreira2009tests}) or the Conditional Likelihood Ratio test from \cite{moreira03}. 
The AR-statistic is given by: 
\begin{align}\label{ar_teststat}
    AR(\beta) = n\hat{g}(\beta)'\Omega(\beta)^{-1}\hat{g}(\beta) \stackrel{H_0}{\to}_d \chi^2_k,
\end{align}
where $\sqrt{n}\hat{g} (\beta_0) = \sqrt{n}(\hat{\xi} - \beta_0 \hat{\pi}) \stackrel{H_0}{\to}_d N(0,\Omega(\beta_0))$ with an appropriate variance estimator for 
\begin{align}\label{var_arstat}
    \Omega(\beta_0) = Var_{\hat{\xi}} - \beta_0 (Cov_{\hat{\xi} \hat{\pi}} + Cov_{\hat{\pi} \hat{\xi}}) + \beta^2_0 Var_{\hat{\pi}}
\end{align}
Equation \eqref{ar_teststat} holds regardless of the strength of the instrument and, thus, the AR test is given as $\phi^{AR}_{n}(\alpha) = \mathbf{1}\{AR(\beta_0) > \chi^2_{k,1-\alpha} \}$ for the null hypothesis  $H_0:\beta = \beta_0$. The confidence set of the test can take different forms including the extreme case where it is the entire real line when $\pi=0$. This is because $\beta$ is not identified in that case and any value of $\beta$ satisfies the restriction condition (the test will have zero power in this case).

However, there is one main distinction of our system \eqref{eq:first_stage}-\eqref{eq:struct_eq2} relative to the classical set-up: the heteroskedasticity induced by network-dependency of $\varepsilon$. This is the term $(\mathbf{I} - \beta \mathbf{G})^{-1} \mathbf{\varepsilon}$ in \eqref{red_model}, implying errors of the form,
\begin{align}\label{spill_errors}
    (\mathbf{I} - \beta \mathbf{G} )^{-1} \varepsilon = \sum_{k=0}^{\infty}\beta^k\mathbf{G}^{k} \varepsilon
\end{align}
causing them to be correlated across the connections even if $\varepsilon$ is assumed to be homoskedastic. We differ from \cite{ross2022}, who also implement weak-IV robust inference in a network setting, by explicitly accounting for network-dependence induced heteroskedasticity.

We provide two possible approaches to inference robust to weak instruments. First, to fully deal with the cross-sectional dependence arising from network spillovers, we propose the use of the variance estimator in \cite{kojevnikov2021limit} (Proposition 4.3). As that variance estimator is consistent for $\Omega(\beta)$ given $\beta$, an application of Slutsky's Lemma guarantees the applicability of the feasible AR-test. This is summarized in the proposition below.\footnote{Other papers, such as \cite{acemoglu2015}, model the cross-sectional dependence as spatially correlated data and use \cite{conley1999gmm} as a consistent estimator for $\Omega(\beta)$. Alternatively, it is common to use clustered variance estimators. The latter requires a block structure (e.g., independence across villages, schools, families) and an asymptotic theory based on "many" networks.}

\begin{prop}\label{prop:Asymp_Dist}
    Consider the null of $H_0: \beta = \beta_0$ and let $ \tilde{V}_{\pi},\tilde{V}_{\xi}$ be the \cite{kojevnikov2021limit} variance estimators for the OLS coefficients of $\mathbf{G}^2 \mathbf{X}$ in equations \eqref{eq:first_stage} and \eqref{eq:struct_eq2}  respectively. Then, under the regularity conditions of Proposition 4.1 in \cite{kojevnikov2021limit},
\begin{eqnarray}\label{eq:var_prop1}
    \hat{\Omega}(\beta) = \tilde{V}_{\hat{\xi}} - \beta (\tilde{Cov}_{\hat{\xi} \hat{\pi}} + \tilde{Cov}_{\hat{\pi} \hat{\xi}}) + \beta^2 \tilde{V}_{\hat{\pi}} \to_p \Omega(\beta).
\end{eqnarray}
Furthermore,
\begin{align}\label{Asymp_Dist_ar_test}
    \tilde{AR}_n(\beta_0) = n\hat{g}(\beta_0)'\hat{\Omega}(\beta_0)^{-1}\hat{g}(\beta_0) \stackrel{H_0}{\to}_d \chi^2_k
\end{align}
\end{prop}




However, we note that even the \textit{homoskedastic} implementation of the Anderson-Rubin test and, thus, the Conditional Likelihood Ratio test (CLR) perform well asymptotically. This is because the higher-order terms in \eqref{spill_errors} are likely to be negligible when the network is sparse and converging to 0 and $\beta$ is small. Indeed, we have that
\begin{eqnarray}
 Var\left((\mathbf{I} - \beta \mathbf{G} )^{-1} \varepsilon\right) &=& Var\left(\sum_{k=0}^{\infty}\beta^k\mathbf{G}^{k} \varepsilon\right) \notag\\
 &\to & Var(\varepsilon),
\end{eqnarray}
as $ \beta^k \mathbf{G}^k \to \mathbf{0}$.

Thus, asymptotic inference based on homoskedastic errors (e.g., \cite{moreira03}) is likely to perform very well under weak instruments induced by network sparsity, since network-induced dependence becomes negligible in such cases. 

\section{Monte Carlo Simulations}\label{S4}

We now provide Monte Carlo simulations to illustrate the finite-sample properties of our theoretical results. We base our data-generating process on those used in \cite{bramoulle2009identification}, but with alternate network structures that showcase the issue of weak identification with peers-of-peers instruments for both the scaled and unscaled specifications.

We consider Erdős--Rényi random graphs (\cite{erdds1959random}) with $\textit{d = np}$ being the average degree. Following \cite{bramoulle2009identification}, we draw a uni-dimensional $X_i$ with approximately $5\%$ of values to be 0\footnote{Using a Bernoulli(0.9458333) as in \cite{bramoulle2009identification}.} and the remaining $95\%$ follow an i.i.d. log-normal ($X_i \stackrel{i.i.d.}{\sim} LogNormal(1,3)$). These are fixed for a given data size and average degree. We draw the error terms $\varepsilon_i \stackrel{i.i.d.}{\sim} Normal(0,1)$, independently for each simulation run. The true coefficients are set at $\alpha = 0.7683, \gamma = 0.0834, \delta = 0.1507$, and $\beta \in \{0.4666, 0.95\}$ to evaluate the impact of a change in intensity of the peer effect. 

We perform two sets of exercises. First, we compute first-stage diagnostics across the degree-growth regimes considered in Propositions~\ref{prop:cov_bound_unscl} and~\ref{prop:cov_bound_scaled}. These include the average first-stage \(F\)-statistic, the sample covariance between the endogenous regressor \(\mathbf G\mathbf Y\) and the instrument \(\mathbf G^{(2)}\mathbf X\), and the variance of the instrument. These are reported in Figures~\ref{fig:Fstat}-\ref{fig:Var} in Supplemental Appendix Section \ref{SA3_figures}. Then, we consider the performance of the TSLS, with the instrument $\mathbf{G}^{(2)} X$ for the endogenous variable $\mathbf{G}Y$, and the exogenous variables as instruments for themselves, under the inference procedures developed in the previous section. We run simulations by varying \textit{d} across settings, taking values from the set $\{0.25,0.5,0.75,1,2,5\}$ with sample size \(n\) varying in $\{250, 500, 1000, 2000\}$.\footnote{Notice that the smaller the value of \textit{d}, the sparser the network.} For every setting, we fix \textit{d} and \textit{n} and set the Monte Carlo Simulation number to 1000, keeping the network structure fixed across all runs. We estimate all parameters for two specifications: the unscaled setting where \(\mathbf G_n=\mathbf A_n\), and the scaled setting where \(\mathbf G_n=\mathbf A_n/w_n\) with \(w_n=\max\{d_n,\sqrt{\Delta_n}\}\). The main text reports the unscaled results in Table~\ref{tab:unscaled_diagnostics_coverage}. Appendix Section~\ref{SA3} reports the corresponding scaled results and additional diagnostics, including confidence-interval lengths.


Figure \ref{fig:Fstat} shows that for the  baseline model, first-stage $F$-statistic remains flat and close to zero across all regimes except the extremely sparse case. In the extremely sparse regime, the $F$-statistic increases with sample size. This is because the instrument variance goes to 0, rather than genuine accumulation of identifying information. In contrast, for the scaled specification, the first-stage $F$-statistic increases monotonically with \(n\) across regimes. Figures~\ref{fig:Cov} and \ref{fig:Var} further support this analysis, showcasing the covariance between the endogenous regressor and the instrument, and the variance of the instrument, respectively. In the unscaled case, the covariance between the endogenous regressor and the instrument remains close to zero in most regimes, while the variance of the instrument is large except in the extremely sparse regime. By contrast, in the scaled specification, both the covariance and the variance of the instrument remain stable across sample sizes once the extremely sparse regime is excluded. 


\afterpage{%
\begin{landscape}
\begin{table}[p]
\centering
\scriptsize
\setlength{\tabcolsep}{2.5pt}
\renewcommand{\arraystretch}{1.05}
\caption{Unscaled Model: TSLS Estimates, First-Stage Diagnostics, and Empirical Coverage Probabilities. Panel A reports TSLS estimates for specification~\eqref{struc_model} under the unscaled network, together with the correlation between \(\mathbf{GY}\) and \(\mathbf G^{(2)}\mathbf X\), and the first-stage \(F\)-statistic. Panel B reports coverage probabilities for conventional \(t\)-test using homoskedastic standard errors,  \(t\)-test with \cite{kojevnikov2021limit} network-dependent standard errors, and Anderson--Rubin inference using homoskedastic standard errors.}
\label{tab:unscaled_diagnostics_coverage}

\begin{adjustbox}{width=\linewidth,totalheight=0.88\textheight,keepaspectratio}
\begin{tabular}{ll|llllll|llllll|llllll}
\toprule
\multicolumn{20}{l}{\textbf{Panel A: TSLS Estimates and First-Stage Diagnostics}} \\
\midrule
 & & \multicolumn{6}{c|}{\textbf{TSLS Estimate}} 
 & \multicolumn{6}{c|}{\textbf{Correlation ($\mathbf{GY}$, $\mathbf{G}^2\mathbf{X}$)}} 
 & \multicolumn{6}{c}{\textbf{First-Stage F-statistic}} \\
\midrule
$\beta_0$ & $n \backslash d$ 
& 0.25 & 0.5 & 0.75 & 1 & 2 & 5 
& 0.25 & 0.5 & 0.75 & 1 & 2 & 5 
& 0.25 & 0.5 & 0.75 & 1 & 2 & 5 \\
\midrule

\multirow{4}{*}{0.4666} 
& 250  & 0.460 & 0.555 & 0.469 & 0.468 & 0.467 & 0.501 & 0.343 & -0.104 & -0.304 & -0.153 & -0.128 & -0.064 & 28.401 & 3.711 & 31.111 & 6.337 & 2.055 & 1.035 \\
& 500  & 0.444 & 0.570 & 0.466 & 0.467 & 0.466 & 0.446 & 0.237 & -0.065 & 0.241 & -0.074 & -0.083 & -0.055 & 21.844 & 6.132 & 21.376 & 1.962 & 2.170 & 1.448 \\
& 1000 & 0.463 & 0.467 & 0.467 & 0.467 & 0.478 & 0.466 & 0.458 & 0.155 & -0.275 & -0.186 & 0.017 & -0.069 & 252.452 & 25.033 & 72.175 & 29.988 & 0.627 & 4.882 \\
& 2000 & 0.465 & 0.466 & 0.459 & 0.476 & 0.467 & 0.462 & 0.537 & 0.226 & 0.039 & -0.013 & -0.016 & -0.008 & 819.442 & 88.291 & 2.197 & 1.671 & 2.164 & 0.178 \\

\noalign{\medskip}

\multirow{4}{*}{0.95} 
& 250  & 0.953 & 0.950 & 0.948 & 0.951 & 0.825 & 0.928 & -0.079 & -0.181 & -0.236 & -0.178 & -0.124 & -0.008 & 4.340 & 12.939 & 18.718 & 7.080 & 3.496 & 0.184 \\
& 500  & 0.952 & 0.951 & 0.952 & 0.951 & 0.905 & 0.967 & -0.082 & -0.125 & -0.188 & -0.179 & -0.057 & 0.001 & 11.345 & 8.463 & 16.339 & 14.451 & 1.527 & 0.887 \\
& 1000 & 0.951 & 0.950 & 0.950 & 0.950 & 0.852 & -4.806 & -0.093 & -0.119 & -0.190 & -0.182 & -0.018 & -0.027 & 16.462 & 15.505 & 46.086 & 31.921 & 1.441 & 0.698 \\
& 2000 & 0.951 & 0.950 & 0.950 & 0.951 & 0.969 & 0.931 & -0.104 & -0.157 & -0.176 & -0.107 & -0.003 & 0.004 & 45.5 & 87.2 & 70.0 & 26.8 & 0.23 & 0.70 \\

\midrule
\multicolumn{20}{l}{\textbf{Panel B: Coverage Probabilities}} \\
\midrule

 & & \multicolumn{6}{c|}{\textbf{\(t\)-test (Homoskedastic SE)}} 
 & \multicolumn{6}{c|}{\textbf{\(t\)-test (\cite{kojevnikov2021limit} SE)}} 
 & \multicolumn{6}{c}{\textbf{AR (Homoskedastic SE)}} \\
\midrule
$\beta_0$ & $n \backslash d$ 
& 0.25 & 0.5 & 0.75 & 1 & 2 & 5 
& 0.25 & 0.5 & 0.75 & 1 & 2 & 5
& 0.25 & 0.5 & 0.75 & 1 & 2 & 5 \\
\midrule

\multirow{4}{*}{0.4666}
& 250  & 0.959 & 0.996 & 0.961 & 0.995 & 0.999 & 1.000 
        & 0.952 & 1.000 & 0.973 & 0.997 & 0.999 & 1.000
        & 0.943 & 0.944 & 0.942 & 0.941 & 0.960 & 0.949 \\
& 500  & 0.979 & 0.988 & 0.973 & 1.000 & 1.000 & 1.000
        & 0.978 & 0.998 & 1.000 & 1.000 & 1.000 & 1.000
        & 0.948 & 0.924 & 0.944 & 0.956 & 0.942 & 0.946 \\
& 1000 & 0.950 & 0.971 & 0.965 & 0.967 & 1.000 & 0.995
        & 0.950 & 0.971 & 0.965 & 0.967 & 1.000 & 0.995
        & 0.948 & 0.949 & 0.959 & 0.953 & 0.947 & 0.949 \\
& 2000 & 0.958 & 0.963 & 0.999 & 1.000 & 1.000 & 1.000
        & 0.947 & 0.994 & 1.000 & 1.000 & 1.000 & 1.000
        & 0.956 & 0.959 & 0.952 & 0.950 & 0.940 & 0.957 \\

\noalign{\medskip}

\multirow{4}{*}{0.95}
& 250  & 0.999 & 0.982 & 0.967 & 0.989 & 1.000 & 1.000
        & 0.994 & 0.960 & 0.946 & 0.981 & 0.999 & 1.000
        & 0.953 & 0.953 & 0.947 & 0.961 & 0.953 & 0.955 \\
& 500  & 0.975 & 0.991 & 0.972 & 0.981 & 1.000 & 1.000
        & 0.954 & 0.957 & 0.948 & 0.975 & 1.000 & 1.000
        & 0.940 & 0.958 & 0.948 & 0.951 & 0.950 & 0.952 \\
& 1000 & 0.977 & 0.973 & 0.969 & 0.965 & 0.998 & 1.000
        & 0.953 & 0.954 & 0.962 & 0.959 & 0.998 & 1.000
        & 0.950 & 0.939 & 0.960 & 0.951 & 0.945 & 0.963 \\
& 2000 & 0.968 & 0.956 & 0.949 & 0.976 & 1.000 & 1.000
        & 0.959 & 0.950 & 0.949 & 0.985 & 1.000 & 1.000
        & 0.962 & 0.953 & 0.945 & 0.953 & 0.960 & 0.958 \\

\bottomrule
\end{tabular}
\end{adjustbox}

\end{table}
\end{landscape}
}

Table~\ref{tab:unscaled_diagnostics_coverage} reports the TSLS estimates, first-stage diagnostics (Panel A) and the coverage probabilities (Panel B) for the baseline specification. 
Several features are worth highlighting. When \(d<1\), the first-stage diagnostics can look deceptively strong, especially as \(n\) grows. This is consistent with our theoretical result that the first-stage estimand may become ill-defined in extremely sparse networks. As shown in Table~\ref{tab:covariance} in the Appendix, the covariance between $\mathbf G^{(2)}\mathbf X$ and $\mathbf G\mathbf Y$ is in fact extremely small in this regime, indicating that the large correlations are mechanically driven by vanishing instrument variance rather than genuine identifying power. Furthermore, as the average degree increases beyond one, we observe very low correlations and weak first-stage $F$-statistics across all sample sizes. This pattern persists as \(n\) grows, suggesting that this is not a small-sample artifact.

Panel B reports empirical coverage probabilities at the 95\% nominal level for three inference procedures: the conventional $t$-test with homoskedastic standard errors, the $t$-test using network-dependent standard errors  (\cite{kojevnikov2021limit}), and the Anderson--Rubin (AR) test with homoskedastic errors.\footnote{We also implement the AR test with network-dependent standard errors (reported in Table~\ref{tab:ar-kojev-cov} in Appendix Section \ref{SA3_tables}). Given the homoskedastic data generating process, the resulting coverage rates are very similar for larger samples.} We find that the $t$-test fails to provide nominal coverage for $d<1$. In these regimes, coverage probabilities are often close to 100\%, reflecting extremely wide confidence intervals (see Tables \ref{tab:ar-kojev-cov}-\ref{tab:ci_length_ar} in Appendix Section \ref{SA3_tables}) —symptomatic of weak instruments and identification failure. This is consistent with the discussion in Section~\ref{S3} and Panel A of Table \ref{tab:unscaled_diagnostics_coverage}. Even in the extremely sparse regime, where coverage improves toward the nominal 95\% level as \(n\) increases, we continue to observe systematic over-coverage. While the use of network-dependent standard errors mitigates some of these distortions,  the resulting confidence intervals remain conservative in many cases. This suggests that the problem cannot be resolved by variance correction alone.

To better understand these distortions, Figures~\ref{fig:cov_beta_un} and~\ref{fig:cov_t-stat_un} in Appendix Section~\ref{SA3_figures} display the empirical distributions of the TSLS estimator \(\hat{\beta}\) and the associated \(t\)-statistic computed using network-dependent standard errors. The figures compare two unscaled cases: a relatively strong first stage (\(d=0.5\)) and a weak first stage (\(d=2\)). In both cases, the sampling distribution of \(\hat{\beta}\) remains centered near the true parameter, but in the weak case it becomes heavy-tailed, while the associated \(t\)-statistics display clear departures from normality and multi-modality (e.g., see \cite{staigerstock97, andrews2019weak} for discussions in the weak-instruments case). These patterns reflect ill-conditioned first-stage inversions in the TSLS variance formula and explain the observed over-coverage of conventional confidence intervals.


We also report the corresponding results for the scaled specification in Table \ref{tab:scaled_diagnostics_coverage} (Appendix Section~\ref{SA3_tables}). As expected, the scaled model generally displays stronger first-stage diagnostics, especially when \(\beta=0.4666\). However, when \(\beta=0.95\), the first stage can still deteriorate in finite samples, consistent with the near-boundary phenomenon discussed in Section~\ref{rem:sign_scaled}. Thus, scaling stabilizes the network operator but does not rule out weak first stages in the scaled specification. As the average degree increases further away from one, the correlation moves away from zero.

\section{Empirical Application of Weak-IV Robust Inference with Network-Based Instruments}\label{S5}


\cite{Konig2017} investigate strategic complementarities in the use of violence within a network of armed groups during the Second Congo War and examine how networks of alliances and hostilities influence the intensity of conflict. Nodes correspond to armed actors, while links encode military relationships, distinguishing between alliances (groups fighting on the same side), enmities (groups that directly clash), and neutrality (groups that are neither allies nor enemies). They model conflict as a network game, and derive the corresponding Nash equilibrium with the optimal level of fighting intensity depending on the fighting of its allies and enemies. Endogeneity of these network-sum regressors is addressed using network-based instrumental variables constructed from exogenous weather shocks. In particular, rainfall in the homeland of linked groups is used as an excluded instrument, and—following \cite{bramoulle2009identification}—the authors explicitly exploit second-degree instruments based on the rainfall of neighbors-of-neighbors. 

We focus on the main empirical specification reported in Table~1 of \cite{Konig2017}, which examines how a group’s own fighting effort responds to the fighting efforts of its network neighbors. The empirical model is over-identified and features three endogenous regressors: total fighting effort of allies (TFA), of enemies (TFE), and of neutral groups (TFN). These are instrumented using rainfall shocks in higher-order network neighborhoods, corresponding to instruments of the form \(\mathbf G^{(2)}\mathbf X\) in our framework.

First, we connect the sparsity documented in Table \ref{tab:deg_sum} to first-stage strength. Figure~\ref{fig:par-var-cor-konig} reports a heatmap of the variance-normalized \emph{partial} covariance between the endogenous network-sum regressors and the higher-order rainfall instruments.\footnote{For expositional purposes, the partial covariances and variances are computed without conditioning on additional exogenous controls that may enter the first-stage regressions.} Across most instruments, the resulting ratios are of the order $10^{-3}$ to $10^{-1}$, indicating very small partial first-stage coefficients. Together with the sparsity of the underlying ally and enemy networks (e.g., Figure \ref{fig:konig_graph_allies} and Table \ref{tab:deg_sum}), these findings are consistent with a weak-identification environment in the sense of Definition \ref{def:weak_network}.

\begin{figure}[!h]
    \centering
    \includegraphics[width=\textwidth]{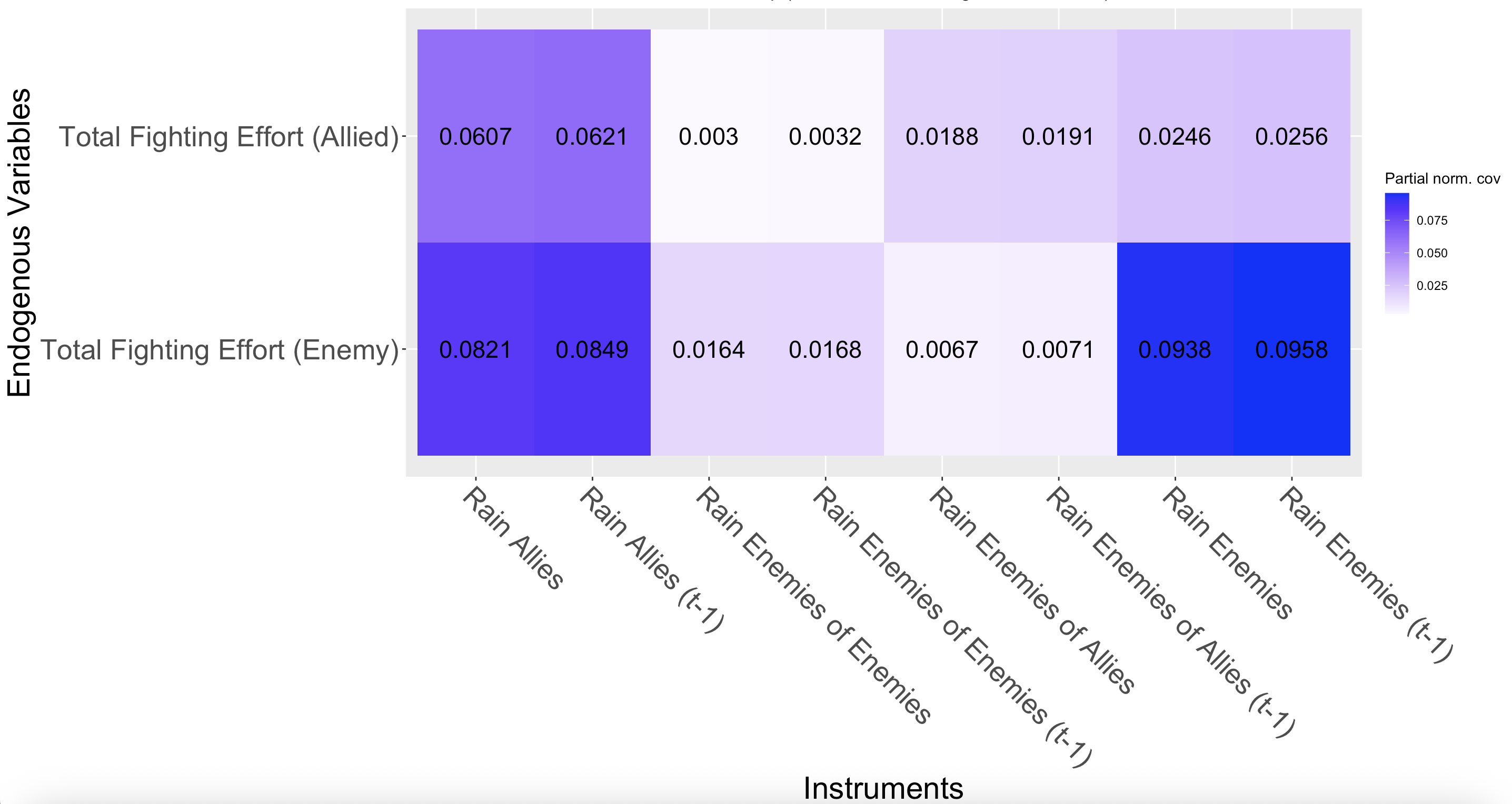}
    \caption{Variance Normalized Partial-Covariance between the two Endogenous Variables with instruments (\cite{Konig2017})}
    \label{fig:par-var-cor-konig}
\end{figure}

Second, we look at the three most relevant specifications from Table~1. These are reproduced in Table~\ref{tab:konig_iv_results}, where we report point estimates, 95\% non-robust confidence intervals based on clustered $t$-tests, and projected 95\% Anderson–Rubin and conditional likelihood ratio confidence intervals.\footnote{The authors use a custom spatial TSLS estimator in \texttt{Stata} to account for spatial correlation. We instead use clustered standard errors via the \texttt{ivreg2} command, which is also used in their replication files and is compatible with the \texttt{weakiv} package of \cite{weakiv2013}. As a result, reported standard errors may differ slightly from those in \cite{Konig2017}, while point estimates remain identical.} We compute standard errors based on the homoskedastic case, given the sparsity of the underlying network, its simplicity in implementation, and its performance in this setting as seen in the previous section.

\begin{table}[ht]
\centering
\caption{IV Estimates and Confidence Intervals for \cite{Konig2017}. We present original estimates from their Table 1, Columns 2-4, together with our adapted Anderson-Rubin and Conditional Likelihood Ratio Confidence Intervals.}
\label{tab:konig_iv_results}
\resizebox{\textwidth}{!}{%
\begin{tabular}{|l c c c c|}
\hline
\multicolumn{5}{|l|}{}\\
\textbf{Endogenous Variable} & \textbf{Point Estimate} & \textbf{Original CI} & \textbf{AR CI} & \textbf{CLR CI} \\
\multicolumn{5}{|l|}{}\\
\hline

\multicolumn{5}{|l|}{}\\
\multicolumn{5}{|l|}{\textbf{Table 1, Column 2 -- Reduced-Form IV Specification}} \\
\multicolumn{5}{|l|}{}\\
Enemies (TFE) & 0.130 & [0.023, 0.240] & [-0.049, +$\infty$) & [-0.067, +$\infty$) \\
Allies (TFA)  & -0.218 & [-0.377, -0.058] & [-0.431, +$\infty$) & [-0.457, +$\infty$) \\
\multicolumn{5}{|l|}{}\\

\multicolumn{5}{|l|}{\textbf{Table 1, Column 3 -- IV Specification}} \\
\multicolumn{5}{|l|}{}\\
Enemies (TFE) & 0.066 & [0.025, 0.106] & (-$\infty$, +$\infty$) & (-$\infty$, 0.140] \\
Allies (TFA)  & -0.117 & [-0.204, -0.029] & [-0.234, +$\infty$) & [-0.146, +$\infty$) \\
\multicolumn{5}{|l|}{}\\

\multicolumn{5}{|l|}{\textbf{Table 1, Column 4 -- IV Specification}} \\
\multicolumn{5}{|l|}{}\\
Enemies (TFE)  & 0.083 & [0.041, 0.125] & (-$\infty$, +$\infty$) & (-$\infty$, +$\infty$) \\
Allies (TFA)   & -0.114 & [-0.198, -0.030] & (-$\infty$, +$\infty$) & (-$\infty$, +$\infty$) \\
Neutrals (TFN) & 0.004 & [-0.005, 0.013] & [-0.011, 0.011] & (-$\infty$, +$\infty$) \\
\hline
\end{tabular}%
}
\end{table}

The non-robust confidence intervals reported in Table~1 of \cite{Konig2017} would suggest statistically significant effects of allies’ and enemies’ fighting efforts. 
However, once weak-IV–robust inference is applied, this conclusion is weakened. Across the same specifications, the Anderson–Rubin and conditional likelihood ratio confidence intervals not only include zero but—with the exception of a single case—are unbounded for all three endogenous variables. This pattern is characteristic of a weak-instrument environment in which the covariance between the instruments and the endogenous regressors is small relative to sampling variability.

The source of this weakness is the sparse structure of the underlying alliance–enmity networks. As shown in Table~\ref{tab:deg_sum}, Figures~\ref{fig:konig_graph_allies}, and Figure~\ref{fig:par-var-cor-konig}, the underlying networks are sparse and the resulting higher-order rainfall instruments provide little independent first-stage variation. Consequently, these instruments are only weakly correlated with the corresponding network-sum regressors. Together, these findings illustrate how standard inference can be misleading in network settings: non-robust confidence intervals mask weak identification due to sparsity and limited higher-order neighborhoods, while weak-IV–robust methods remain valid under weak identification.

\section{Conclusion}\label{S6}

The growth in measuring peer-effects in academia and policy should also bring renewed attention to the challenges in inference. In this paper, we showed how ill-defined first-stage estimands and/or weak instruments can arise naturally in the linear-in-means model with network-based instruments due to a specific mechanism related to network topology. We characterized the conditions using Erdős--Rényi graphs, a standard network formation model, and illustrated their impact on inference. We presented an alternative scaled specification that can help mitigate some identification challenges when the empirical application supports the corresponding reinterpretation of the peer-effect parameter. In cases where weakness persists, we adapt results from standard weak-IV robust testing literature to the network-based instruments setting and combine it with a consistent network-dependent variance estimator. Such concerns are relevant and examples of sparse networks span many fields, including salient cases in Economics (see Table \ref{tab:deg_sum}). 
While our characterization results focus on the linear-in-means model, our main intuition extends to any linear regression framework with network-based instruments. Indeed, the adapted weak instruments definition \ref{def:weak_network} for network-based IVs does not rely on that specific model. Furthermore, issues with first-stage estimands and lack of information for higher-order networks would be prevalent in all such cases. Thus, we deem that such concerns are warranted in many more applications. Finally, it is likely that such issues extend to non-linear models of network interactions, such as discrete choice models with network-based instruments (e.g., \cite{volpe2025discrete}).
Future work should investigate such settings.



\bibliography{report_bib}

@article{battaglini2018,
author = {Battaglini, Marco and Patacchini, Eleonora},
title = {Influencing Connected Legislators},
journal = {Journal of Political Economy},
volume = {126},
number = {6},
pages = {2277-2322},
year = {2018},
doi = {10.1086/700074},
}

@article{calvo2009peer,
  title={Peer effects and social networks in education},
  author={Calv{\'o}-Armengol, Antoni and Patacchini, Eleonora and Zenou, Yves},
  journal={The review of economic studies},
  volume={76},
  number={4},
  pages={1239--1267},
  year={2009},
  publisher={Wiley-Blackwell}
}

@article{gibbons,
  title={Mostly pointless spatial econometrics?},
  author={Gibbons, Stephen and Overman, Henry G},
  journal={Journal of regional Science},
  volume={52},
  number={2},
  pages={172--191},
  year={2012},
  publisher={Wiley Online Library}
}

@article{bramoulle2020peer,
  title={Peer effects in networks: A survey},
  author={Bramoull{\'e}, Yann and Djebbari, Habiba and Fortin, Bernard},
  journal={Annual Review of Economics},
  volume={12},
  number={1},
  pages={603--629},
  year={2020},
  publisher={Annual Reviews}
}

@Article{Konig2017,
  author={Michael D. König and Dominic Rohner and Mathias Thoenig and Fabrizio Zilibotti},
  title={{Networks in Conflict: Theory and Evidence From the Great War of Africa}},
  journal={Econometrica},
  year=2017,
  volume={85},
  number={},
  pages={1093-1132},
  doi={},
  url={https://ideas.repec.org/a/wly/emetrp/v85y2017ip1093-1132.html}
}

@article{Cruz2017,
Author = {Cruz, Cesi and Labonne, Julien and Querubín, Pablo},
Title = {Politician Family Networks and Electoral Outcomes: Evidence from the Philippines},
Journal = {American Economic Review},
Volume = {107},
Number = {10},
Year = {2017},
Pages = {3006-37},
DOI = {10.1257/aer.20150343},
URL = {https://www.aeaweb.org/articles?id=10.1257/aer.20150343}}

@article{bramoulle2009identification,
  title={Identification of peer effects through social networks},
  author={Bramoulle, Yann and Djebbari, Habiba and Morton, Florian},
  journal={Econometrica},
  volume={77},
  number={4},
  pages={1467-1514},
  year={2009},
}

@article{dufour97,
  title={Some impossibility theorems in econometrics with applications to structural and dynamic models},
  author={Dufour, Jean-Marie},
  journal={Econometrica: Journal of the Econometric Society},
  pages={1365--1387},
  year={1997},
  publisher={JSTOR}
}

@article{staigerstock97,
  title={Instrumental Variables Regression with Weak Instruments},
  author={Staiger, Douglas and Stock, James H},
  journal={Econometrica},
  volume={65},
  number={3},
  pages={557--586},
  year={1997},
  publisher={Econometric Society}
}

@article{andrews2019weak,
  title={Weak instruments in instrumental variables regression: theory and practice},
  author={Andrews, Isaiah and Stock, James H and Sun, Liyang},
  journal={Annual Review of Economics},
  volume={11},
  number={1},
  pages={727--753},
  year={2019}
}

@article{manski1993identification,
  title={Identification of endogenous social effects: The reflection problem},
  author={Manski, Charles F},
  journal={The review of economic studies},
  volume={60},
  number={3},
  pages={531--542},
  year={1993},
  publisher={Wiley-Blackwell}
}

@article{anderson1949estimation,
  title={Estimation of the parameters of a single equation in a complete system of stochastic equations},
  author={Anderson, Theodore W and Rubin, Herman},
  journal={The Annals of mathematical statistics},
  volume={20},
  number={1},
  pages={46--63},
  year={1949},
  publisher={Institute of Mathematical Statistics}
}

@misc{stock2005testing,
  title={Testing for Weak Instruments in Linear IV Regression},
  author={Stock, James H. and Yogo, Motohiro},
  booktitle={Identification and Inference for Econometric Models: Essays in Honor of Thomas Rothenberg},
  pages={80--108},
  year={2005},
  publisher={Cambridge University Press}
}

@article{erdds1959random,
  title={On random graphs I},
  author={Erdös, P. and Rényi, A},
  journal={Publ. math. debrecen},
  volume={6},
  number={290-297},
  pages={18},
  year={1959}
}

@article{kojevnikov2021limit,
  title={Limit theorems for network dependent random variables},
  author={Kojevnikov, Denis and Marmer, Vadim and Song, Kyungchul},
  journal={Journal of Econometrics},
  volume={222},
  number={2},
  pages={882--908},
  year={2021},
  publisher={Elsevier}
}

@article{acemoglu2015,
Author = {Acemoglu, Daron and García-Jimeno, Camilo and Robinson, James A.},
Title = {State Capacity and Economic Development: A Network Approach},
Journal = {American Economic Review},
Volume = {105},
Number = {8},
Year = {2015},
Month = {8},
Pages = {2364-2409},
DOI = {10.1257/aer.20140044},
URL = {https://www.aeaweb.org/articles?id=10.1257/aer.20140044}}

@article{conley1999gmm,
  title={GMM estimation with cross sectional dependence},
  author={Conley, Timothy G},
  journal={Journal of econometrics},
  volume={92},
  number={1},
  pages={1--45},
  year={1999},
  publisher={Elsevier}
}

@article{durlauf2008understanding,
  title={Understanding regression versus variance tests for social interactions},
  author={Durlauf, Steven N and Tanaka, Hisatoshi},
  journal={Economic Inquiry},
  volume={46},
  number={1},
  pages={25--28},
  year={2008},
  publisher={Wiley Online Library}
}

@incollection{brock2001interactions,
  title={Interactions-based models},
  author={Brock, William A and Durlauf, Steven N},
  booktitle={Handbook of econometrics},
  volume={5},
  pages={3297--3380},
  year={2001},
  publisher={Elsevier}
}

@book{blume2005identifying,
  title={Identifying social interactions: A review},
  author={Blume, Lawrence E and Durlauf, Steven N and others},
  year={2005},
  publisher={Social Systems Research Institute, University of Wisconsin Madison}
}

@article{gaviria2001school,
  title={School-based peer effects and juvenile behavior},
  author={Gaviria, Alejandro and Raphael, Steven},
  journal={Review of Economics and Statistics},
  volume={83},
  number={2},
  pages={257--268},
  year={2001},
  publisher={MIT Press 238 Main St., Suite 500, Cambridge, MA 02142-1046, USA journals~…}
}

@article{ioannides2003neighbourhood,
  title={Neighbourhood effects and housing demand},
  author={Ioannides, Yannis M and Zabel, Jeffrey E},
  journal={Journal of applied Econometrics},
  volume={18},
  number={5},
  pages={563--584},
  year={2003},
  publisher={Wiley Online Library}
}

@article{durlauf2004neighborhood,
  title={Neighborhood effects},
  author={Durlauf, Steven N},
  journal={Handbook of regional and urban economics},
  volume={4},
  pages={2173--2242},
  year={2004},
  publisher={Elsevier}
}

@article{field2016friendship,
  title={Friendship at work: Can peer effects catalyze female entrepreneurship?},
  author={Field, Erica and Jayachandran, Seema and Pande, Rohini and Rigol, Natalia},
  journal={American Economic Journal: Economic Policy},
  volume={8},
  number={2},
  pages={125--153},
  year={2016},
  publisher={American Economic Association 2014 Broadway, Suite 305, Nashville, TN 37203-2425}
}

@article{moreira2009tests,
  title={Tests with correct size when instruments can be arbitrarily weak},
  author={Moreira, Marcelo J},
  journal={Journal of Econometrics},
  volume={152},
  number={2},
  pages={131--140},
  year={2009},
  publisher={Elsevier}
}

@article{stock00,
  title={GMM with weak identification},
  author={Stock, James H. and Wright, Jonathan H},
  journal={Econometrica},
  volume={68},
  number={5},
  pages={1055--1096},
  year={2000},
  publisher={Wiley Online Library}
}

@article{moreira03,
  title={A conditional likelihood ratio test for structural models},
  author={Moreira, Marcelo J},
  journal={Econometrica},
  volume={71},
  number={4},
  pages={1027--1048},
  year={2003},
  publisher={Wiley Online Library}
}

@article{campbell24,
  title={The network origins of entry},
  author={Campbell, Arthur and Ushchev, Philip and Zenou, Yves},
  journal={Journal of Political Economy},
  volume={132},
  number={11},
  pages={3867--3916},
  year={2024},
  publisher={The University of Chicago Press Chicago, IL}
}

@book{graham20,
  title={The econometric analysis of network data},
  author={Graham, Bryan and De Paula, {\'A}ureo},
  year={2020},
  publisher={Academic Press}
}

@article{mele17,
  title={A structural model of dense network formation},
  author={Mele, Angelo},
  journal={Econometrica},
  volume={85},
  number={3},
  pages={825--850},
  year={2017},
  publisher={Wiley Online Library}
}

@book{jackson10,
  title={Social and Economic Networks},
  author={Jackson, Matthew O},
  year={2010},
  publisher={Princeton University Press}
}

@article{Sacerdote2001,
  author = {Sacerdote, Bruce},
  title = {Peer Effects with Random Assignment: Results for Dartmouth Roommates},
  journal = {Quarterly Journal of Economics},
  year = {2001},
  volume = {116},
  number = {2},
  pages = {681--704},
  doi = {10.1162/00335530151144131}
}

@article{MasMoretti2009,
  author = {Mas, Alexandre and Moretti, Enrico},
  title = {Peers at Work},
  journal = {American Economic Review},
  year = {2009},
  volume = {99},
  number = {1},
  pages = {112--145},
  doi = {10.1257/aer.99.1.112}
}

@article{LoughranSchultz2004,
  author = {Loughran, Tim and Schultz, Paul},
  title = {IPO Pricing and Peer Influence: Evidence from the Years 1980–2000},
  journal = {Journal of Financial Economics},
  year = {2004},
  volume = {71},
  number = {3},
  pages = {225--254},
  doi = {10.1016/S0304-405X(03)00213-9}
}

@article{caria2024village,
  title={Village social structure and labor market performance: Evidence from the Philippines},
  author={Caria, A. Stefano and Labonne, Julien},
  journal={Journal of Economic Behavior \& Organization},
  volume={219},
  pages={371-380},
  year={2024},
  publisher={Elsevier}
}

@misc{krivelevich2001,
      title={The largest eigenvalue of sparse random graphs}, 
      author={Michael Krivelevich and Benny Sudakov},
      year={2001},
      eprint={math/0106066},
      archivePrefix={arXiv},
      primaryClass={math.CO},
      url={https://arxiv.org/abs/math/0106066}, 
}

@Misc{weakiv2013,
  author={Keith Finlay and Leandro Magnusson and Mark E Schaffer},
  title={{WEAKIV: Stata module to perform weak-instrument-robust tests and confidence intervals for instrumental-variable (IV) estimation of linear, probit and tobit models}},
  year=2013,
  month=Aug,
  howpublished={Statistical Software Components, Boston College Department of Economics},
  doi={},
  url={https://ideas.repec.org/c/boc/bocode/s457684.html},
}

@article{Leung2020,
	author = {Leung, Michael P.},
	doi = {10.1162/rest_a_00818},
	eprint = {https://direct.mit.edu/rest/article-pdf/102/2/368/1881176/rest_a_00818.pdf},
	issn = {0034-6535},
	journal = {The Review of Economics and Statistics},
	month = {05},
	number = {2},
	pages = {368-380},
	title = {Treatment and Spillover Effects Under Network Interference},
	url = {https://doi.org/10.1162/rest_a_00818},
	volume = {102},
	year = {2020}}

@article{dePaula2018,
author = {de Paula, {\'A}ureo and Richards-Shubik, Seth and Tamer, Elie},
title = {Identifying Preferences in Networks With Bounded Degree},
journal = {Econometrica},
volume = {86},
number = {1},
pages = {263-288},
doi = {https://doi.org/10.3982/ECTA13564},
url = {https://onlinelibrary.wiley.com/doi/abs/10.3982/ECTA13564},
eprint = {https://onlinelibrary.wiley.com/doi/pdf/10.3982/ECTA13564},
year = {2018}
}

@article{lambotte2025peer,
author = {Lambotte, Mathieu},
title = {Peer Effects in Binary Outcomes: Strategic Complementarity and Taste for Conformity With Endogenous Networks},
journal = {Journal of Applied Econometrics},
volume = {40},
number = {6},
pages = {608-626},
doi = {https://doi.org/10.1002/jae.3128},
url = {https://onlinelibrary.wiley.com/doi/abs/10.1002/jae.3128},
eprint = {https://onlinelibrary.wiley.com/doi/pdf/10.1002/jae.3128},
year = {2025}
}

@TechReport{Wang2025,
type={Papers},
institution={arXiv.org},
author={William W. Wang and Ali Jadbabaie},
title={Weak Identification in Peer Effects Estimation},
year={2025},
month={8},
number={2508.04897},
doi={None},
url={https://ideas.repec.org/p/arx/papers/2508.04897.html},
}

@techreport{volpe2025discrete,
  title={Discrete Choice with Generalized Social Interactions},
  author={Volpe, Oscar},
  year={2025},
  institution={Working paper, Harvard University}
}

@article{tchuente2019weakidentification,
title = {Weak identification and estimation of social interaction models},
journal = {Journal of Statistical Planning and Inference},
volume = {246},
pages = {106439},
year = {2027},
issn = {0378-3758},
doi = {https://doi.org/10.1016/j.jspi.2026.106439},
url = {https://www.sciencedirect.com/science/article/pii/S0378375826000674},
author = {Guy Tchuente},
}

@article{stockwrightyogo2002,
journal={Journal of Business \& Economic Statistics},
author={Stock, James H and Wright, Jonathan H and Yogo, Motohiro},
title={A Survey of Weak Instruments and Weak Identification in Generalized Method of Moments},
year={2002},
month={October},
pages={518-529},
volume={20},
number={4},
doi={None},
url={https://ideas.repec.org/a/bes/jnlbes/v20y2002i4p518-29.html},
}

@article{ross2022,
author = {Stephen L. Ross and Zhentao Shi},
title = {Measuring Social Interaction Effects When Instruments Are Weak},
journal = {Journal of Business \& Economic Statistics},
volume = {40},
number = {3},
pages = {995--1006},
year = {2022},
publisher = {Taylor \& Francis},
doi = {10.1080/07350015.2021.1895811},
URL = { 
    
        https://doi.org/10.1080/07350015.2021.1895811
    
    

},
eprint = {
        https://doi.org/10.1080/07350015.2021.1895811
}

}

\newpage                          
\pagenumbering{roman}
\setcounter{page}{1}
\counterwithin{figure}{section}
\appendix

\section*{Supplemental Appendix to ``Empirical Challenges with Peers-of-Peers Instruments
in the Linear-In-Means Model"}\label{App}

\renewcommand{\thesection}{\Alph{section}}
\renewcommand{\theHassumption}{appendix.\thesection.\arabic{assumption}}
\renewcommand{\theHdefinition}{appendix.\thesection.\arabic{definition}}
\section{Proofs and Additional Results}\label{SA1}

\subsection{Proofs}
\begin{customproof}{Proposition \ref{prop:cov_bound_unscl}}
Let $\{G_n\}_{n\ge 1}$ be a sequence of random graphs with 
$G_n \sim \mathsf{ER}(n,p_n)$, where $p_n\in(0,1)$ may depend on \(n\). 
Denote the expected degree by $d_n = (n-1)p_n \sim n p_n$, 
and let $\mathbf{G}_n=(a_{ij})$ be the adjacency matrix of $G_n$ 
(with $a_{ii}=0$ and $a_{ij}=a_{ji}$). Define the degree matrix 
$\mathbf{D}_n = \operatorname{diag}(\mathbf{G}_n\mathbf{1})$.
Let
\[
\mathbf{G}^{(2)}_n := \mathbf{G}_n^2 - \mathbf{D}_n.
\]
Then $(\mathbf{G}^{(2)}_n)_{ij}$ equals the number of two–step walks between nodes $i,j$ where $i\neq j$.  
Using Assumption \ref{A3.2} we know that $ |\beta| \|\mathbf{G}_n\|_2 < 1$.

Our goal is to prove asymptotic upper and lower bounds on the variance-normalized covariance between the friends-of-friends instrument and the endogenous variable given by:
\[
\frac{\frac{1}{n}\mathrm{Cov}\!\left(\mathbf{G}^{(2)}_n \mathbf{X}, \mathbf{G}_n \mathbf{Y}\right) }{\frac{1}{n}\mathrm{Var}\!\left(\mathbf{G}^{(2)}_n \mathbf{X}\right) }.
\]
For Erdős–Rényi graphs, the probability of link formation is independent of $\mathbf{X}$ then by the Law of Iterated Expectations (LoIE), independence of $\mathbf{G}_n$ from $\mathbf{X}$, and mean zero of $X$ ($\mu_x=0$), we obtain
\begin{align}\label{eq:pf_mux}
    \mathbb{E}\!\left[ \mathbf{X}^T \mathbf{G}^{(2)}_n \right] 
    &= \mathbb{E}\!\left[ \mathbf{X}^T \, \mathbb{E}\!\left[ \mathbf{G}^{(2)}_n \mid \mathbf{X} \right] \right] \\
    &= \mathbb{E}\!\left[\mathbf{X}^T \right] \mathbb{E}\!\left[ \mathbf{G}^{(2)}_n \mid \mathbf{X} \right]  \notag \\
    &= \mu_x \, \mathbb{E}\!\left[ \mathbf{G}^{(2)}_n \right] \iota^T \;=\; 0. \notag
\end{align}
Using~\eqref{eq:pf_mux} together with equation~\eqref{exp_eqn},
$\mathrm{Cov}\!\left(\mathbf{G}^{(2)}_n \mathbf{X}, \mathbf{G}_n \mathbf{Y}\right)$ expands as
\begin{align*}
    \mathrm{Cov}\!\left(\mathbf{G}^{(2)}_n \mathbf{X}, \mathbf{G}_n \mathbf{Y}\right) 
    &= \mathbb{E}\!\left[ \mathbf{X}^T \mathbf{G}^{(2)}_n \mathbf{G}_n \mathbf{Y} \right] 
       - \mathbb{E}\!\left[ \mathbf{X}^T \mathbf{G}^{(2)}_n \right] \mathbb{E}\!\left[ \mathbf{G}_n \mathbf{Y} \right] \\
    &= \mathbb{E}\!\left[ \mathbf{X}^T \mathbf{G}^{(2)}_n \left(
         \tfrac{\alpha}{1-\beta} \mathbf{G}_n \iota
         + \gamma \sum_{k=0}^{\infty} \beta^k \mathbf{G}_n^{k+1} \mathbf{X}
         + \delta \sum_{k=0}^{\infty} \beta^k \mathbf{G}_n^{k+2} \mathbf{X}
         + \sum_{k=0}^{\infty} \beta^k \mathbf{G}_n^{k+2} \boldsymbol{\epsilon}
     \right) \right] \\
    &= \mathbb{E}\!\left[
         \tfrac{\alpha}{1-\beta} \mathbf{X}^T \mathbf{G}^{(2)}_n \mathbf{G}_n \iota
         + \gamma \, \mathbf{X}^T \mathbf{G}^{(2)}_n \sum_{k=0}^{\infty}\beta^k \mathbf{G}_n^{k+1}\mathbf{X}
         + \delta \, \mathbf{X}^T \mathbf{G}^{(2)}_n \sum_{k=0}^{\infty}\beta^k \mathbf{G}_n^{k+2}\mathbf{X}
     \right].
\end{align*}
The final step follows from Assumption~\ref{A3.1}(iv). Furthermore, by the LoIE and $\mu_x = 0$, the first term vanishes. Hence, we obtain
\begin{align}\label{eq:cov_exp}
    \mathrm{Cov}\!\left(\mathbf{G}^{(2)}_n \mathbf{X}, \mathbf{G}_n \mathbf{Y} \right) 
    &= \mathbb{E}\!\left[
        \gamma \, \mathbf{X}^T \left(\mathbf{G}^2_n - \mathbf{D}_n \right) 
        \sum_{k=0}^{\infty}\beta^k \mathbf{G}_n^{k+1}\mathbf{X} 
        + \delta \, \mathbf{X}^T \left(\mathbf{G}^2_n - \mathbf{D}_n \right) 
        \mathbf{G}_n \sum_{k=0}^{\infty}\beta^k \mathbf{G}_n^{k+1}\mathbf{X}
    \right] \notag \\
    &= \mathbb{E}\!\left[
        \mathbf{X}^T \left(\mathbf{G}^2_n - \mathbf{D}_n \right) 
        \left( \gamma \mathbf{I}_n + \delta \mathbf{G}_n \right)
        \sum_{k=0}^{\infty}\beta^k \mathbf{G}_n^{k+1}\mathbf{X}
    \right] \notag \\
        &= \sigma_x^2 \mathbb{E}\!\left[
        \operatorname{Tr} \left( \left(\mathbf{G}^2_n - \mathbf{D}_n \right) 
        \left( \gamma \mathbf{I}_n + \delta \mathbf{G}_n \right)
        \sum_{k=0}^{\infty}\beta^k \mathbf{G}_n^{k+1} \right)
    \right],
\end{align}
where the last equality follows from the cyclical property of the trace and $\mathbb{E}[\mathbf{X} \mathbf{X}^T \mid \mathbf{G}_n] = \mathbb{E}[\mathbf{X} \mathbf{X}^T ] = \sigma_x^2 \mathbb{I}_n$.
Applying the three–factor trace inequality 
$|\operatorname{Tr}(ABC)| \le \|A\|_F\|B\|_2\|C\|_F$, 
and using $|\beta| \|\mathbf{G}_n\|_2<1$, gives
\begin{align}\label{eq:norm_ineq}
    \big|\mathbb{E} \left[\operatorname{Tr}(\left(\mathbf{G}^2_n - \mathbf{D}_n \right) 
        \left( \gamma \mathbf{I}_n + \delta \mathbf{G}_n \right)
        \sum_{k=0}^{\infty}\beta^k \mathbf{G}_n^{k+1}) \right] \big|
\;\le\;
\mathbb{E} ( \|\mathbf{G}^2_n - \mathbf{D}_n\|_F
        \| \gamma \mathbf{I}_n + \delta \mathbf{G}_n\|_2
        \|\sum_{k=0}^{\infty}\beta^k \mathbf{G}_n^{k+1}\|_F).
\end{align}
By the triangle inequality for the spectral norm,
\begin{align}\label{eq:B}
\|\gamma \mathbf{I}_n+\delta \mathbf{G}_n\|_2
\le |\gamma|\,\|\mathbf{I}_n\|_2+|\delta|\,\|\mathbf{G}_n\|_2
< |\gamma|+\frac{|\delta|}{|\beta|}.
\end{align}
Furthermore, using sub-multiplicativity of the Frobenius norm, $\|XY\|_F\le \|X\|_2\|Y\|_F$, iteratively we get,
\begin{align}\label{eq:C}
    \Big\|\sum_{k\ge 0}\beta^k \mathbf{G}_n^{k+1}\Big\|_F
\le \sum_{k\ge 0}|\beta|^k\,\|\mathbf{G}_n^{k+1}\|_F
\le \sum_{k\ge 0}|\beta|^k\,\|\mathbf{G}_n\|_2^{\,k}\,\|\mathbf{G}_n\|_F
< \frac{\|\mathbf{G}_n\|_F}{1-|\beta| \|\mathbf{G}_n\|_2},
\end{align}
since $|\beta| \ \|\mathbf{G}_n\|_2< 1$.
Combining equations \eqref{eq:norm_ineq}, \eqref{eq:B} and \eqref{eq:C} gives the bound
\begin{align}\label{eq:final_bound}
    \frac{1}{n}\big|\mathbb{E}\left[ \operatorname{Tr}(\left(\mathbf{G}^2_n - \mathbf{D}_n \right) 
        \left( \gamma \mathbf{I}_n + \delta \mathbf{G}_n \right)
        \sum_{k=0}^{\infty}\beta^k \mathbf{G}_n^{k+1})\right] \big|
\;<\;
\frac{|\gamma|+|\delta|/|\beta|}{1-|\beta|\|\mathbf{G}_n\|_2}\,
\frac{1}{n}\,
\mathbb{E}\!\big[\|\mathbf{G}^{(2)}_n\|_F\,\|\mathbf{G}_n\|_F\big].
\end{align}
Using the Cauchy-Schwartz inequality for matrix norm i.e. $\mathbb{E}[UV] \leq (\mathbb{E}[U^2])^{1/2} (\mathbb{E}[V^2])^{1/2}$, gives us
\begin{align}\label{eq:cov_fin}
\frac{1}{n}\big|\operatorname{Cov}(\mathbf{G}^{(2)}_n\mathbf{X},\mathbf{G}_n\mathbf{Y})\big|
\le
\frac{|\gamma|+|\delta|/|\beta|}{1-|\beta|\|\mathbf{G}_n\|_2} \frac{\sigma_x^2}{n}
\Big(\mathbb{E}\|\mathbf{G}^{(2)}_n\|_F^2\Big)^{1/2}
\Big(\mathbb{E}\|\mathbf{G}_n\|_F^2\Big)^{1/2}.
\end{align}

Now, to get an expression for the variance of the instrument we use the fact that $\mathbf{X}\in\mathbb{R}^n$ satisfies $\mathbb{E}[\mathbf{X}]=0$,
$\mathrm{Cov}(\mathbf{X})=\sigma_x^2 \mathbf{I}_n$, and $\mathbf{X}\perp\!\!\!\perp \mathbf{G}_n$,
\[\operatorname{Var}\!\big(\mathbf{G}^{(2)}_n \mathbf{X}\,\big|\mathbf{G}_n\big)
=\mathbb{E}\!\left[\big\|\mathbf{G}^{(2)}_n \mathbf{X}\big\|^2\,\big|\,\mathbf{G}_n\right]
=\mathbb{E}\!\left[\mathbf{X}^\top \big(\mathbf{G}^{(2)}_n\big)^\top \mathbf{G}^{(2)}_n \mathbf{X}\,\big|\,\mathbf{G}_n\right].
\]
Using $\mathbb{E}[\mathbf{X}\mathbf{X}^\top]=\sigma_x^2 \mathbf{I}_n$ and independence of $\mathbf{X}$ and $\mathbf{G}_n$ which follows from the independence of link formation in ER graphs,
\[
\mathbb{E}\!\left[\mathbf{X}^\top (\mathbf{G}^{(2)}_n)^\top \mathbf{G}^{(2)}_n \mathbf{X}\right]=\mathrm{Tr}\!\big((\mathbf{G}^{(2)}_n)^\top \mathbf{G}^{(2)}_n\,\mathbb{E}[\mathbf{X}\mathbf{X}^\top]\big)
=\sigma_x^2\,\mathrm{Tr}((\mathbf{G}^{(2)}_n)^\top \mathbf{G}^{(2)}_n).
\]
Using again $||A||_F = \sqrt{\operatorname{Tr}(A^TA)}$ and the symmetry of $\mathbf{G}$, we get
\[
\operatorname{Var}\!\big(\mathbf{G}^{(2)}_n \mathbf{X}\,\big|\mathbf{G}_n\big)
=\sigma_x^2\,\mathrm{Tr}\!\big((\mathbf{G}^{(2)}_n)^\top \mathbf{G}^{(2)}_n\big)
=\sigma_x^2\,\big\|\mathbf{G}^{(2)}_n\big\|_F^2.
\]
Taking expectations over $\mathbf{G}_n$ and using the homoskedasticity of $\mathbf{X}$, the Law of Total Variance gives us,
\begin{align}\label{eq:var_ins}
\frac{1}{n}\operatorname{Var}\!\big(\mathbf{G}^{(2)}_n \mathbf{X}\big)
=\frac{1}{n}\,\mathbb{E}\!\left[\operatorname{Var}\!\big(\mathbf{G}^{(2)}_n \mathbf{X}\,\big|\mathbf{G}_n\big)\right]
=\frac{\sigma_x^2}{n}\,\mathbb{E}\big\|\mathbf{G}^{(2)}_n\big\|_F^2.
\end{align}
Normalizing \eqref{eq:cov_fin} by this variance yields
\begin{equation}\label{eq:slope_def}
\frac{\tfrac{1}{n}\operatorname{Cov}( \mathbf{G}^{(2)}_n \mathbf{X},\mathbf{G}_n \mathbf{Y})}{\tfrac{1}{n}\operatorname{Var}( \mathbf{G}^{(2)}_n \mathbf{X})}
\;\le\;
\frac{|\gamma|+|\delta|/|\beta|}{1-|\beta|\|\mathbf{G}_n\|_2}\cdot
\sqrt{\frac{\mathbb{E}\|\mathbf{G}_n\|_F^2}{\mathbb{E}\|\mathbf{G}^{(2)}_n\|_F^2}}.
\end{equation}

We now solve for the two expected Frobenius norm terms. For $i\neq j$, the $(i,j)^{th}$ entry of $\mathbf{G}_n^2$ counts the number of two–step walks between $i$ and $j$:
\[
(\mathbf{G}_n^2)_{ij}
= \sum_{k\neq i,j} G_{ik} G_{kj},
\]
since $G_{ik}G_{kj}=1$ precisely when $k$ is a common neighbor of $i$ and $j$ hence, for any fixed $k\neq i,j$,
\[
\Pr(G_{ik}G_{kj}=1)
= \Pr(G_{ik}=1)\Pr(G_{kj}=1)
= p_n^2.
\]
Furthermore, the indicators $\{G_{ik}G_{kj}\}_{k\neq i,j}$ are independent across $k$. Thus, each term in the sum is Bernoulli$(p_n^2)$, and there are exactly $n-2$ such terms. Hence, $(\mathbf{G}_n^2)_{ij}
\sim \mathrm{Bin}(n-2,\,p_n^2)$ which gives us,
\begin{align*}
\mathbb{E}[(G_n^2)_{ij}^2] = \operatorname{Var}((G_n^2)_{ij}) + (\mathbb{E}(G_n^2)_{ij})^2  
=&(n-2)p_n^2(1-p_n^2)+(n-2)^2p_n^4 \notag \\
=& n p_n^2(1-p_n^2) + n^2p_n^4 + O(p^2_n + n p_n^4) \notag \\
=& p_n^2\left[ n + p_n^2 n^2 \right] + O(p^2_n + n p_n^4) 
\end{align*}
for large n. Since edges are $i.i.d.$, the expected value of the Frobenius norm for $\mathbf{G}^{(2)}_n$ gives us,
\begin{align}\label{eq:den}
\mathbb{E}\|\mathbf{G}^{(2)}_n\|_F^2
&= \sum_{i \neq j} \mathbb{E}((G^2_n)_{ij})^2  \notag \\
&= n(n-1) \left(p_n^2\left[ n + p_n^2 n^2\right] + O(p^2_n + n p_n^4)\right)
\end{align}
Similarly, for the Frobenius norm of $\mathbf{G}_n$ we have
\[\|\mathbf{G}_n\|_F^2 
= \sum_{i\neq j}(G_n)_{ij},
\] where $G_{ij}\sim \mathrm{Bernoulli}\!\left(p_n\right)$.
Taking expectation and summing over all non-diagonal elements gives us,
\begin{align}\label{eq:num}
\mathbb{E}\|\mathbf{G}_n\|_F^2 = n(n-1) p_n.
\end{align}
Putting equations \eqref{eq:den} and \eqref{eq:num} together and using $d_n = n p_n$, we get
\begin{align}\label{eq:ratio_bound}
\sqrt{\frac{\mathbb{E}\|\mathbf{G}_n\|_F^2}{\mathbb{E}\|\mathbf{G}^{(2)}_n\|_F^2}}
&=
\frac{1}{\sqrt{n p_n + n^2 p_n^3 + O(p_n + n p_n^3)}} \notag \\
&=
\frac{1}{\sqrt{d_n + \tfrac{d_n^3}{n} + O(p_n + n p_n^3)}}.
\end{align}

Consider $f(x) = x^{-1/2}$.  
A first-order Taylor expansion (mean-value form) around  
\[
x_0 := d_n + \tfrac{d_n^3}{n}
\]
gives, for some $\theta_n\in(0,1)$,
\[
f\bigl(x_0 + O(p_n + n p_n^3)\bigr)
=
f(x_0) 
+ 
f'\!\bigl(x_0 + \theta_n\,O(p_n + n p_n^3)\bigr)\,
O(p_n + n p_n^3).
\]
Since $f'(x) = -\tfrac12 x^{-3/2}$, this yields
\[
\frac{1}{\sqrt{d_n + \tfrac{d_n^3}{n} + O(p_n + n p_n^3)}}
=
\frac{1}{\sqrt{d_n + \tfrac{d_n^3}{n}}}
+
O\!\left(
\frac{p_n + n p_n^3}{(n p_n + n^2 p_n^3)^{3/2}}
\right).
\]
To bound the remainder, note that $0 \le p_n \le 1$ implies
\[
\frac{p_n + n p_n^3}{(n p_n + n^2 p_n^3)^{3/2}}
=
\frac{1}{n (np_n+n^2 p_n^3)^{1/2}}
= \frac{1}{n \sqrt{d_n + d_n^3/n}}.
\]
Thus,
\[
\sqrt{\frac{\mathbb{E}\|\mathbf{G}_n\|_F^2}{\mathbb{E}\|\mathbf{G}^{(2)}_n\|_F^2}}
=
\frac{1}{\sqrt{d_n + \tfrac{d_n^3}{n}}} (1 + o(1)).
\]
This gives the normalized covariance bound
\[
\left|
\frac{\tfrac{1}{n}\operatorname{Cov}(\mathbf{G}_n^{(2)}\mathbf{X},\mathbf{G}_n\mathbf{Y})}
     {\tfrac{1}{n}\operatorname{Var}(\mathbf{G}_n^{(2)}\mathbf{X})}
\right|
\le
\frac{|\gamma|+|\delta|/|\beta|}{1-|\beta|\|\mathbf{G}_n\|_2}
\cdot
\Biggl(
\frac{1}{\sqrt{d_n + \tfrac{d_n^3}{n}}} (1 + o(1))
\Biggr),
\]
and therefore,
\begin{equation}\label{eq:varnorm_final}
\left|
\frac{\tfrac{1}{n}\operatorname{Cov}(\mathbf{G}_n^{(2)}\mathbf{X},\mathbf{G}_n\mathbf{Y})}
     {\tfrac{1}{n}\operatorname{Var}(\mathbf{G}_n^{(2)}\mathbf{X})}
\right| 
=
O\!\left(
\frac{1}{\sqrt{d_n + d_n^3/n}}
\right),
\qquad n\to\infty.
\end{equation}
\end{customproof}

\newpage

\begin{lemma}[Eigenvalue decomposition of the first-stage covariance]
\label{lem:eigen_cov}
Consider the set-up in Propositions~\ref{prop:cov_bound_unscl} and ~\ref{prop:cov_bound_scaled} and let Assumption \ref{A3.1} hold. Suppose the symmetric adjacency matrix admit the eigen-decomposition
$\mathbf G_n=\mathbf V_n\mathbf\Lambda_n\mathbf V_n'$, where
$\mathbf\Lambda_n=\mathrm{diag}(\lambda_1(n),\dots,\lambda_n(n))$
and $\mathbf V_n$ is orthonormal. If $(\mathbf I_n-\beta\mathbf G_n)$ is invertible, then,
conditional on $\mathbf G_n$,
\begin{align*}
\frac{1}{n}Cov(\mathbf G_n^{(2)}\mathbf X,\mathbf G_n\mathbf Y\mid \mathbf G_n)
&=
\frac{\sigma_x^2}{n}\sum_{j=1}^n
\frac{\lambda_j(n)^3(\gamma+\delta\lambda_j(n))}{1-\beta\lambda_j(n)}
\;+\;
\frac{\sigma_x^2}{n}\,R_{n,\mathrm{diag}},
\end{align*}
where
\[
R_{n,\mathrm{diag}}
:=
-\Tr\!\Big(
\mathbf D_n\,\mathbf G_n(\mathbf I_n-\beta\mathbf G_n)^{-1}(\gamma\mathbf I_n+\delta\mathbf G_n)
\Big),
\qquad
\mathbf D_n:=\operatorname{diag}(\mathbf G_n^2).
\]
Moreover, the remainder term satisfies the deterministic bound
\[
|R_{n,\mathrm{diag}}|
\;\le\;
\operatorname{Tr}(\mathbf D_n)\,
\|\mathbf G_n\|_2\,
\|(\mathbf I_n-\beta\mathbf G_n)^{-1}\|_2\,
\|\gamma\mathbf I_n+\delta\mathbf G_n\|_2.
\]
\end{lemma}

\begin{proof}
Fix $\mathbf G_n$ and write the reduced form
\[
\mathbf Y
=
(\mathbf I_n-\beta\mathbf G_n)^{-1}\big(\alpha\iota+(\gamma\mathbf I_n+\delta\mathbf G_n)\mathbf X+\boldsymbol\varepsilon\big),
\]
which holds by invertibility of $\mathbf I_n-\beta\mathbf G_n$. Using $\mathbf X\perp\!\!\!\perp \boldsymbol\varepsilon$ and
$\mathbb E[\mathbf X]=0$, conditional on $\mathbf G_n$ we have
\begin{align*}
Cov(\mathbf G_n^{(2)}\mathbf X,\mathbf G_n\mathbf Y\mid \mathbf G_n)
&=
Cov\!\Big(\mathbf G_n^{(2)}\mathbf X,\,
\mathbf G_n(\mathbf I_n-\beta\mathbf G_n)^{-1}(\gamma\mathbf I_n+\delta\mathbf G_n)\mathbf X
\;\Big|\;\mathbf G_n\Big),
\end{align*}
since $\mathbb E[\mathbf G_n^{(2)}\mathbf X\mid \mathbf G_n]=\mathbf 0$ implies
$Cov(\mathbf G_n^{(2)}\mathbf X,\mathbf G_n\iota\mid \mathbf G_n)=0$, and the
$\boldsymbol\varepsilon$ term drops out by Assumption~\ref{A3.1}(iv).

By the conditional second-moment restriction
$\mathbb E[\mathbf X\mathbf X'\mid \mathbf G_n]=\sigma_x^2\mathbf I_n$ and symmetry of
$\mathbf G_n$, we obtain
\[
Cov\!\Big(\mathbf G_n^{(2)}\mathbf X,\,
\mathbf G_n(\mathbf I_n-\beta\mathbf G_n)^{-1}(\gamma\mathbf I_n+\delta\mathbf G_n)\mathbf X
\;\Big|\;\mathbf G_n\Big)
=
\sigma_x^2Tr\!\Big(\mathbf G_n^{(2)}\mathbf G_n(\mathbf I_n-\beta\mathbf G_n)^{-1}(\gamma\mathbf I_n+\delta\mathbf G_n)\Big).
\]

Decomposing $\mathbf G_n^{(2)}=\mathbf G_n^2-\mathbf D_n$ yields
\begin{align*}
\sigma_x^2\Tr\!\Big(\mathbf G_n^{(2)}\mathbf G_n(\mathbf I_n-\beta\mathbf G_n)^{-1}(\gamma\mathbf I_n+\delta\mathbf G_n)\Big)
&=
\sigma_x^2\Tr\!\Big(\mathbf G_n^3(\mathbf I_n-\beta\mathbf G_n)^{-1}(\gamma\mathbf I_n+\delta\mathbf G_n)\Big)\\
&\quad
-\sigma_x^2\Tr\!\Big(\mathbf D_n\,\mathbf G_n(\mathbf I_n-\beta\mathbf G_n)^{-1}(\gamma\mathbf I_n+\delta\mathbf G_n)\Big).
\end{align*}

For the first trace, use the eigen-decomposition $\mathbf G_n=\mathbf V_n\mathbf\Lambda_n\mathbf V_n'$ and orthonormality of $\mathbf V_n$ to obtain
\begin{align*}
    Tr\!\Big(\mathbf G_n^3(\mathbf I_n-\beta\mathbf G_n)^{-1}(\gamma\mathbf I_n+\delta\mathbf G_n)\Big)
&=Tr\Big(
\mathbf V_n\mathbf\Lambda_n^3\mathbf V_n'
\mathbf V_n(\mathbf I_n-\beta\mathbf\Lambda_n)^{-1}\mathbf V_n'
\mathbf V_n(\gamma\mathbf I_n+\delta\mathbf\Lambda_n)\mathbf V_n'
\Big) \\
&=Tr\Big(
\mathbf V_n\mathbf\Lambda_n^3(\mathbf I_n-\beta\mathbf\Lambda_n)^{-1}(\gamma\mathbf I_n+\delta\mathbf\Lambda_n)\mathbf V_n'\Big).
\end{align*}

Using the cyclicity of the trace operator,
\begin{align*}
    Tr\Big(\mathbf G_n^3(\mathbf I_n-\beta\mathbf G_n)^{-1}(\gamma\mathbf I_n+\delta\mathbf G_n)\Big)
    &=Tr\Big(\mathbf\Lambda_n^3(\mathbf I_n-\beta\mathbf\Lambda_n)^{-1}(\gamma\mathbf I_n+\delta\mathbf\Lambda_n)\mathbf V_n'\mathbf V_n
\Big)\ \\
&=Tr\Big(
\mathbf\Lambda_n^3(\mathbf I_n-\beta\mathbf\Lambda_n)^{-1}(\gamma\mathbf I_n+\delta\mathbf\Lambda_n)
\Big), \\
&=\sum_{j=1}^n \frac{\lambda_j(n)^3(\gamma+\delta\lambda_j(n))}{1-\beta\lambda_j(n)}.
\end{align*}

For the remainder term, define
\[
\mathbf M_n
:=
\mathbf G_n(\mathbf I_n-\beta\mathbf G_n)^{-1}(\gamma\mathbf I_n+\delta\mathbf G_n).
\]
Then
\[
R_{n,\mathrm{diag}}
=
-\operatorname{Tr}(\mathbf D_n\mathbf M_n).
\]
Since \(\mathbf D_n=\operatorname{diag}(\mathbf G_n^2)\) is diagonal, we have
\[
\operatorname{Tr}(\mathbf D_n\mathbf M_n)
=
\sum_{i=1}^n (\mathbf D_n)_{ii}(\mathbf M_n)_{ii}.
\]
Therefore,
\[
|R_{n,\mathrm{diag}}|
=
\left|\sum_{i=1}^n (\mathbf D_n)_{ii}(\mathbf M_n)_{ii}\right|
\le
\sum_{i=1}^n (\mathbf D_n)_{ii}\,|(\mathbf M_n)_{ii}|.
\]
Now let \(e_i\) denote the \(i\)-th canonical basis vector in \(\mathbb R^n\). Since
\[
(\mathbf M_n)_{ii}=e_i'\mathbf M_n e_i,
\]
we obtain
\[
|(\mathbf M_n)_{ii}|
=
|e_i'\mathbf M_n e_i|
\le
\|\mathbf M_n\|_2\,\|e_i\|_2^2
=
\|\mathbf M_n\|_2.
\]
Hence
\[
|R_{n,\mathrm{diag}}|
\le
\sum_{i=1}^n (\mathbf D_n)_{ii}\,\|\mathbf M_n\|_2
=
\operatorname{Tr}(\mathbf D_n)\,\|\mathbf M_n\|_2.
\]
Using submultiplicativity of the operator norm,
\[
\|\mathbf M_n\|_2
\le
\|\mathbf G_n\|_2\,
\|(\mathbf I_n-\beta\mathbf G_n)^{-1}\|_2\,
\|\gamma\mathbf I_n+\delta\mathbf G_n\|_2,
\]
so
\[
|R_{n,\mathrm{diag}}|
\le
\operatorname{Tr}(\mathbf D_n)\,
\|\mathbf G_n\|_2\,
\|(\mathbf I_n-\beta\mathbf G_n)^{-1}\|_2\,
\|\gamma\mathbf I_n+\delta\mathbf G_n\|_2,
\]
which completes the proof.
\end{proof}

\bigskip

\begin{customproof}{Proposition \ref{prop:Asymp_Dist}}

We begin with the first part of the proposition which is to show that the variance in equation \eqref{eq:var_prop1} can be estimated by using the Network HAC estimator of \cite{kojevnikov2021limit}. The variance of the AR-test statistic is given by $\Omega(\beta) = V_{\hat{\xi}} - \beta_0 (Cov_{\hat{\xi} \hat{\pi}} + Cov_{\hat{\pi} \hat{\xi}}) + \beta^2_0 V_{\hat{\pi}}$ where $\hat{\xi}$ and $\hat{\pi}$ are the OLS estimators in equations \eqref{eq:struct_eq2} and \eqref{eq:first_stage} respectively. To apply the limit theorem results of \cite{kojevnikov2021limit}, we define :
\begin{align}
    Y_{n,\pi} = \tilde{\epsilon} z_i \\
    Y_{n,\xi} =  \eta z_i
\end{align}
and 
\begin{align}
    S_{n,\pi} = \sum_{i \in N_n} Y_{n,\pi} \\
    S_{n,\xi} =  \sum_{i \in N_n} Y_{n,\xi}
\end{align}
where $\tilde{\epsilon}$ and $\eta$ are the errors from equations \eqref{eq:first_stage} and \eqref{eq:struct_eq2} and $z_i$ is the $i^{th}$ row of $Z = (\iota, X, GX, G^2 X)$. Note that even if the error terms $\tilde{\epsilon}$ and $\eta$ are independent of the network, the vector $z_i$ is not and thus $Y_{n,\pi}$ and $Y_{n,\xi}$ are network dependent random variables.
Let $S_n = (S'_{n,\pi}, S'_{n,\xi})'$ be the stacked vector of network dependent variables of interest, then the variance of $S_{n}$ can be represented as the block matrix given by:
\[
V \left(S_n \right) =
\begin{bmatrix}
V_n(S_{n,\pi}) & Cov(S_{n,\xi}, S_{n,\pi})\\
Cov(S_{n,\pi}, S_{n,\xi})  & V(S_{n,\xi}) 
\end{bmatrix}
\]
We need a consistent estimator for $V\left(S_n \right)$ such that we can use it to get a consistent estimator for $\Omega(\beta)$ defined above which will be a simple application of the Slutsky Theorem since:

\[ \Omega(\beta) =
\begin{bmatrix} \beta I_k &-I_k \end{bmatrix}
(I_2 \otimes \Sigma_{zz}^{-1}) V_n(S_n) (I_2 \otimes \Sigma_{zz}^{-1}). 
\begin{bmatrix} \beta I_k \\ -I_k \end{bmatrix} \]
where $\Sigma_{zz}  = \mathbb{E}[z_iz_i']$ and $I_k$ and $I_2$ are the identity matrices of size k and 2 respectively.

For the limit theorem to apply for $V _n\left(S_n\right)$, we assume that our network dependent random variables, satisfy the following two assumptions (see \cite{kojevnikov2021limit} for detailed definitions and a discussion):

\begin{namedassumption}[Assumption 2.1 from \cite{kojevnikov2021limit}]
The triangular array $\{Y_{n,i}\}$ is conditionally $\psi$-dependent given $\{C_n\}$ with the dependence coefficients $\{\theta_n\}$ satisfying the following conditions.
\begin{itemize}
    \item[(a)] For some constant $C > 0$,
    \[
    \psi_{a,b}(f,g) \leq C \cdot ab \left( \|f\|_\infty + \text{Lip}(f) \right) \left( \|g\|_\infty + \text{Lip}(g) \right).
    \]
    
    \item[(b)] $\displaystyle \sup_{n \geq 1} \max_{s \geq 1} \theta_{n,s} < \infty$ a.s.
\end{itemize}
\end{namedassumption}

$\{\mathcal{C}_n\}_{n \geq 1}$ is a given sequence of $\sigma$-fields, $\{\theta_{n,s}\}_{s \geq 0}$ a $\mathcal{C}_n$-measurable sequence ($\theta_{n,0} = 1$) and $\psi_{a,b}$ are a collection of nonrandom functions.

\begin{namedassumption}[Assumption 4.1 from \cite{kojevnikov2021limit}]
There exists $p > 4$ such that
\begin{itemize}
  \item[(i)] $\displaystyle \sup_{n \geq 1} \max_{i \in N_n} \left\| Y_{n,i} \right\|_{C_n, p} < \infty \text{ a.s.},$
  
  \item[(ii)] $\displaystyle \lim_{n \to \infty} \sum_{s \geq 1} \left| \omega_n(s) - 1 \right| \delta_n^{\partial}(s) \theta_{n,s}^{1 - (2/p)} = 0 \text{ a.s.},$
  
  \item[(iii)] $\displaystyle \lim_{n \to \infty} n^{-1} \sum_{s \geq 0} c_n(s, b_n; 2) \theta_{n,s}^{1 - (4/p)} = 0 \text{ a.s.}$
\end{itemize}
\end{namedassumption}

with $\delta_n^\partial(s; k) = \frac{1}{n} \sum_{i \in \mathcal{N}_n} |N_n^\partial(i; s)|^k,$ where $N_n(i; s) = \{ j \in \mathcal{N}_n : d_n(i, j) \leq s \}$ and $ N_n^\partial(i; s) = \{ j \in \mathcal{N}_n : d_n(i, j) = s \}.$ Furthermore, $c_n(s, m; k) = \inf_{\alpha > 1} \left[ \Delta_n(s, m; k\alpha) \right]^{\frac{1}{\alpha}} \left[ \delta_n^\partial \left( s; \left( \frac{\alpha}{\alpha - 1} \right)^k \right) \right]^{1 - \frac{1}{\alpha}}.$ Finally, $\omega: \overline{\mathbf{R}} \rightarrow [-1, 1]$ is an appropriate kernel function and $b_n$ a bandwidth parameter such that $\omega_n(s) = \omega(s/b_n)$.

These two assumptions put restrictions on the denseness of the network as well as how strong the network effect is as we go further way from any particular node.

We define the following estimators:

\begin{align}
    \tilde{V}_n\left(\frac{S_n}{\sqrt{n}} \right) = \sum_{s\geq0} \omega_n(s) \tilde{\Omega}_n(s) \\
    \tilde{\Omega}_n(s) = n^{-1} \sum_{i \in N_n} \sum_{j \in N_n^\partial(i; s)} Y_{n,i} Y^T_{n,j}
    \end{align} where s is the path length between nodes in the network and $\omega_n(s)$ is the kernel function defined above. 
    
Then by Proposition 4.1 of \cite{kojevnikov2021limit}, $\tilde{V}_n  \left(\frac{S_n}{\sqrt{n}} \right) \xrightarrow{a.s.} V_n\left(\frac{S_n}{\sqrt{n}} \right)$ and finally, applying Slutsky's Theorem,  we get that $\hat{\Omega}_n(\beta) \xrightarrow{p} \Omega(\beta) $ which proves the first part of the proposition.

Now, we know that $\sqrt{n}\hat{g}(\beta_0) \stackrel{H_0}{\to}_d N(0,\Omega(\beta_0))$ (\citealp{anderson1949estimation}). Again, by the Slutsky Theorem and the first part of the proof, $\hat{\Omega}(\beta_0))^{-1/2} \sqrt{n}\hat{g}(\beta_0) \stackrel{H_0}{\to}_d N(0,\mathbf{I}_{k X k})$. Finally, by definition of $\chi^2$ distribution, $\tilde{AR}_n(\beta_0) = n\hat{g}(\beta_0)'\hat{\Omega}(\beta_0)^{-1}\hat{g}(\beta_0) \stackrel{H_0}{\to}_d \chi^2_k$.  

\end{customproof}

\newpage

\subsection{Scaled Specification}\label{SA1.scl}

This subsection states the scaled analogue of Assumption~\ref{A3.2} and Proposition~\ref{prop:cov_bound_unscl}. Let \(\mathbf A_n\) be an Erdős--Rényi adjacency matrix with average degree \(d_n=np_n\) and maximum degree \(\Delta_n\). Define
\[
w_n := \max\{d_n,\sqrt{\Delta_n}\},
\qquad
\mathbf G_n := \frac{1}{w_n}\mathbf A_n .
\]
The stability condition for the scaled specification is then
\begin{assumption}\label{A3.3}
\begin{equation}
\frac{\lambda_1^A(n)}{w_n} < \frac{1}{|\beta|},
\end{equation}
where \(\lambda_1^A(n)\) denotes the largest eigenvalue of \(\mathbf A_n\).
\end{assumption}

When \(w_n\) tracks the order of \(\lambda_1^A(n)\), Assumption~\ref{A3.3} allows the average degree to grow while keeping the scaled network operator stable. If \(w_n=\lambda_1^A(n)\), Assumption~\ref{A3.3} reduces to the exact spectral normalization case.

\begin{prop}[Upper bound on variance--normalized covariance in Erd\H{o}s--R\'enyi graphs]
\label{prop:cov_bound_scaled}
Let \(\{G(n,p_n)\}_{n\ge1}\) be a sequence of Erd\H{o}s--R\'enyi graphs with adjacency matrix \(\mathbf A_n\), expected degree \(d_n=np_n\), and maximum degree \(\Delta_n\). Let \(\{X_i\}_{i=1}^n\) be i.i.d.\ uni-dimensional real-valued random variables with \(\mathbb E[X_i]=0\) and \(\mathbb E[X_i^2]=\sigma_x^2\in(0,\infty)\), independent of \(\mathbf A_n\). Suppose Assumption~\ref{A3.1} holds. Define
\[
w_n := \max\{d_n,\sqrt{\Delta_n}\},
\qquad
\mathbf G_n := \frac{1}{w_n}\mathbf A_n,
\qquad
\mathbf G_n^{(2)} := \mathbf G_n^2 - \mathbf D_n,
\qquad
\mathbf D_n := \operatorname{diag}(\mathbf G_n^2).
\]
If Assumption~\ref{A3.3} holds, then there exists a constant \(c<\infty\) such that, for all sufficiently large \(n\),
\begin{equation}\label{eq:upper_scaled}
\left|
\frac{\frac{1}{n}\operatorname{Cov}(\mathbf G_n^{(2)}\mathbf X,\mathbf G_n\mathbf Y)}
     {\frac{1}{n}\operatorname{Var}(\mathbf G_n^{(2)}\mathbf X)}
\right|
\;\le\;
c\,\frac{w_n}{\sqrt{d_n + d_n^3/n}} .
\end{equation}
\end{prop}

\begin{proof}
Suppose \(\mathbf G_n=\mathbf A_n/w_n\), and write
\[
\mathbf A_n^{(2)}:=\mathbf A_n^2-\operatorname{diag}(\mathbf A_n\mathbf 1),
\qquad
\mathbf G_n^{(2)}=\frac{1}{w_n^2}\mathbf A_n^{(2)}.
\]
The proof follows the same decomposition strategy as Proposition~\ref{prop:cov_bound_unscl}. Under Assumptions~\ref{A3.1} and~\ref{A3.3}, the linear representation of \(\mathbf Y\) and independence of \(\mathbf X\) and \(\mathbf G_n\) imply that the covariance admits the same trace representation as in \eqref{eq:cov_exp}, now with the scaled network objects. Since \(|\beta|\|\mathbf G_n\|_2<1\), the derivations leading to \eqref{eq:slope_def} yield
\[
\left|
\frac{\tfrac{1}{n}\operatorname{Cov}(\mathbf{G}_n^{(2)}\mathbf{X},\,\mathbf{G}_n\mathbf{Y})}
     {\tfrac{1}{n}\operatorname{Var}(\mathbf{G}_n^{(2)}\mathbf{X})}
\right|
\;\le\;
\frac{|\gamma|+|\delta|/|\beta|}{1-|\beta|\|\mathbf{G}_n\|_2}\,
\sqrt{\frac{\mathbb{E}\|\mathbf{G}_n\|_F^2}
           {\mathbb{E}\|\mathbf{G}_n^{(2)}\|_F^2}} .
\]
Moreover,
\[
\|\mathbf G_n\|_F^2=\frac{\|\mathbf A_n\|_F^2}{w_n^2},
\qquad
\|\mathbf G_n^{(2)}\|_F^2=\frac{\|\mathbf A_n^{(2)}\|_F^2}{w_n^4},
\]
so
\[
\sqrt{\frac{\mathbb E\|\mathbf G_n\|_F^2}{\mathbb E\|\mathbf G_n^{(2)}\|_F^2}}
=
w_n
\sqrt{\frac{\mathbb E\|\mathbf A_n\|_F^2}{\mathbb E\|\mathbf A_n^{(2)}\|_F^2}}.
\]
Using the Frobenius-norm calculations in \eqref{eq:den}--\eqref{eq:ratio_bound}, applied to \(\mathbf A_n\) and \(\mathbf A_n^{(2)}\),
\[
\sqrt{\frac{\mathbb E\|\mathbf A_n\|_F^2}{\mathbb E\|\mathbf A_n^{(2)}\|_F^2}}
=
\frac{1}{\sqrt{d_n+\tfrac{d_n^3}{n}}}+o(1).
\]
Substituting this into the preceding display gives
\[
\left|
\frac{\tfrac{1}{n}\operatorname{Cov}(\mathbf{G}_n^{(2)}\mathbf{X},\mathbf{G}_n\mathbf{Y})}
     {\tfrac{1}{n}\operatorname{Var}(\mathbf{G}_n^{(2)}\mathbf{X})}
\right|
=
O\!\left(
\frac{w_n}{\sqrt{d_n + d_n^3/n}}
\right),
\]
which proves Proposition~\ref{prop:cov_bound_scaled}.
\end{proof}

Proposition~\ref{prop:cov_bound_scaled} shows that scaling rescales the variance and covariance of the friends-of-friends instrument at different powers of \(w_n\). In sparse regimes, where \(d_n=o(1)\) or \(d_n=O(1)\), the maximum degree may still grow even when the average degree is bounded. The scaled variance and covariance satisfy
\[
\operatorname{Var}(\mathbf G_n^{(2)}\mathbf X)
=
O\!\left(\frac{1}{w_n^4}\,\mathbb E\|\mathbf A_n^{(2)}\|_F^2\right),
\qquad
\operatorname{Cov}(\mathbf G_n^{(2)}\mathbf X,\mathbf G_n\mathbf Y)
=
O\!\left(\frac{1}{w_n^3}\,\mathbb E\|\mathbf A_n^{(2)}\|_F^2\right).
\]
Thus, scaling aligns the relevant growth rates and prevents explosive behavior of the first-stage ratio. Weak identification can still arise in the scaled specification, but it is governed by whether the information index \(k(n)\) grows fast enough relative to \(w_n\).

\subsection{Supplementary Results and Definitions}\label{SA2}

\begin{theoremnamed}[Theorem~1.1 of \citealp{krivelevich2001}]\label{thm:krivelevich_sudakov}
Let $G = G(n,p)$ be a random graph and let $\Delta$ be the maximum degree of $G$. Then almost surely the largest eigenvalue of the adjacency matrix of $G$ satisfies    
\[
\lambda_{1}(G) = (1 + o(1)) \max \{\sqrt{\Delta},\, np\},
\]
where $o(1) \to 0$ as $\max\{\sqrt{\Delta}, np\} \to \infty$.
\end{theoremnamed}
While their paper mentions random graphs in general, the discussion is focused on binomial or ER graphs, holding for all values of link formation probability $p(n)$.

\begin{definition}[Spectral Norm]\label{def:spec_norm}
Let $A \in \mathbb{R}^{m \times n}$. The \emph{spectral norm} of $A$ is defined as
\[
\|A\|_2 = \sup_{x \neq 0} \frac{\|Ax\|_2}{\|x\|_2} = \sqrt{\lambda_{\max}(A^\top A)},
\]
where $\lambda_{\max}(A^\top A)$ denotes the largest eigenvalue of $A^\top A$, and $\|\cdot\|_2$ on vectors is the usual Euclidean norm.
\end{definition}

\begin{definition}[Frobenius Norm]\label{def:frob_norm}
Let $A \in \mathbb{R}^{m \times n}$. The \emph{Frobenius norm} of $A$ is defined as
\[
\|A\|_F = \sqrt{\sum_{i=1}^{m}\sum_{j=1}^{n} |a_{ij}|^2}
= \sqrt{\operatorname{Tr}(A^\top A)}.
\]
Equivalently, if $\sigma_1, \dots, \sigma_{\min(m,n)}$ are the singular values of $A$, then
\[
\|A\|_F = \sqrt{\sum_{i} \sigma_i^2}.
\]
\end{definition}

\begin{definition}[Cauchy--Schwarz Inequality for Matrices]\label{def:CS-matrix}
Let $\mathbb{R}^{n\times n}$ be equipped with the \emph{Frobenius inner product}
\[
\langle A, B \rangle_F := \mathrm{Tr}(A^\top B),
\]
and the induced Frobenius norm
\[
\|A\|_F := \sqrt{\langle A, A \rangle_F} = \sqrt{\mathrm{Tr}(A^\top A)}.
\]
Then, for all $A,B \in \mathbb{R}^{n\times n}$, the \emph{Cauchy--Schwarz inequality} holds:
\begin{equation}\label{eq:CS-trace}
|\langle A,B\rangle_F| \;\le\; \|A\|_F \, \|B\|_F .
\end{equation}
Moreover, since $\mathrm{Tr}(AB) = \mathrm{Tr}(A^\top B)$ for symmetric $A$, this implies the general trace bound
\[
|\mathrm{Tr}(AB)| \;\le\; \|A\|_F \, \|B\|_F .
\]
\end{definition}

\begin{definition}[Asymptotic notation (deterministic)]\label{def:asymp}
Let $f,g:\mathbb{N}\to\mathbb{R}$ with $g(n)>0$ for all sufficiently large \(n\). Then
\begin{align*}
f(n)=o\!\big(g(n)\big)
&\;\Longleftrightarrow\;
\forall \varepsilon>0\ \exists n_0\ \forall n\ge n_0:\ |f(n)|\le \varepsilon\,g(n). \\[4pt]
f(n)=O\!\big(g(n)\big)
&\;\Longleftrightarrow\;
\exists C>0,\ \exists n_0\ \forall n\ge n_0:\ |f(n)|\le C\,g(n). \\[4pt]
f(n)=\omega\!\big(g(n)\big)
&\;\Longleftrightarrow\;
g(n)=o\!\big(f(n)\big). \\[4pt]
f(n)=\Omega\!\big(g(n)\big)
&\;\Longleftrightarrow\;
g(n)=O\!\big(f(n)\big). \\[4pt]
f(n)=\Theta\!\big(g(n)\big)
&\;\Longleftrightarrow\;
f(n)=O\!\big(g(n)\big)\ \text{and}\ f(n)=\Omega\!\big(g(n)\big).
\end{align*}
\noindent\textit{Additional comparisons:}
\begin{align*}
f(n)\sim g(n)
&\;\Longleftrightarrow\;
\lim_{n\to\infty}\frac{f(n)}{g(n)}=1. \\[4pt]
f(n)\asymp g(n)
&\;\Longleftrightarrow\;
\exists c_1,c_2>0,\ \exists n_0:\ c_1 g(n)\le f(n)\le c_2 g(n)\ \ \forall n\ge n_0. \\[4pt]
f(n)\ll g(n)
&\ \text{means}\ f(n)=O\!\big(g(n)\big),
\qquad
f(n)\gg g(n)\ \text{means}\ f(n)=\Omega\!\big(g(n)\big).
\end{align*}
\end{definition}

\newpage

\section{Additional Tables and Figures}\label{SA3}

Here, we present additional results, tables and figures referenced in the main text.

\subsection{Tables}\label{SA3_tables}

We first present additional Tables from the Monte Carlo simulations that were referenced in Section \ref{S4}.

\begin{table}[!htb]
\centering
\caption{Covariance Between Endogenous Variable $(\mathbf{GY})$ and Network-based Instrument $(\mathbf G^{(2)} \mathbf X )$.}
\begin{adjustbox}{max width=\textwidth}
\begin{tabular}{ll|llllll}
\toprule
 & & \multicolumn{6}{c}{\textbf{Covariance ($\mathbf{GY}$, $\mathbf{G}^{(2)} \mathbf{X}$)}} \\
\midrule
$\beta_0$ & $n \backslash d$ & 0.25 & 0.5 & 0.75 & 1 & 2 & 5 \\
\midrule

\multicolumn{8}{l}{\textbf{Panel A: Unscaled Model}} \\

\multirow{4}{*}{0.4666}
& 250  & 0.380  & -2.101 & -10.246 & -8.244  & -43.062 & -27.252 \\
& 500  & 0.274  & -0.384 & 44.601  & -11.667 & -48.312 & -21.111 \\
& 1000 & 0.618  & 15.352 & -10.376 & -8.785  & 16.535  & -59.370 \\
& 2000 & 0.779  & 5.660  & 2.979   & -0.972  & -6.548  & -20.346 \\

\noalign{\medskip}

\multirow{4}{*}{0.95}
& 250  & -0.316 & -2.066 & -4.748  & -3.004  & -14.685 & -6.357 \\
& 500  & -0.341 & -1.220 & -2.963  & -3.773  & -2.526  & 4.635  \\
& 1000 & -0.433 & -1.312 & -3.511  & -3.617  & -0.879  & -4.838 \\
& 2000 & -0.521 & -1.711 & -2.658  & -3.966  & -17.429 & 2.318  \\

\midrule
\multicolumn{8}{l}{\textbf{Panel B: Scaled Model}} \\

\multirow{4}{*}{0.4666}
& 250  & 0.024 & 0.139 & 0.168 & 0.303 & 0.638 & 0.507 \\
& 500  & 0.019 & 0.029 & 0.118 & 0.192 & 0.399 & 0.537 \\
& 1000 & 0.038 & 0.066 & 0.154 & 0.103 & 0.770 & 0.512 \\
& 2000 & 0.043 & 0.076 & 0.097 & 0.116 & 0.442 & 0.563 \\

\noalign{\medskip}

\multirow{4}{*}{0.95}
& 250  & 0.190 & 2.669 & 0.669 & -0.657 & 39.013 & -1.942 \\
& 500  & 0.135 & 0.751 & -2.798 & 1.088 & -2.575 & -1.515 \\
& 1000 & 0.525 & -0.498 & -0.993 & -0.055 & -0.632 & -1.515 \\
& 2000 & -0.580 & 0.190 & 0.379 & -0.626 & -9.527 & -1.512 \\

\bottomrule
\end{tabular}
\end{adjustbox}
\label{tab:covariance}
\end{table}

\clearpage
\begin{landscape}
\begin{table}[p]
\centering
\scriptsize
\setlength{\tabcolsep}{2.5pt}
\renewcommand{\arraystretch}{1.05}
\caption{Scaled Model: TSLS Estimates, First-Stage Diagnostics, and Empirical Coverage Probabilities. Panel A reports TSLS estimates for specification~\eqref{struc_model} under the scaled network with scaling factor $w_n=\max(d_n, \sqrt{\Delta_n})$, together with the correlation between \(\mathbf{GY}\) and \(\mathbf G^{(2)}\mathbf X\), and the first-stage \(F\)-statistic. Panel B reports coverage probabilities for conventional \(t\)-test using homoskedastic standard errors,  \(t\)-test with \cite{kojevnikov2021limit} network-dependent standard errors, and Anderson--Rubin inference using homoskedastic standard errors.}
\label{tab:scaled_diagnostics_coverage}

\begin{adjustbox}{width=\linewidth,totalheight=0.88\textheight,keepaspectratio}
\begin{tabular}{ll|llllll|llllll|llllll}
\toprule
\multicolumn{20}{l}{\textbf{Panel A: TSLS Estimates and First-Stage Diagnostics}} \\
\midrule
 & & \multicolumn{6}{c|}{\textbf{TSLS Estimate}} 
 & \multicolumn{6}{c|}{\textbf{Correlation ($\mathbf{GY}$, $\mathbf{G}^2\mathbf{X}$)}} 
 & \multicolumn{6}{c}{\textbf{First-Stage F-statistic}} \\
\midrule
$\beta_0$ & $n \backslash d$ 
& 0.25 & 0.5 & 0.75 & 1 & 2 & 5 
& 0.25 & 0.5 & 0.75 & 1 & 2 & 5 
& 0.25 & 0.5 & 0.75 & 1 & 2 & 5 \\
\midrule

\multirow{4}{*}{0.4666} 
& 250  & 0.775 & 0.440 & 0.450 & 0.450 & 0.457 & 0.462 & 0.261 & 0.444 & 0.555 & 0.484 & 0.687 & 0.752 & 14.4 & 48.6 & 79.0 & 62.3 & 147.1 & 283.9 \\
& 500  & 1.330 & 0.429 & 0.449 & 0.456 & 0.462 & 0.465 & 0.211 & 0.263 & 0.424 & 0.488 & 0.574 & 0.743 & 15.3 & 47.1 & 66.9 & 128.8 & 211.0 & 455.6 \\
& 1000 & 0.453 & 0.466 & 0.465 & 0.464 & 0.466 & 0.467 & 0.357 & 0.335 & 0.531 & 0.491 & 0.647 & 0.758 & 124.3 & 153.5 & 325.0 & 236.0 & 591.3 & 1006.7 \\
& 2000 & 0.463 & 0.464 & 0.464 & 0.467 & 0.468 & 0.466 & 0.394 & 0.429 & 0.430 & 0.498 & 0.651 & 0.757 & 322.6 & 332.3 & 398.3 & 583.0 & 1095.800 & 2231.7 \\

\noalign{\medskip}

\multirow{4}{*}{0.95}   
& 250  & 0.943 & 0.950 & 0.946 & 0.956 & 0.950 & 0.950 & 0.352 & 0.289 & 0.273 & -0.078 & 0.122 & -0.726 & 28.98 & 19.16 & 12.17 & 2.14 & 1.28 & 225.8 \\
& 500  & 0.962 & 0.949 & 0.951 & 0.949 & 0.950 & 0.950 & 0.144 & 0.375 & -0.248 & 0.060 & -0.200 & -0.655 & 9.5 & 84.18 & 25.99 & 0.85 & 18.35 & 288.98 \\
& 1000 & 0.949 & 0.951 & 0.950 & 0.999 & 0.969 & 0.950 & 0.371 & -0.178 & -0.229 & -0.008 & -0.028 & -0.648 & 138.85 & 32.98 & 47.94 & 5.37 & 1.43 & 583.3 \\
& 2000 & 0.951 & 0.950 & 0.949 & 0.950 & 0.950 & 0.950 & -0.305 & 0.065 & 0.147 & -0.184 & -0.170 & -0.643 & 188.5 & 14.0 & 33.2 & 66.2 & 32.8 & 1166.6 \\

\midrule
\multicolumn{20}{l}{\textbf{Panel B: Coverage Probabilities}} \\
\midrule

 & & \multicolumn{6}{c|}{\textbf{\(t\)-test (Homoskedastic SE)}} 
 & \multicolumn{6}{c|}{\textbf{\(t\)-test (\cite{kojevnikov2021limit} SE)}} 
 & \multicolumn{6}{c}{\textbf{AR (Homoskedastic SE)}} \\
\midrule
$\beta_0$ & $n \backslash d$ 
& 0.25 & 0.5 & 0.75 & 1 & 2 & 5 
& 0.25 & 0.5 & 0.75 & 1 & 2 & 5
& 0.25 & 0.5 & 0.75 & 1 & 2 & 5 \\
\midrule

\multirow{4}{*}{0.4666}
& 250  & 0.986 & 0.968 & 0.939 & 0.949 & 0.957 & 0.958
        & 0.978 & 0.962 & 0.936 & 0.943 & 0.951 & 0.955
        & 0.956 & 0.961 & 0.927 & 0.940 & 0.953 & 0.958 \\
& 500  & 0.984 & 0.965 & 0.958 & 0.948 & 0.950 & 0.953
        & 0.983 & 0.942 & 0.957 & 0.945 & 0.951 & 0.949
        & 0.956 & 0.942 & 0.953 & 0.945 & 0.948 & 0.954 \\
& 1000 & 0.935 & 0.959 & 0.950 & 0.961 & 0.961 & 0.962
        & 0.931 & 0.955 & 0.949 & 0.956 & 0.960 & 0.957
        & 0.929 & 0.959 & 0.948 & 0.960 & 0.961 & 0.962 \\
& 2000 & 0.945 & 0.956 & 0.944 & 0.965 & 0.941 & 0.949
        & 0.953 & 0.958 & 0.949 & 0.961 & 0.940 & 0.950
        & 0.942 & 0.955 & 0.942 & 0.963 & 0.942 & 0.947 \\

\noalign{\medskip}

\multirow{4}{*}{0.95}
& 250  & 0.968 & 0.973 & 0.982 & 1.000 & 1.000 & 0.948
        & 0.960 & 0.995 & 0.993 & 1.000 & 1.000 & 0.946
        & 0.958 & 0.944 & 0.953 & 0.946 & 0.950 & 0.948 \\
& 500  & 0.987 & 0.970 & 0.973 & 1.000 & 0.975 & 0.948
        & 1.000 & 0.977 & 0.998 & 1.000 & 1.000 & 0.943
        & 0.931 & 0.963 & 0.953 & 0.936 & 0.948 & 0.946 \\
& 1000 & 0.956 & 0.952 & 0.962 & 0.996 & 1.000 & 0.948
        & 0.975 & 0.998 & 0.990 & 1.000 & 1.000 & 0.946
        & 0.952 & 0.935 & 0.951 & 0.949 & 0.957 & 0.946 \\
& 2000 & 0.948 & 0.991 & 0.970 & 0.958 & 0.971 & 0.946
        & 0.983 & 1.000 & 0.990 & 0.989 & 0.989 & 0.948
        & 0.947 & 0.940 & 0.950 & 0.948 & 0.961 & 0.947 \\

\bottomrule
\end{tabular}
\end{adjustbox}
\end{table}
\end{landscape}

\begin{table}[!htb]
\centering
\caption{Anderson--Rubin Test Coverage Probabilities at the 95\% Nominal Level using \cite{kojevnikov2021limit}  Standard Errors.}
\begin{adjustbox}{max width=\textwidth}
\begin{tabular}{ll|llllll}
\toprule
 & & \multicolumn{6}{c}{\makecell{\textbf{AR Test Coverage}\\\textbf{(\cite{kojevnikov2021limit} SEs)}}} \\
\midrule
$\beta_0$ & $n \backslash d$ & 0.25 & 0.5 & 0.75 & 1 & 2 & 5 \\
\midrule

\multicolumn{8}{l}{\textbf{Panel A: Unscaled Model}} \\

\multirow{4}{*}{0.4666}
& 250  & 0.882 & 0.926 & 0.921 & 0.931 & 0.948 & 0.945 \\
& 500  & 0.920 & 0.898 & 0.937 & 0.949 & 0.936 & 0.945 \\
& 1000 & 0.935 & 0.940 & 0.953 & 0.952 & 0.942 & 0.949 \\
& 2000 & 0.956 & 0.959 & 0.952 & 0.950 & 0.940 & 0.957 \\

\noalign{\medskip}

\multirow{4}{*}{0.95}
& 250  & 0.888 & 0.935 & 0.924 & 0.953 & 0.943 & 0.946 \\
& 500  & 0.925 & 0.929 & 0.936 & 0.943 & 0.945 & 0.946 \\
& 1000 & 0.936 & 0.937 & 0.956 & 0.944 & 0.943 & 0.962 \\
& 2000 & 0.947 & 0.946 & 0.943 & 0.953 & 0.960 & 0.957 \\

\midrule
\multicolumn{8}{l}{\textbf{Panel B: Scaled Model}} \\

\multirow{4}{*}{0.4666}
& 250  & 0.895 & 0.928 & 0.914 & 0.933 & 0.943 & 0.952 \\
& 500  & 0.928 & 0.918 & 0.939 & 0.935 & 0.943 & 0.948 \\
& 1000 & 0.921 & 0.951 & 0.945 & 0.956 & 0.959 & 0.959 \\
& 2000 & 0.945 & 0.956 & 0.946 & 0.961 & 0.942 & 0.948 \\

\noalign{\medskip}

\multirow{4}{*}{0.95}
& 250  & 0.878 & 0.914 & 0.941 & 0.931 & 0.943 & 0.943 \\
& 500  & 0.906 & 0.931 & 0.944 & 0.936 & 0.944 & 0.942 \\
& 1000 & 0.935 & 0.932 & 0.950 & 0.949 & 0.956 & 0.947 \\
& 2000 & 0.943 & 0.933 & 0.947 & 0.948 & 0.957 & 0.950 \\

\bottomrule
\end{tabular}
\end{adjustbox}
\label{tab:ar-kojev-cov}
\end{table}

\begin{table}[!htb]
\centering
\caption{Average Confidence Interval Length for 95\% Coverage of $\beta$ 
(Homoskedastic Standard Errors and \citealp{kojevnikov2021limit} Standard Errors).}
\begin{adjustbox}{max width=\textwidth}
\begin{tabular}{ll|llllll|llllll}
\toprule
 & 
 & \multicolumn{6}{c|}{\textbf{CI Length (Standard Homoskedastic SE)}} 
 & \multicolumn{6}{c}{\textbf{CI Length (\cite{kojevnikov2021limit} SE)}} \\
\midrule
$\beta_0$ & $n \backslash d$ 
& 0.25 & 0.5 & 0.75 & 1 & 2 & 5
& 0.25 & 0.5 & 0.75 & 1 & 2 & 5 \\
\midrule

\multicolumn{14}{l}{\textbf{Panel A: Unscaled Model}} \\

\multirow{4}{*}{0.4666} 
& 250  & 0.528 & 117   & 0.058 & 0.080 & 0.061 & 56 
       & 0.538 & 305   & 0.058 & 0.081 & 0.059 & 61 \\
& 500  & 0.561 & 245   & 0.010 & 0.086 & 1.384 & 54 
       & 0.680 & 778   & 0.011 & 0.092 & 1.274 & 51 \\
& 1000 & 0.135 & 0.147 & 0.022 & 0.033 & 163   & 0.166
       & 0.133 & 0.418 & 0.023 & 0.033 & 169   & 0.159 \\
& 2000 & 0.072 & 0.021 & 3.696 & 5.270 & 0.039 & 111
       & 0.071 & 0.024 & 7.207 & 7.636 & 0.040 & 103 \\

\noalign{\medskip}

\multirow{4}{*}{0.95} 
& 250  & 0.284 & 0.108 & 0.084 & 0.215 & 269   & 30
       & 0.291 & 0.102 & 0.081 & 0.208 & 242   & 28 \\
& 500  & 0.143 & 0.150 & 0.110 & 0.106 & 252   & 64
       & 0.135 & 0.137 & 0.106 & 0.105 & 246   & 59 \\
& 1000 & 0.121 & 0.100 & 0.055 & 0.072 & 790250 & 896132
       & 0.116 & 0.098 & 0.054 & 0.071 & 763355 & 857568 \\
& 2000 & 0.066 & 0.040 & 0.046 & 0.050 & 309   & 14
       & 0.065 & 0.040 & 0.045 & 0.051 & 271   & 14 \\

\midrule
\multicolumn{14}{l}{\textbf{Panel B: Scaled Model}} \\

\multirow{4}{*}{0.4666} 
& 250  & 72.89  & 0.91 & 0.70 & 0.60 & 0.36 & 0.35
       & 96.21  & 0.90 & 0.69 & 0.59 & 0.35 & 0.34 \\
& 500  & 921.26 & 1.43 & 0.76 & 0.47 & 0.28 & 0.24
       & 1975.95& 1.39 & 0.75 & 0.46 & 0.28 & 0.24 \\
& 1000 & 0.77   & 0.50 & 0.34 & 0.44 & 0.15 & 0.16
       & 0.77   & 0.50 & 0.34 & 0.44 & 0.15 & 0.16 \\
& 2000 & 0.46   & 0.39 & 0.31 & 0.28 & 0.15 & 0.11
       & 0.46   & 0.38 & 0.31 & 0.28 & 0.14 & 0.11 \\

\noalign{\medskip}

\multirow{4}{*}{0.95} 
& 250  & 0.43 & 0.53 & 0.31 & 1.43  & 0.01 & 0.08
       & 0.54 & 1.26 & 0.50 & 1.96  & 0.01 & 0.08 \\
& 500  & 9.28 & 0.06 & 0.03 & 0.23  & 0.04 & 0.08
       & 16.59& 0.06 & 0.04 & 0.28  & 0.04 & 0.08 \\
& 1000 & 0.05 & 0.07 & 0.05 & 13.81 & 66.93& 0.05
       & 0.05 & 0.08 & 0.05 & 36.22 & 95.16& 0.05 \\
& 2000 & 0.03 & 23.54& 0.08 & 0.05  & 0.01 & 0.04
       & 0.03 & 90.99& 0.09 & 0.05  & 0.01 & 0.04 \\

\bottomrule
\end{tabular}
\end{adjustbox}
\label{tab:ci_length_TSLS}
\end{table}

\begin{table}[!htb]
\centering
\caption{Average Length and Percentage of Infinite 95\% Anderson--Rubin Confidence Intervals for \(\beta\), Using Homoskedastic Standard Errors. The length is computed only among finite confidence intervals, while the second panel reports the percentage of simulations in which the interval is infinite.}
\begin{adjustbox}{max width=\textwidth}
\begin{tabular}{ll|llllll|llllll}
\toprule
 & 
 & \multicolumn{6}{c|}{\textbf{CI Length (AR Homoskedastic, finite)}} 
 & \multicolumn{6}{c}{\textbf{\% CI Infinite (AR Homoskedastic)}} \\
\midrule
$\beta_0$ & $n \backslash d$ 
& 0.25 & 0.5 & 0.75 & 1 & 2 & 5 
& 0.25 & 0.5 & 0.75 & 1 & 2 & 5 \\
\midrule

\multicolumn{14}{l}{\textbf{Panel A: Unscaled Model}} \\

\multirow{4}{*}{0.4666}
& 250  & 0.641 & 0.482 & 0.064 & 0.160 & 0.135 & 0.458
       & 0.010 & 0.606 & 0.001 & 0.424 & 0.870 & 0.976 \\
& 500  & 0.585 & 0.686 & 0.011 & 0.266 & 0.593 & 0.293
       & 0.058 & 0.512 & 0.003 & 0.907 & 0.954 & 0.927 \\
& 1000 & 0.136 & 0.012 & 0.023 & 0.036 & 0.136 & 0.112
       & 0.000 & 0.005 & 0.000 & 0.000 & 0.957 & 0.445 \\
& 2000 & 0.072 & 0.022 & 0.168 & 0.692 & 0.520 & 0.216
       & 0.000 & 0.002 & 0.820 & 0.898 & 0.889 & 0.999 \\

\noalign{\medskip}

\multirow{4}{*}{0.95}
& 250  & 1.725     & 0.133 & 0.099 & 0.570 & 0.820 & NA
       & 0.341 & 0.000 & 0.000 & 0.136 & 0.674 & 1.000 \\
& 500  & 0.184 & 0.243 & 0.129 & 0.140 & 0.786 & 1.236
       & 0.000 & 0.009 & 0.000 & 0.001 & 0.864 & 0.999 \\
& 1000 & 0.141 & 0.118 & 0.057 & 0.077 & 0.812 & 88.7
       & 0.000 & 0.000 & 0.000 & 0.000 & 0.925 & 0.989 \\
& 2000 & 0.070 & 0.041 & 0.047 & 0.172 & 0.256 & 0.403
       & 0.000 & 0.000 & 0.000 & 0.014 & 0.999 & 0.996 \\

\midrule
\multicolumn{14}{l}{\textbf{Panel B: Scaled Model}} \\

\multirow{4}{*}{0.4666}
& 250  & 5.063 & 1.021 & 0.728 & 0.635 & 0.365 & 0.349
       & 0.149 & 0.000 & 0.000 & 0.000 & 0.000 & 0.000 \\
& 500  & 4.110 & 1.641 & 0.789 & 0.476 & 0.285 & 0.243
       & 0.157 & 0.004 & 0.000 & 0.000 & 0.000 & 0.000 \\
& 1000 & 0.791 & 0.511 & 0.347 & 0.450 & 0.150 & 0.164
       & 0.000 & 0.000 & 0.000 & 0.000 & 0.000 & 0.000 \\
& 2000 & 0.465 & 0.389 & 0.315 & 0.281 & 0.146 & 0.107
       & 0.000 & 0.000 & 0.000 & 0.000 & 0.000 & 0.000 \\

\noalign{\medskip}

\multirow{4}{*}{0.95}
& 250  & 0.366 & 0.072 & 0.310 & 0.808 & NA    & 0.086
       & 0.014 & 0.105 & 0.100 & 0.856 & 1.000 & 0.000 \\
& 500  & 0.574 & 0.063 & 0.056 & 0.244 & 0.047 & 0.078
       & 0.235 & 0.000 & 0.006 & 0.981 & 0.000 & 0.000 \\
& 1000 & 0.054 & 0.079 & 0.052 & 0.604 & 0.361 & 0.050
       & 0.000 & 0.003 & 0.000 & 0.557 & 0.897 & 0.000 \\
& 2000 & 0.032 & 0.178 & 0.090 & 0.048 & 0.008 & 0.037
       & 0.000 & 0.468 & 0.000 & 0.000 & 0.000 & 0.000 \\

\bottomrule
\end{tabular}
\end{adjustbox}
\label{tab:ci_length_ar}
\end{table}

\begin{table}[!htb]
\centering
\caption{Average Length and Percentage of Infinite 95\% Anderson--Rubin Confidence Intervals for \(\beta\), Using \cite{kojevnikov2021limit} Standard Errors. The length is computed only among finite confidence intervals, while the second panel reports the percentage of simulations in which the interval is infinite.}
\begin{adjustbox}{max width=\textwidth}
\begin{tabular}{ll|llllll|llllll}
\toprule
 & 
 & \multicolumn{6}{c|}{\textbf{CI Length (AR \cite{kojevnikov2021limit}, finite)}} 
 & \multicolumn{6}{c}{\textbf{\% CI Infinite (AR \cite{kojevnikov2021limit})}} \\
\midrule
$\beta_0$ & $n \backslash d$ 
& 0.25 & 0.5 & 0.75 & 1 & 2 & 5 
& 0.25 & 0.5 & 0.75 & 1 & 2 & 5 \\
\midrule

\multicolumn{14}{l}{\textbf{Panel A: Unscaled Model}} \\

\multirow{4}{*}{0.4666}
& 250  & 1.573     & 0.580 & 0.165 & 0.405 & 0.271 & 1.039
       & 0.274 & 0.991 & 0.037 & 0.605 & 0.871 & 0.981 \\
& 500  & 1.839     & 0.564 & NA    & 4.194     & 2.805     & 0.563
       & 0.313 & 0.995 & 1.000 & 0.992 & 0.969 & 0.907 \\
& 1000 & 0.144 & 0.074 & 0.027 & 0.042 & 0.185 & 0.150
       & 0.000 & 0.658 & 0.000 & 0.002 & 0.967 & 0.426 \\
& 2000 & 0.073 & 0.070 & 0.755 & 0.311 & 0.208 & 0.245
       & 0.000 & 0.456 & 0.997 & 0.994 & 0.910 & 0.998 \\

\noalign{\medskip}

\multirow{4}{*}{0.95}
& 250  & 4.387     & 0.116 & 0.092 & 0.489 & 0.597 & NA
       & 0.997 & 0.000 & 0.000 & 0.094 & 0.717 & 1.000 \\
& 500  & 0.174 & 0.169 & 0.115 & 0.837 & 0.827 & 1.438
       & 0.000 & 0.000 & 0.000 & 0.008 & 0.845 & 0.999 \\
& 1000 & 0.134 & 0.108 & 0.056 & 0.074 & 1.086     & 4.773
       & 0.000 & 0.000 & 0.000 & 0.000 & 0.926 & 0.983 \\
& 2000 & 0.068 & 0.041 & 0.046 & 0.100 & 0.201 & 2.138
       & 0.000 & 0.000 & 0.000 & 0.165 & 0.999 & 0.994 \\

\midrule
\multicolumn{14}{l}{\textbf{Panel B: Scaled Model}} \\

\multirow{4}{*}{0.4666}
& 250  & 11.957 & 1.556 & 0.793 & 0.653 & 0.359 & 0.342
       & 0.567  & 0.019 & 0.001 & 0.000 & 0.000 & 0.000 \\
& 500  & 5.847  & 1.507 & 0.804 & 0.480 & 0.282 & 0.240
       & 0.477  & 0.013 & 0.000 & 0.000 & 0.000 & 0.000 \\
& 1000 & 0.914  & 0.521 & 0.347 & 0.452 & 0.149 & 0.164
       & 0.001  & 0.000 & 0.000 & 0.000 & 0.000 & 0.000 \\
& 2000 & 0.479  & 0.394 & 0.317 & 0.280 & 0.145 & 0.107
       & 0.000  & 0.000 & 0.000 & 0.000 & 0.000 & 0.000 \\

\noalign{\medskip}

\multirow{4}{*}{0.95}
& 250  & 1.634 & 0.160 & 0.726 & 2.058 & NA    & 0.084
       & 0.255 & 0.710 & 0.769 & 0.999 & 1.000 & 0.000 \\
& 500  & 1.437 & 0.083 & 1.120 & NA    & 0.280 & 0.077
       & 0.846 & 0.003 & 0.966 & 1.000 & 0.316 & 0.000 \\
& 1000 & 0.064 & 0.281 & 0.082 & 0.673 & 0.988 & 0.050
       & 0.001 & 0.180 & 0.009 & 0.983 & 0.979 & 0.000 \\
& 2000 & 0.060 & 0.233 & 0.157 & 0.072 & 0.010 & 0.037
       & 0.028 & 0.948 & 0.040 & 0.001 & 0.002 & 0.000 \\

\bottomrule
\end{tabular}
\end{adjustbox}
\label{tab:ci_length_ar_koj}
\end{table}

\clearpage

\subsection{Figures}\label{SA3_figures}

Here, we present figures referenced in the numerical exercise in Section \ref{S3} and the Monte Carlo simulations in Section \ref{S4}.

\begin{figure}[!ht]
          \centering
    \captionsetup[subfigure]{justification=centering}
    \begin{subfigure}[b]{0.9\textwidth}
        \centering
        \includegraphics[width=\textwidth]{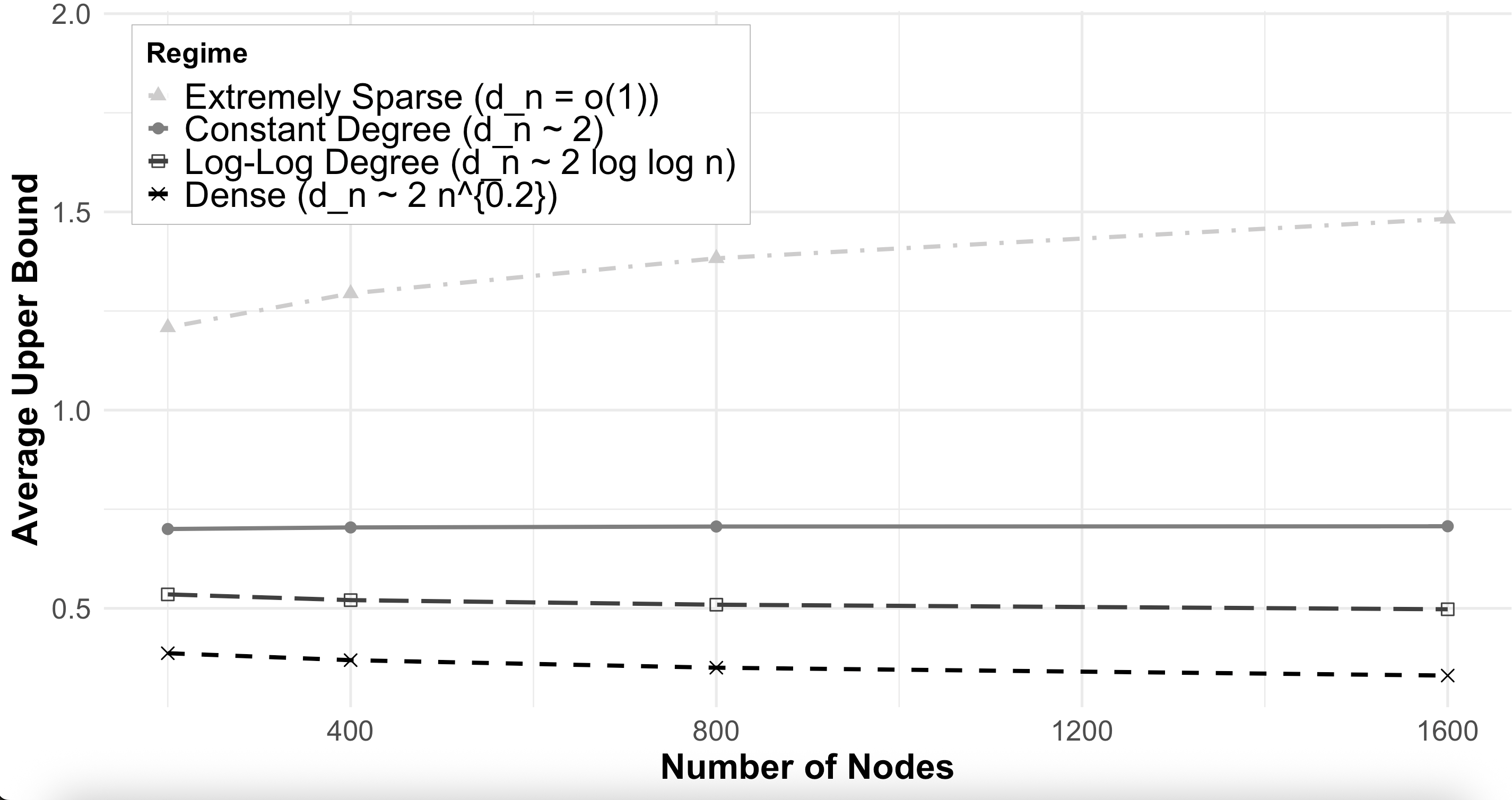}
        \caption{Upper Bound: Unscaled Model}
        \label{fig:upper_non-norm}
    \end{subfigure}
    \hfill
    \begin{subfigure}[b]{0.9\textwidth}
        \centering
        \includegraphics[width=\textwidth]{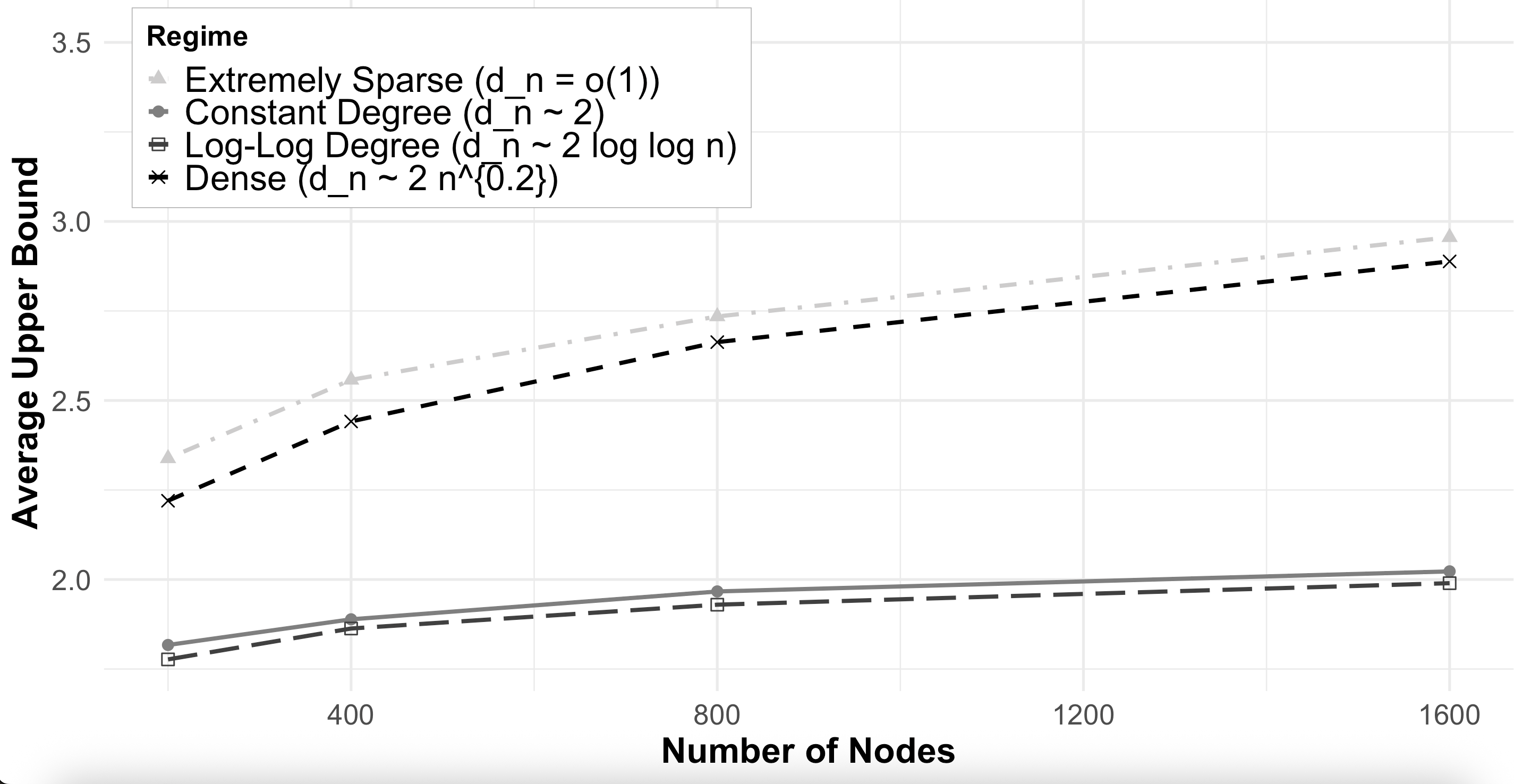}
        \caption{Upper Bound: Scaled Model}
        \label{fig:upper_norm}
        \end{subfigure}
    \caption{Numerical Exercise Section \ref{S3}: Growth Rates of Upper bounds on the variance-normalized covariance for the unscaled and scaled adjacency matrix. Please note the different y-axis scales, chosen for ease of presentation.}
    \label{fig:upper-bounds}
\end{figure}

\begin{figure}[!ht]
          \centering
    \captionsetup[subfigure]{justification=centering}
    \begin{subfigure}[b]{0.8\textwidth}
        \centering
        \includegraphics[width=\textwidth]{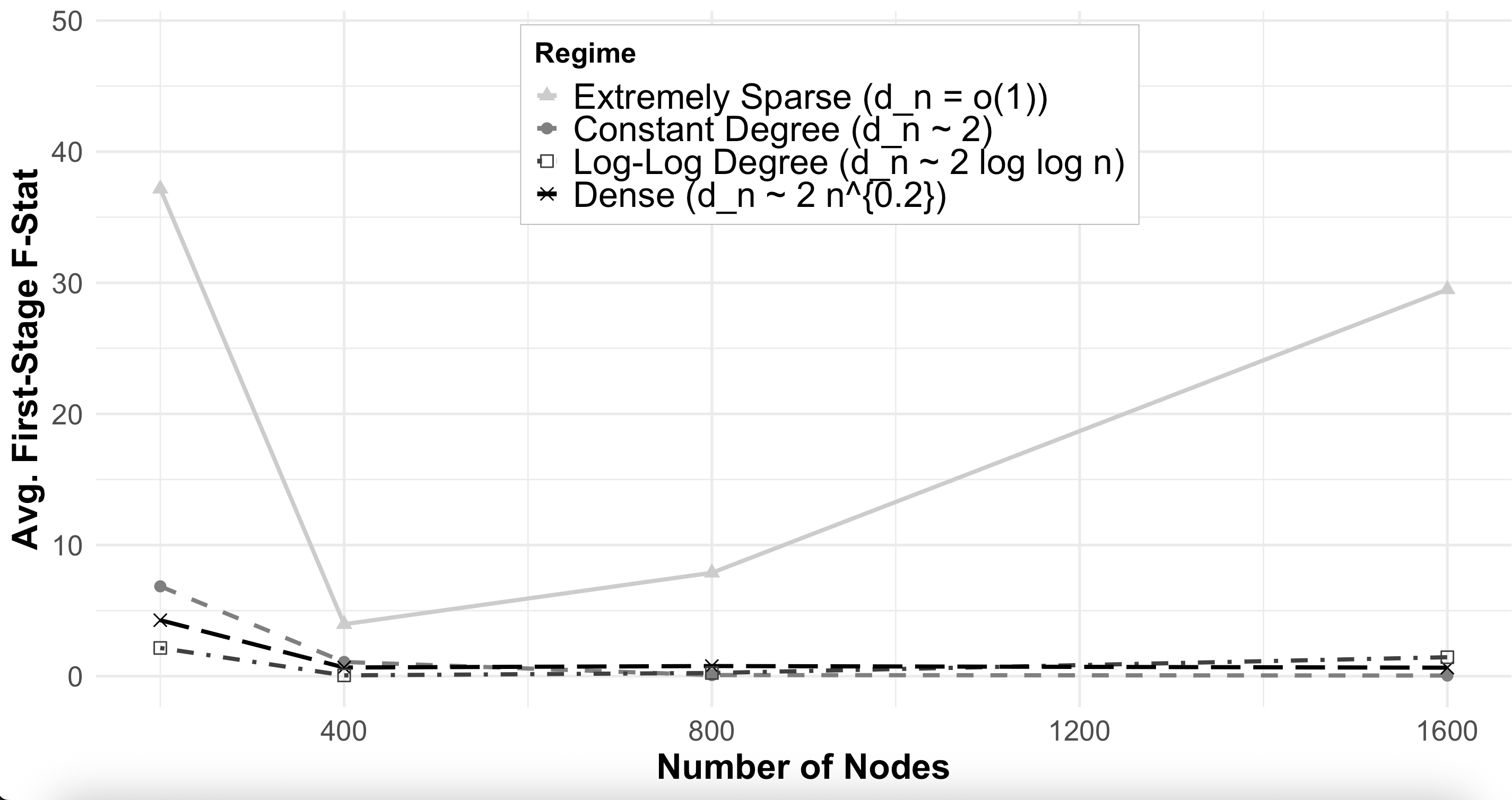}
        \caption{F-Statistic: Unscaled Model}
        \label{fig:Fstat_non-norm}
    \end{subfigure}
    \hfill
    \begin{subfigure}[b]{0.8\textwidth}
        \centering
        \includegraphics[width=\textwidth]{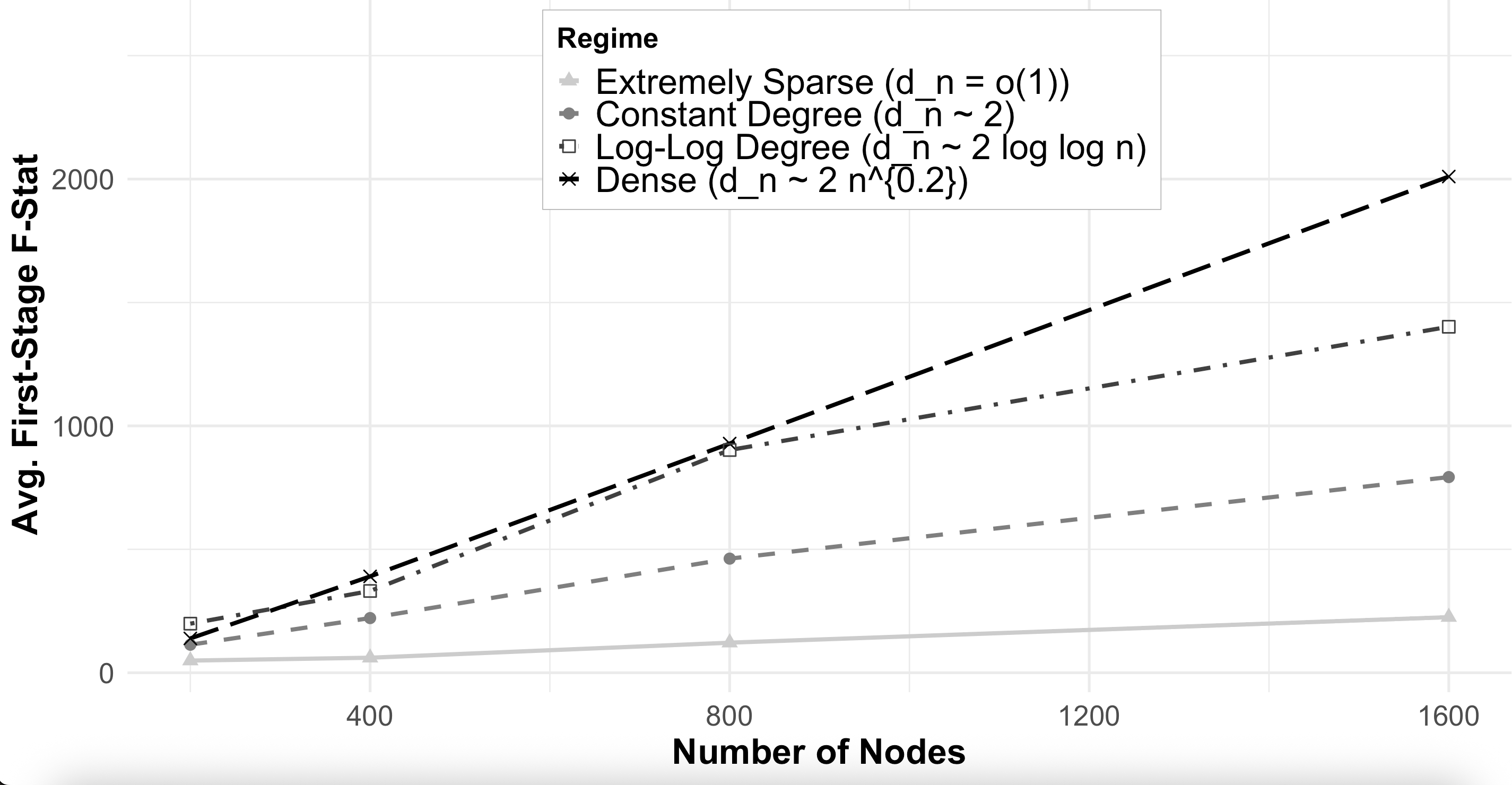}
        \caption{F-statistic: Scaled Model}
        \label{fig:Fstat_norm}
        \end{subfigure}
    \caption{Monte Carlo Simulations Section \ref{S4}: First-stage F-statistic for the unscaled and scaled adjacency matrix.}
    \label{fig:Fstat}
\end{figure}

\begin{figure}[!ht]
          \centering
    \captionsetup[subfigure]{justification=centering}
    \begin{subfigure}[b]{0.9\textwidth}
        \centering
        \includegraphics[width=\textwidth]{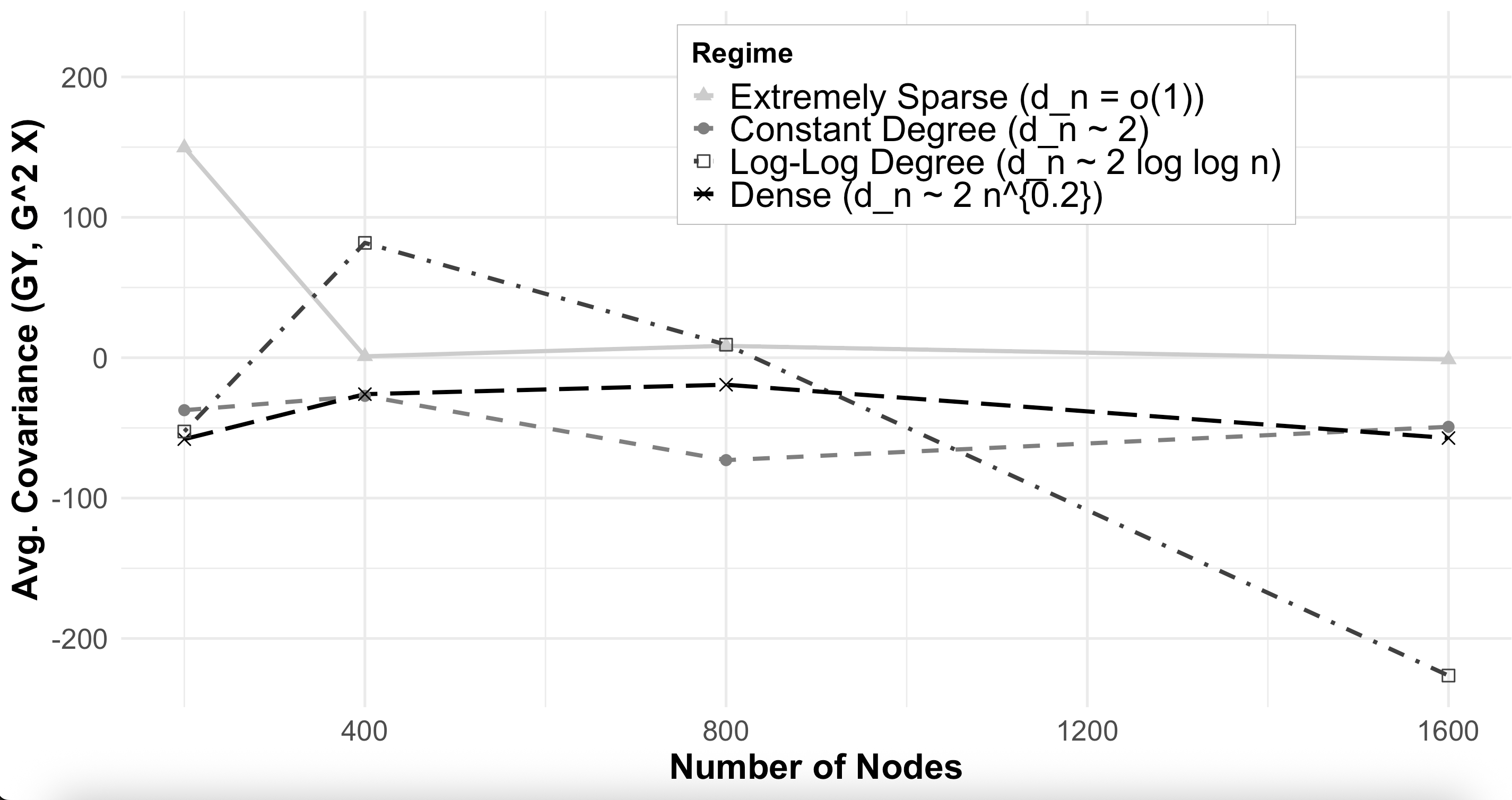}
        \caption{First-stage Relevance: Unscaled Model}
        \label{fig:Cov_non-norm}
    \end{subfigure}
    \hfill
    \begin{subfigure}[b]{0.9\textwidth}
        \centering
        \includegraphics[width=\textwidth]{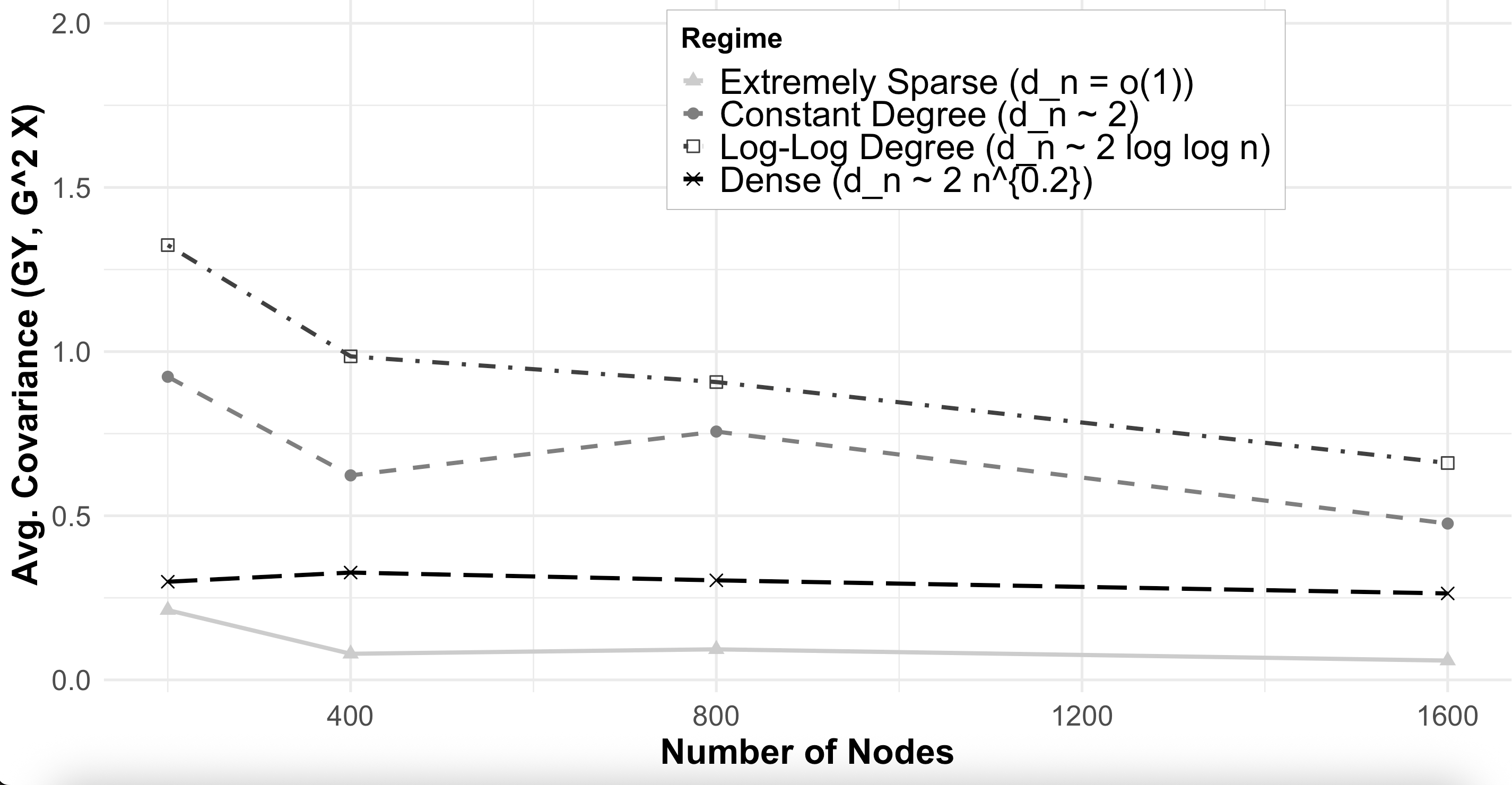}
        \caption{First-stage Relevance: Scaled Model}
        \label{fig:Cov_norm}
        \end{subfigure}
    \caption{Monte Carlo Simulations Section \ref{S4}: Covariance between $\mathbf{GY}$ and $\mathbf{G^{(2)}X}$ for the unscaled and scaled adjacency matrix.}
        \label{fig:Cov}
\end{figure}

\begin{figure}[!ht]
          \centering
    \captionsetup[subfigure]{justification=centering}
    \begin{subfigure}[b]{0.9\textwidth}
        \centering
        \includegraphics[width=\textwidth]{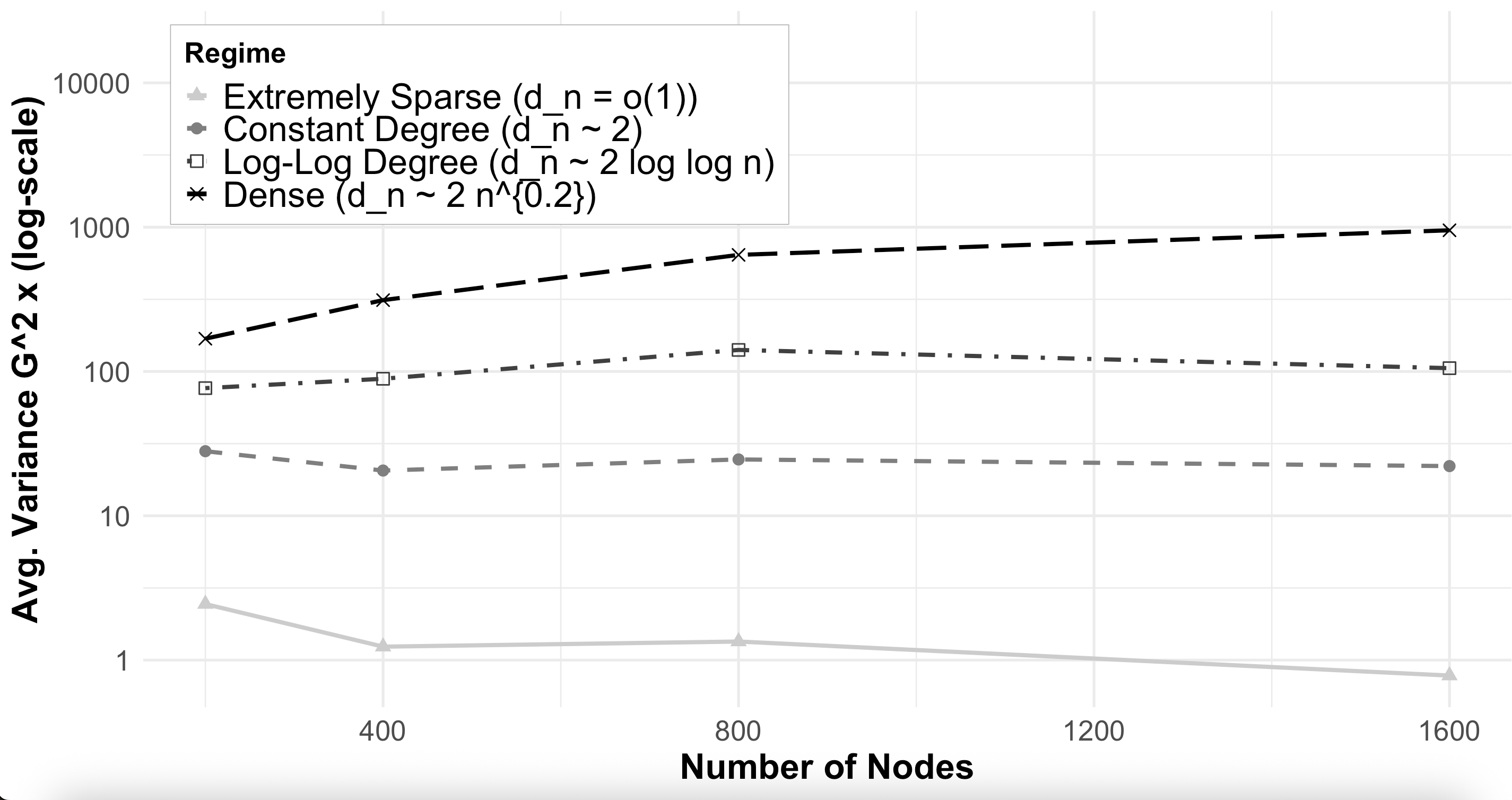}
        \caption{Variance of Instrument: Unscaled Model}
        \label{fig:Var_non-norm}
    \end{subfigure}
    \hfill
    \begin{subfigure}[b]{0.9\textwidth}
        \centering
        \includegraphics[width=\textwidth]{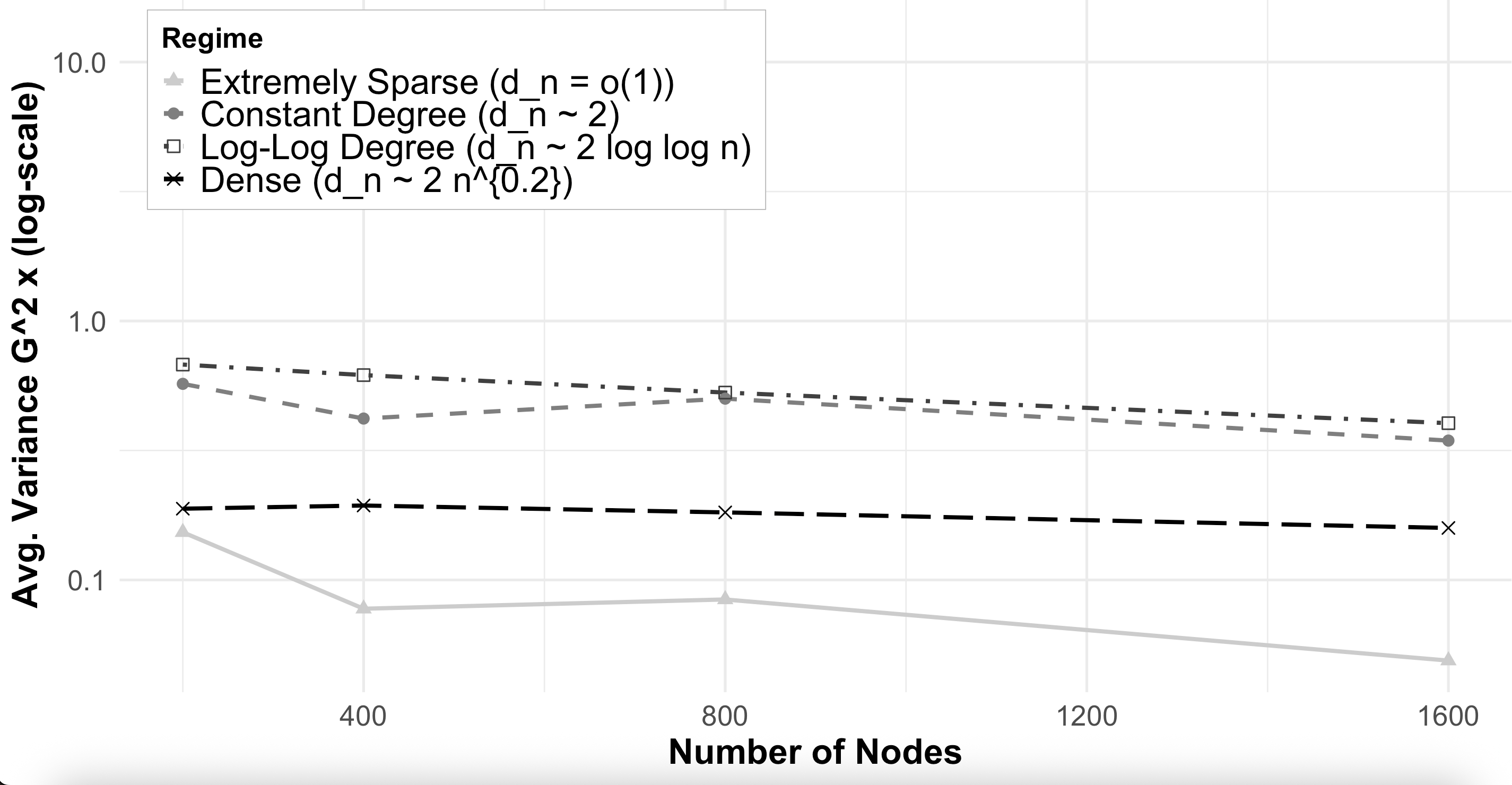}
        \caption{Variance of Instrument: Scaled Model}
        \label{fig:Var_norm}
        \end{subfigure}
    \caption{Monte Carlo Simulations Section \ref{S4}: Variance of instrument $\mathbf{G^{(2)}X}$ for the unscaled and scaled adjacency matrix.}
    \label{fig:Var}
\end{figure}

\begin{figure}[ht]
    \centering
    \captionsetup[subfigure]{justification=centering}
    \begin{subfigure}[b]{0.49\textwidth}
    \centering
    \includegraphics[width=\linewidth]{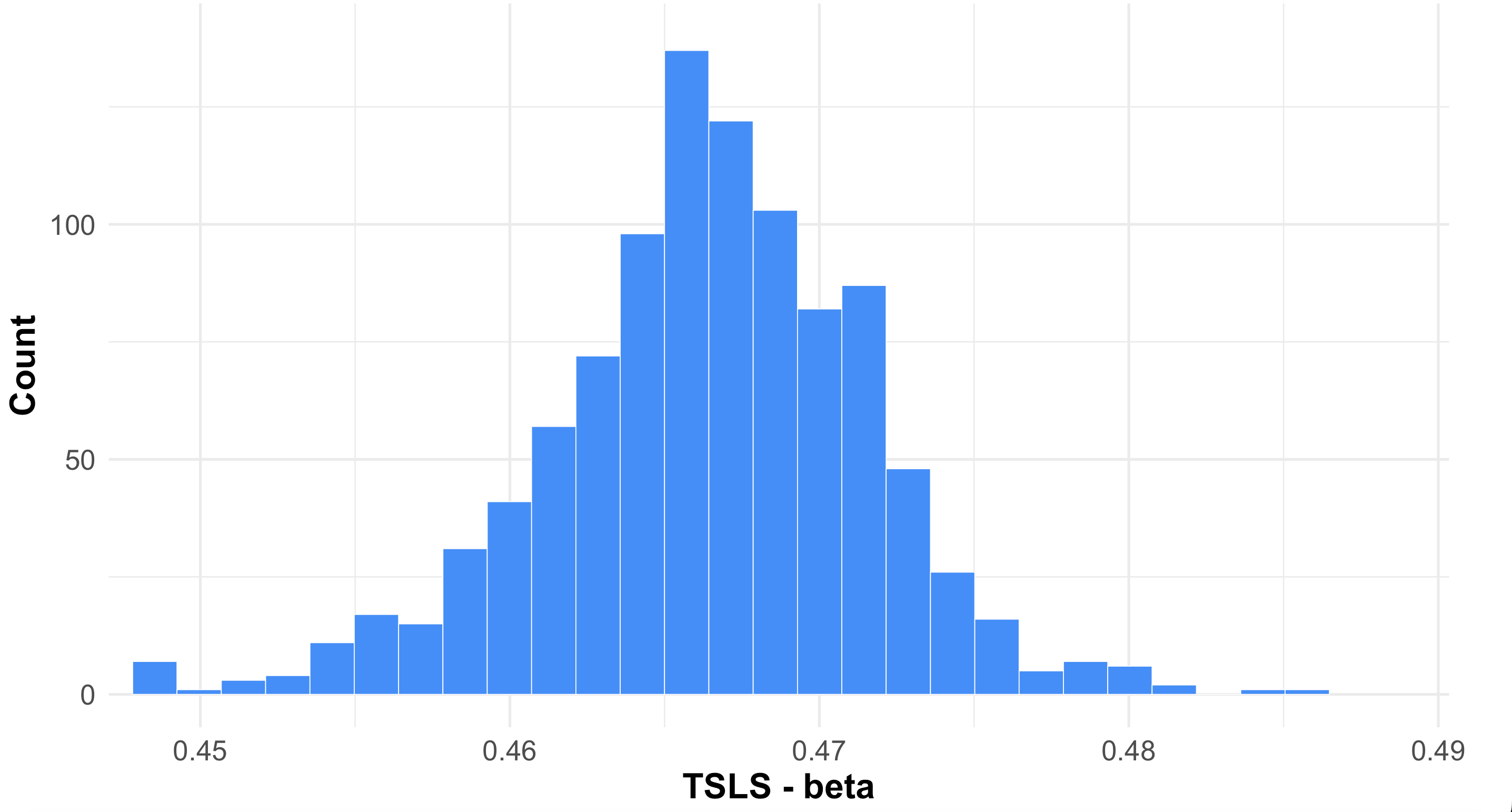}
    \caption{Strong First-Stage ($d=0.5$)}
    \label{fig:cov_beta_strong_un}    \end{subfigure}
    \hfill
    \begin{subfigure}[b]{0.49\textwidth}
  \centering
    \includegraphics[width=\linewidth]{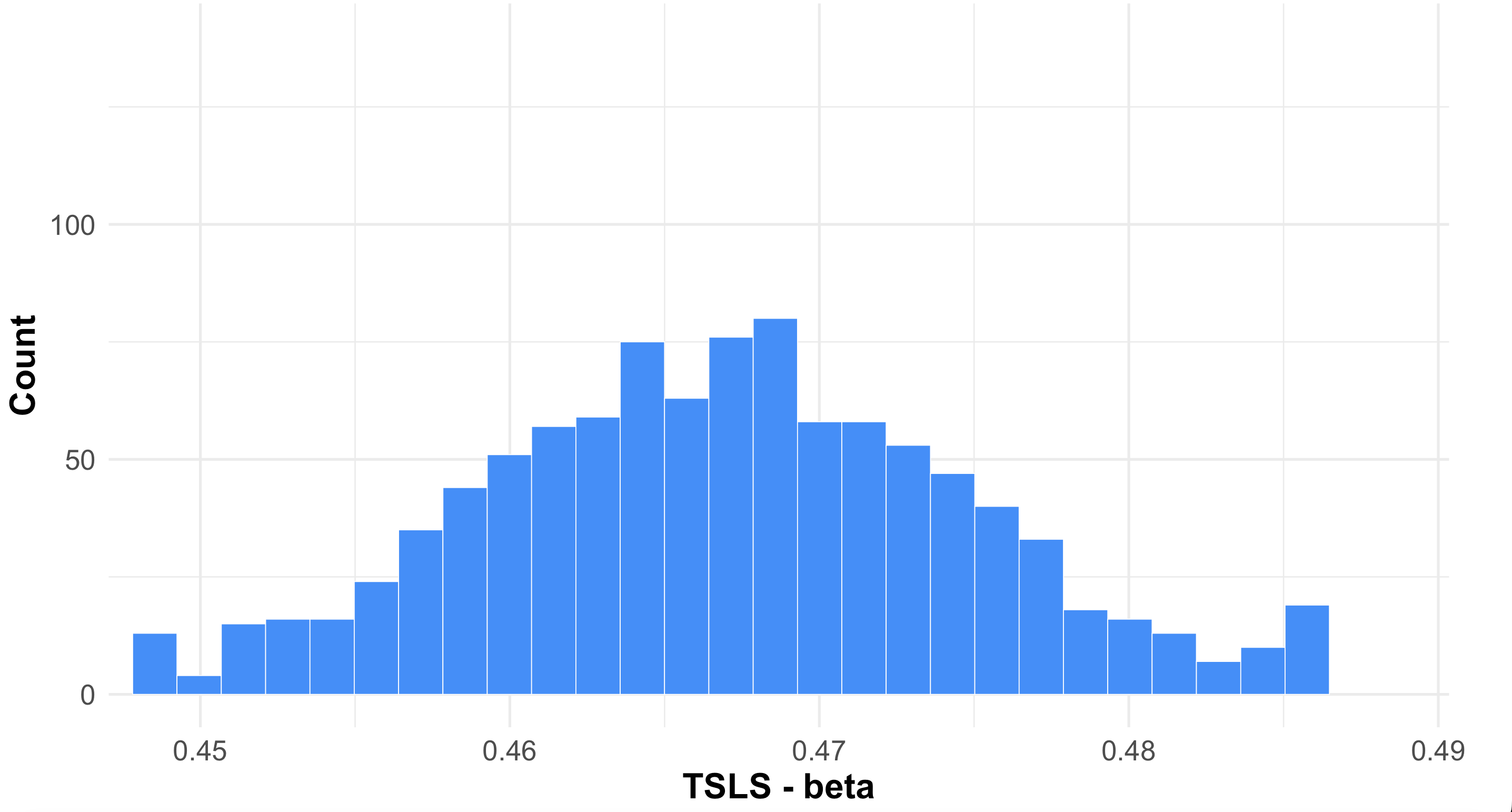}
    \caption{Weak First-Stage ($d=2$)}
    \label{fig:cov_beta_weak_un}
    \end{subfigure}
    \caption{Monte Carlo Simulations Section \ref{S4}: Strong versus Weak First-Stage (Unscaled Model): TSLS $\beta$ estimates ($\beta_0 = 0.4666$).}
    \label{fig:cov_beta_un}
\end{figure}

\newpage

\begin{figure}[ht]
    \centering
    \captionsetup[subfigure]{justification=centering}
    \begin{subfigure}[b]{0.49\textwidth}
    \centering
    \includegraphics[width=\linewidth]{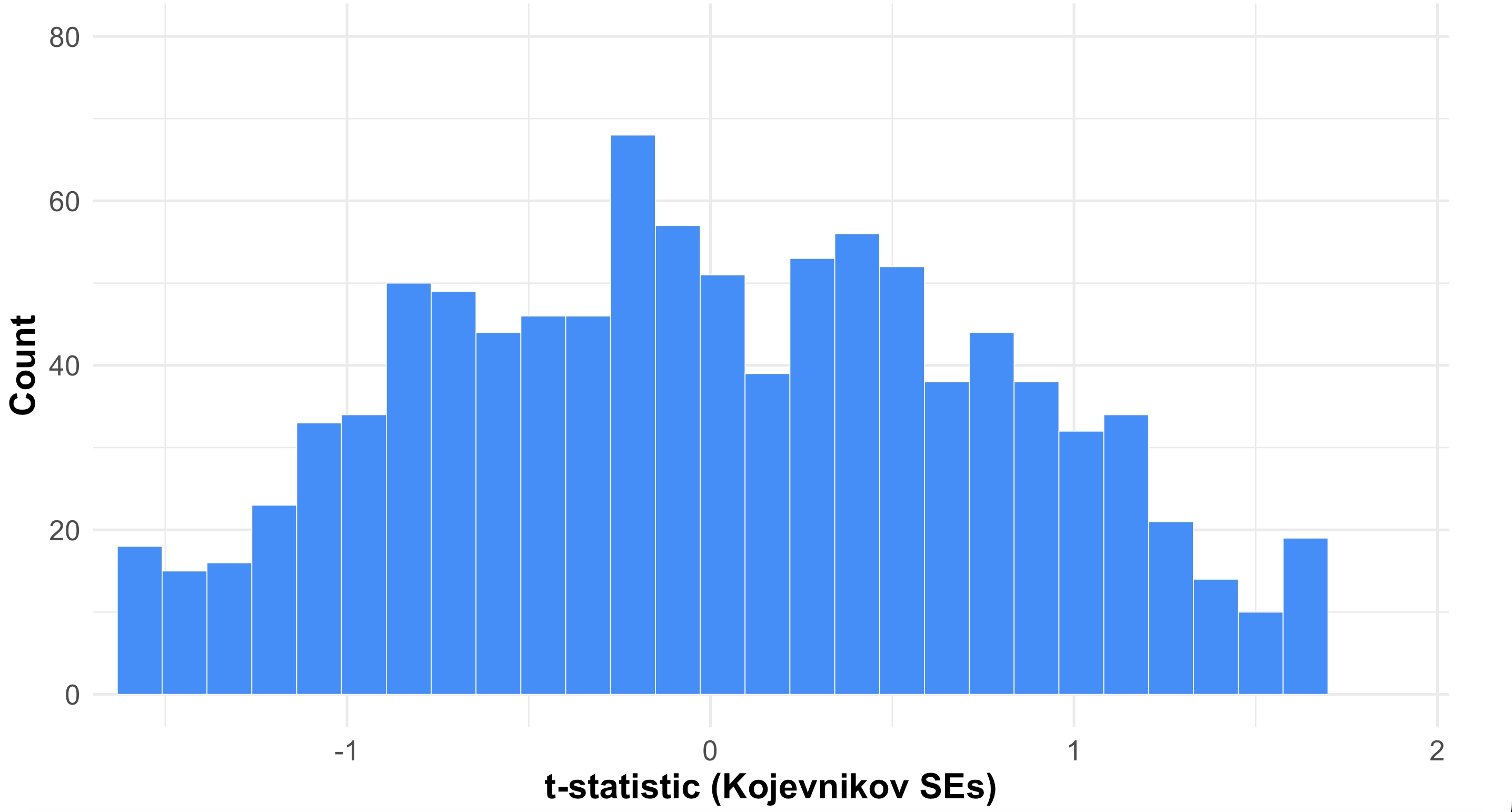}
    \caption{Strong First-Stage ($d=0.5$)}
    \label{fig:cov_t-stat_strong_un}    \end{subfigure}
    \hfill
    \begin{subfigure}[b]{0.49\textwidth}
  \centering
    \includegraphics[width=\linewidth]{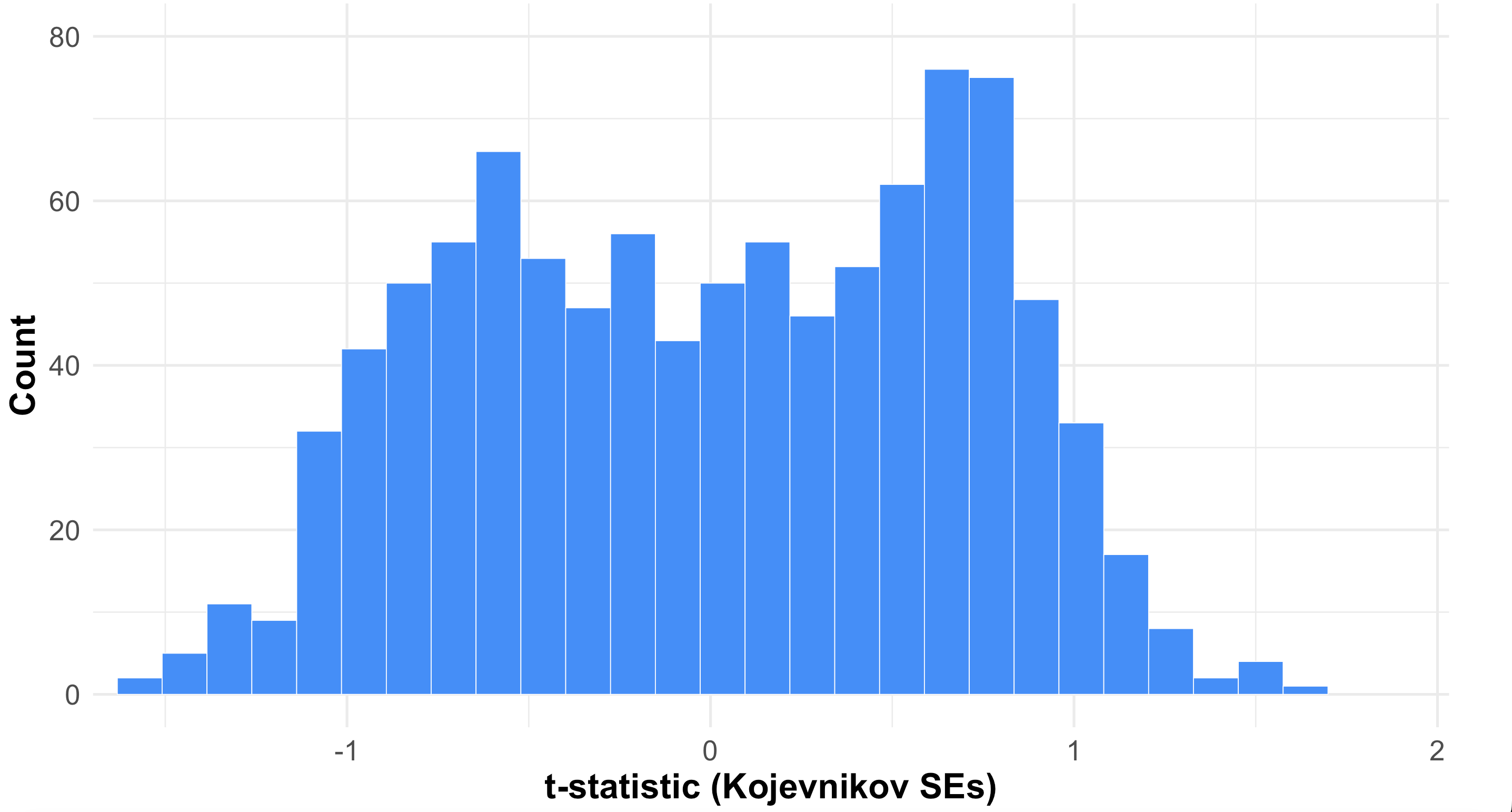}
    \caption{Weak First-Stage ($d=2$)}
    \label{fig:cov_t-stat_weak_un}
    \end{subfigure}
    \caption{Monte Carlo Simulations Section \ref{S4}: Strong versus Weak First-Stage (Unscaled Model): t-Statistic for TSLS $\beta$ estimates with standard errors from \cite{kojevnikov2021limit} and $\beta_0 = 0.4666$.}
    \label{fig:cov_t-stat_un}
\end{figure}

\begin{figure}[ht]
    \centering
    \captionsetup[subfigure]{justification=centering}
    \begin{subfigure}[b]{0.49\textwidth}
    \centering
    \includegraphics[width=\linewidth]{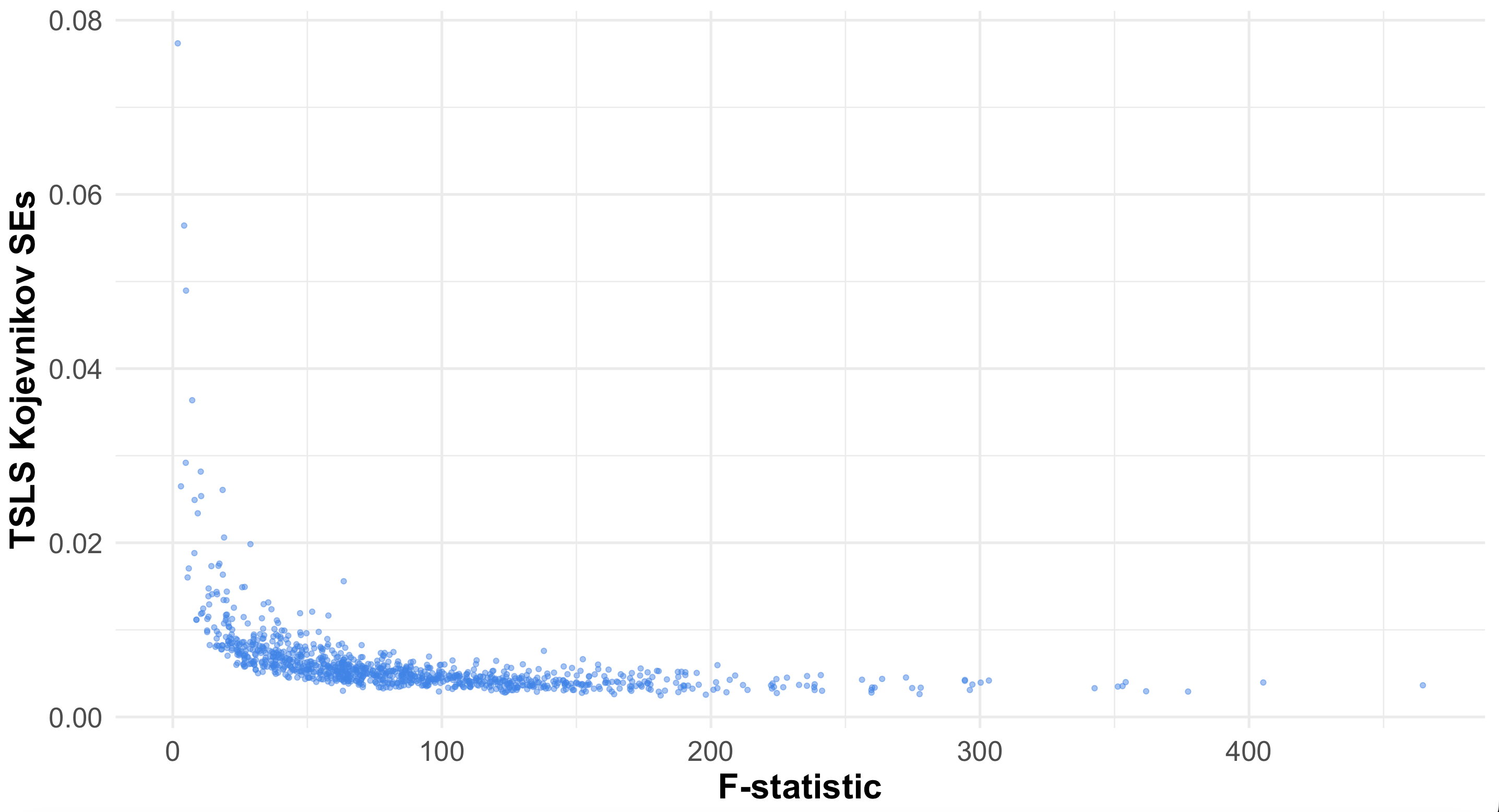}
    \caption{Strong First-Stage ($d=0.5$)}
    \label{fig:cov_Fstat_strong_un}    \end{subfigure}
    \hfill
    \begin{subfigure}[b]{0.49\textwidth}
  \centering
    \includegraphics[width=\linewidth]{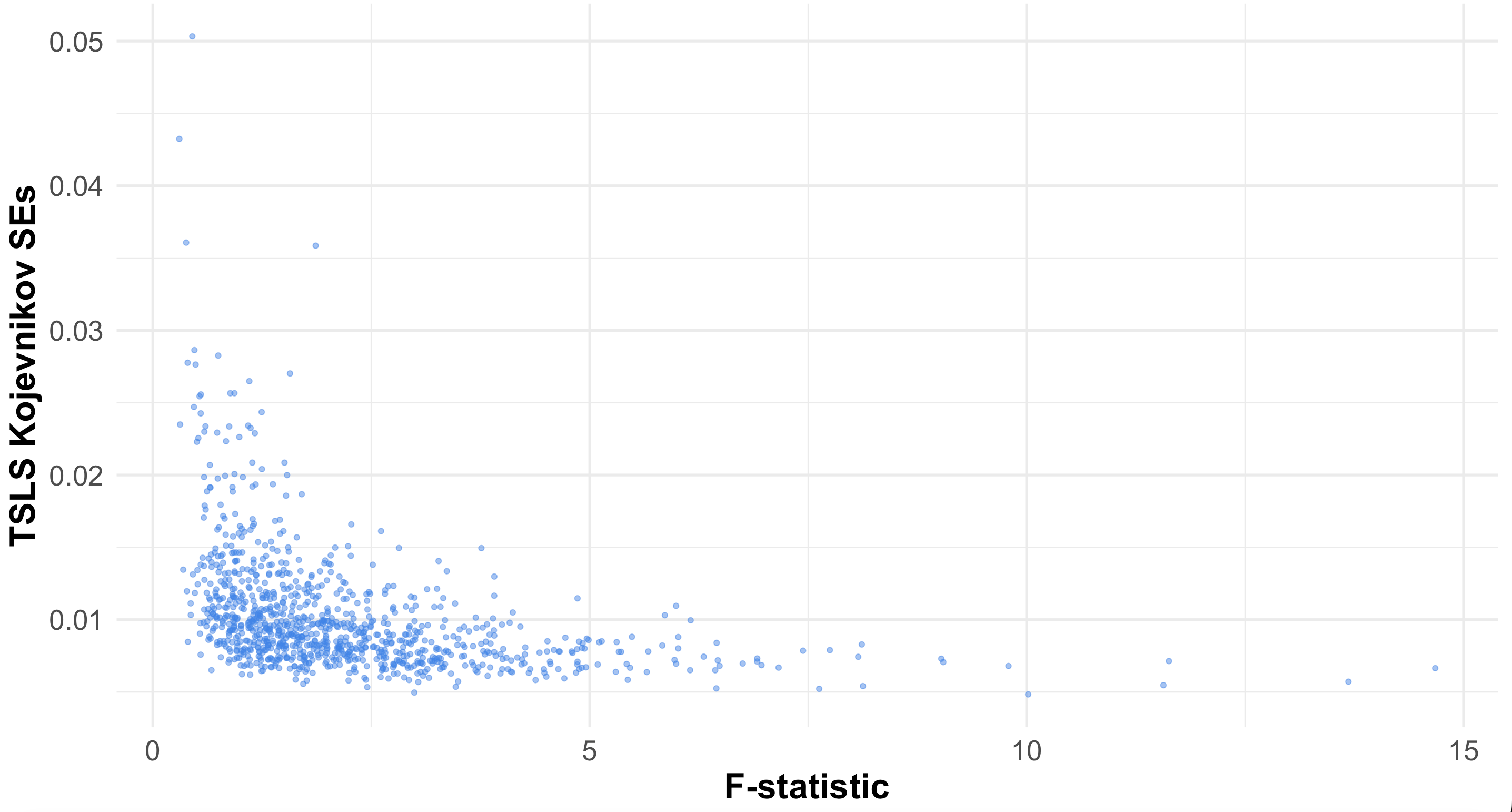}
    \caption{Weak First-Stage ($d=2$)}
    \label{fig:cov_Fstat_weak}
    \end{subfigure}
    \caption{Monte Carlo Simulations Section \ref{S4}: Strong versus Weak First-Stage (Unscaled Model): first-stage F-statistic versus the TSLS Kojevnikov SEs ($\beta_0 = 0.4666$).}
    \label{fig:cov_Fstat_un}
\end{figure}

\end{document}